\documentclass[12pt]{article}

\usepackage{amsfonts,amssymb,amsmath,amsthm,epsfig,euscript,bbm,amsthm}
\usepackage{algorithm}
\usepackage{algpseudocode}
\usepackage{amsthm}
\usepackage[margin=1in]{geometry}
\usepackage{setspace}
  
\usepackage{amsfonts,amssymb,amsmath,amsthm,epsfig,euscript,bbm,amsthm}
\usepackage[usenames,dvipsnames]{xcolor}
\usepackage{algorithm}
\usepackage[dvipsnames]{xcolor}
\usepackage{algpseudocode}
\usepackage{amsthm}
\usepackage{threeparttable,subcaption}
 \usepackage{booktabs}
\newcommand{\ra}[1]{\renewcommand{\arraystretch}{#1}}
 
\usepackage{rotating,multirow,makecell,array}
\usepackage{enumitem}
\usepackage{rotating}
\usepackage{lscape}
\usepackage{float}
\usepackage{graphicx}
\usepackage{subcaption}
\usepackage{tikz}
\usetikzlibrary{arrows,calc}
\usepackage{titletoc}
 
\usepackage{rotating,multirow,makecell,array}

\usepackage{tikz}
\usetikzlibrary{arrows}
\usepackage {graphicx}
\usepackage{enumerate}
\definecolor{darkbrown}{rgb}{0.4, 0.26, 0.13}

\usepackage{subcaption}
\usepackage{caption}
\usepackage[authoryear]{natbib}
\usepackage[titletoc,title]{appendix}
\usepackage[pdftex,colorlinks]{hyperref}
\usepackage{mathtools}
 \usepackage{xcolor}

\usepackage{algorithm,algpseudocode}

\newcounter{algsubstate}
\renewcommand{\thealgsubstate}{\alph{algsubstate}}
\newenvironment{algsubstates}
  {\setcounter{algsubstate}{0}%
   \renewcommand{\State}{%
     \stepcounter{algsubstate}%
     \Statex {\footnotesize\thealgsubstate:}\space}}
  {}

\usepackage{graphicx}
\usepackage{subcaption}
 \usetikzlibrary{graphs, graphs.standard}
 \definecolor{orange}{RGB}{230,170,120}
  \definecolor{green}{RGB}{120,200,120}

 \hypersetup{
 	colorlinks=true,
 	linkcolor=blue,
 	filecolor=blue,      
 	urlcolor=cyan,
 }
 \hypersetup{colorlinks,linkcolor={blue},citecolor={blue}} 
 \usepackage{geometry}                
 \usepackage{multirow}
\usepackage[section]{placeins}
 \usepackage{bbm}
 \usepackage{graphicx}
 \usepackage{amssymb}
 \usepackage{epstopdf}
 \usepackage{tikz} 
 \usepackage{amsmath} 
\usepackage[utf8]{inputenc}
\usepackage{algorithm}
\usetikzlibrary{calc}
\usepackage{setspace}
 \usepackage{epstopdf}
\usetikzlibrary{arrows, chains, calc, positioning, shapes.multipart}
 \theoremstyle{plain}
 \newtheorem{thm}{Theorem}[section]
 \newtheorem{lem}[thm]{Lemma}
 \newtheorem{prop}[thm]{Proposition}
 \newtheorem{cor}{Corollary}
 \usepackage{caption}
 \theoremstyle{definition}
 \newtheorem{defn}{Definition}[section]
 
 \newtheorem{exmp}{Example}[section]
  \newtheorem{ass}{Assumption}[section]
  \usepackage{multibib}
 \theoremstyle{definition}
 \newtheorem{rem}{Remark}
 
 \makeatletter
 \def\BState{\State\hskip-\ALG@thistlm}
 \makeatother
 
 \usepackage{authblk}
\def\spacingset#1{\renewcommand{\baselinestretch}%
{#1}\small\normalsize} \spacingset{1}
 \usepackage{xr}
\definecolor{coquelicot}{rgb}{1.0, 0.22, 0.0}

\algnewcommand\algorithmicforeach{\textbf{for each}}
\algdef{S}[FOR]{ForEach}[1]{\algorithmicforeach\ #1\ \algorithmicdo}
 
\title{ Policy Targeting under Network Interference\footnote{Version accepted for publication at \textit{The Review of Economic Studies}. First version of the paper: June, 2019. I am grateful to Graham Elliott, James Fowler, Paul Niehaus, Yixiao Sun, and Kaspar W\"uthrich for advice and support, and 
Isaiah Andrews, Brendan Beare, Jelena Bradic, Guido Imbens, Toru Kitagawa, Michal Kolesar, Craig Mcintosh, Karthik Muralidharan, James Rauch, Fredrick Savje, Jesse Shapiro, Elie Tamer, Alex Tetenov, Ye Wang, the editor and anonymous referees for comments and discussion. I particularly thank Vikram Jambulapati for invaluable discussions at the beginning of this project. I also thank participants at numerous seminars and conferences. Jake Carlson provided excellent research assistance. The method is implemeted in the R package {\tt NetworkTargeting} available on the author's website. All mistakes are my own.  
        }     
     }
\author{
  Davide Viviano \footnote{Department of Economics, Harvard University. Correspondence: dviviano@fas.harvard.edu.}
}
\date{ This Version: \today
}

\begin{document}
\maketitle

\begin{abstract}
\noindent 
This paper studies the problem of optimally allocating treatments in the presence of spillover effects, using information from a (quasi-)experiment. I introduce a method that maximizes the sample analog of average social welfare when spillovers occur. 
I construct semi-parametric welfare estimators with known and unknown propensity scores and cast the optimization problem into a mixed-integer linear program, which can be solved using off-the-shelf algorithms.  
I derive a strong set of guarantees on regret, i.e., the difference between the maximum attainable welfare and the welfare evaluated at the estimated policy. The proposed method presents attractive features for applications:  (i) it does not require network information of the target population; (ii) it exploits heterogeneity in treatment effects for targeting individuals; (iii) it does not rely on the correct specification of a particular structural model; and (iv) it accommodates constraints on the policy function. 
An application for targeting information on social networks illustrates the advantages of the method. 

\end{abstract}

\noindent%
\vfill

\section{Introduction} 
\onehalfspacing

Consider a policymaker who must use a quasi-experiment, such as an existing experiment or observational study, to design a decision rule (policy) that assigns treatments based on observable characteristics. The main challenge is treating an individual may generate spillovers on her friends or neighbors. Spillovers may, in turn, affect the design of the optimal policy. This paper studies the problem of allocating treatments in the presence of spillover effects to maximize welfare, using information from a quasi-experiment. 
Applications include cash-transfer programs, education programs, and information campaigns, among others \citep[e.g.,][]{egger2019general, opper2016does, bond201261}.


A (large) population of $n$ individuals is connected in a \textit{single} network. Treatments generate spillovers to neighbors in the network (i.e., network interference). Researchers randomly sample $n_e \ll n$ units in a (quasi)experiment and randomize treatments among sampled individuals and their neighbors (the remaining units are not necessarily in the experiment). They then collect sampled individuals' covariates, treatment assignments, outcomes, neighbors' covariates, and assignments. The population network is not necessarily observed. The goal is to estimate a treatment rule to deploy on the entire population. 
Consider the example of targeting information to increase insurance take-up in a region subject to environmental disasters \citep{cai2015social}. Using variation from experiment participants sampled from a random subset of villages in this region, we estimate whom to target in the entire region.


The first challenge is that the population network may be unobserved due to the cost of collecting network data on large populations. Researchers may only observe neighbors' information about the experiment participants.
Collecting network information from the individuals in the entire population, such as a region or country, is often costly or infeasible \citep[see][for a discussion]{breza2017using}. Motivated by this, I develop a method that does not require we observe the population network. I allow for arbitrary constraints on the policy space, such as informational constraints.  
A second challenge is treatment effects heterogeneity. I leverage the assumption that spillovers occur through the number of treated neighbors, as is often documented in applications, and allow for treatment effects heterogeneity in arbitrary individual characteristics (e.g., covariates and number of neighbors).\footnote{Models consistent with this restriction are models of exogenous and anonymous spillover effects; see, e.g., \cite{manski2013identification}. For instance, \cite{cai2015social} leverage a two-stage experimental design to show ``the network effect is driven by the diffusion of insurance knowledge" (i.e., treatment) ``rather than purchase decisions" (i.e., outcome) \citep[][abstract]{cai2015social}, consistent with the model proposed in this paper.  Other examples of empirical applications using models consistent with our model include \cite{sinclair2012detecting, duflo2011peer, muralidharan2017general}, where for the second reference, networks can be considered groups of classrooms with units within each classroom being fully connected. }

The proposed method, which I call Network Empirical Welfare Maximization (NEWM), estimates the welfare as a function of the policy using arbitrary estimators (e.g., based on machine learning). It then solves an exact optimization procedure over the policy space. I interpret policy targeting as a treatment choice problem \citep{manski2004, KitagawaTetenov_EMCA2018, athey2017efficient}, here studied in the context of network interference.  
I evaluate the method's performance based on its maximum regret, that is, the difference between the largest achievable welfare and the welfare from deploying the estimated policy.

 From a theoretical perspective, this paper makes three contributions: (i) it derives the first set of guarantees on the regret for treatment rules with spillovers; (ii) it introduces an estimation procedure with fast convergence rates of regret with machine-learning (non-parametric) estimators and networked units; and (iii) it shows that for a large class of policy functions, the optimization problem can be written as a mixed-integer linear program, solved using off-the-shelf optimization routines.

The analysis proceeds as follows. First, I discuss the identification of social welfare under interference. Identification relies on the unconfoundedness of treatment assignments and of the sampling indicators. I then study semi-parametric estimators for the welfare and analyze the performance of the estimated policy. I show that under regularity conditions, the regret of the estimated policy scales at the rate $1/\sqrt{n_e}$, whenever the maximum degree (i.e., the number of neighbors) is uniformly bounded \citep[e.g.,][]{de2018identifying}. If the maximum degree grows with the population size, the rate depends on the degree, and converges to zero when the degree grows at an appropriate slower rate than $n$.  
Finally, I derive lower bounds that guarantee a maximin convergence rate of the regret with a bounded degree. Throughout the analysis, I do not impose assumptions on the (joint) distribution of characteristics used for targeting and on the network other than restrictions on the maximum degree.

A condition for these results to hold is that the optimization procedure achieves
the in-sample optimum. I guarantee it by showing that we can cast the problem in a mixed-integer
linear program.

The derivations present several challenges: (i) individuals depend on neighbors' assignments that I control through contraction inequalities; (ii) statistical dependence invalidates standard symmetrization arguments \citep{wainwright2019high}; and (iii) in the presence of observational studies with networks, machine-learning estimators may present non-vanishing bias even when using existing methods \citep[e.g.,][]{athey2017efficient}. For (iii), I introduce a novel cross-fitting algorithm for networked observations and characterize its properties.  



  I study the numerical properties of the method using 
data from \cite{cai2015social}. I design a policy that informs farmers about insurance benefits to increase insurance take-up. The NEWM method leads to (out-of-sample) improvements in insurance take-up up to thirty percentage points compared to methods that ignore network effects \citep{KitagawaTetenov_EMCA2018, athey2017efficient}. I obtain these improvements despite not using network information for the design of the policy. Finally, I present several extensions, including trimming when individuals present poor overlap due to a large maximum degree, different target, and sampled populations, and spillovers over non-compliance (in the Appendix).



This paper builds on the growing literature on statistical treatment choice \citep[][]{KitagawaTetenov_EMCA2018, kitagawa2017equality, athey2017efficient, mbakop2016model, armstrong2015inference, bhattacharya2012inferring, hirano2009asymptotics, stoye2009minimax, stoye2012minimax, tetenov2012statistical, zhou2018offline}, and classification \citep[][among others]{elliott2013predicting, boucheron2005theory}. 
Unlike previous references, I estimate the policy when treatments generate spillovers here. This paper is the first to study the properties of targeting on networks in the context of the empirical welfare maximization literature. 

A conceptual difference from the $i.i.d.$ setting with single and multi-valued treatments as in \cite{KitagawaTetenov_EMCA2018}, \cite{zhou2018offline} is that here individuals depend on neighbors' assignments, whereas treatments are individual-specific. This structure permits the population network to be unobserved. In addition, I can bound the complexity of the function class using properties of the maximum degree. The second difference is that individuals exhibit dependence and arguments based on $i.i.d.$ sampling, such as symmetrization, fail here.
Optimization differs because individuals depend on neighbors' treatments.


This paper connects the literature on treatment choice with the one on targeting and networks.  I provide an overview below and an extensive discussion in Section \ref{sec:connections}.  

The influence-maximization literature mostly focuses on detecting the most influential ``seeds" based on centrality measures. These measures are often motivated by a particular model. See \cite{bloch2017centrality} for a review. Recent advances include \cite{jackson2018behavioral}, \cite{akbarpour2018just}, \cite{banerjee2017using}, \cite{banerjee2014gossip}, \cite{galeotti2017targeting} in economics, and \cite{kempe2003maximizing}, \cite{eckles2019seeding}, among others in computer science. 
This paper differs in (i) its approach because I leverage experimental variation to construct policies that maximize the \textit{empirical} welfare (instead of policies justified by game theoretic structures); (ii) setup because I allow for constraints on the policy class and heterogeneity in treatment effects. These differences leverage the assumption that spillovers propagate locally in the network, which differs from some of the models in the influence maximization literature. 
\cite{su2019modelling} study first-best policies for linear models \textit{without} policy constraints. I do not impose such structural assumptions. The presence of constraints (and infeasibility of the first-best policy) justifies the regret analysis in the current paper. \cite{laber2018optimal} consider a Bayesian model whose estimation relies on Monte Carlo methods and the correct model specification.

This paper also connects to the literature on social interaction \citep{manski2013identification,manresa2013estimating, auerbach2019identification}, and causal inference under interference or dependence \citep{liu2019doubly, li2019randomization, hudgens2008toward, goldsmith2013social, sobel2006randomized,  savje2017average, aronow2017estimating, chiang2019multiway}. The exogenous and anonymous interference condition is closely related to \cite{leung2019treatment}. However, knowledge of treatment effects is insufficient to construct welfare-optimal treatment rules in the presence of either (or both) constraints on the policy functions or treatment effects heterogeneity.
Additional references include 
 \cite{bhattacharya2019demand} and \cite{wager2019experimenting}, who study pricing with social interactions through partial identification and sequential experiments, respectively. Here,  instead, I study empirical welfare maximization for individualized treatment rules. \cite{li2019randomization}, \cite{graham2010measuring}, and \cite{bhattacharya2009inferring} study optimal configurations of individuals into small groups, such as assigning students to classes, which differs from here where policies denote (constrained) treatment assignments. 
 See 
\cite{kline2020econometric} and \cite{graham2020econometric} for further references.

Finally, more recent works that study targeting in new directions include \cite{kitagawa2020should} in the context of a parametric model of disease diffusion, \cite{ananth} in settings with an observed network of the target population, and \cite{viviano2020policy} in the context of experimental design and sequential experiments.

The paper is organized as follows. Section \ref{sec:identification} presents the problem setup and main conditions. Estimation and theoretical analysis are contained in Section \ref{sec:estimation}. Section \ref{sec:extensions} and online Appendix \ref{app:ext2} present extensions. Section \ref{sec:application} contains an application. Section \ref{sec:conclusion} concludes. Appendix \ref{app:algorithm} (at the end of the main text) presents a practical guide to implement the algorithm, online Appendix \ref{app:numerics} a numerical study and online Appendix \ref{app:derivations_all} theoretical derivations.

\vspace{-2mm}

\section{Problem description} \label{sec:identification}


 
In this section, I introduce the notation and problem setup. I first introduce the outcome model in Section \ref{sec:outcome}. Section \ref{sec:experiment} formalizes the sampling and design in the experiment. The policy targeting exercise is discussed in Section \ref{sec:policy}, and restrictions on the network in Section \ref{sec:network}.  Algorithm \ref{alg:all} in Appendix \ref{app:algorithm} presents a user-friendly description of the procedure.

 \subsection{Outcome model with interference} \label{sec:outcome}
 
Consider a population of $n$ individuals connected under an adjacency matrix $A$. Each individual is associated with an arbitrary vector of characteristics
$
Z_i \in \mathcal{Z} 
$
and a binary indicator $D_i \in \{0,1\}$, with $D_i = 1$, indicating that individual $i$ was assigned the treatment in the experiment, and  $D_i = 0$ if no treatment was assigned. 
Define
$$
A \in \mathcal{A}_n \subseteq \{0,1\}^{n \times n}, \quad N_i = \Big \{j \in \{1, \cdots, n\} \setminus \{i\} : A_{i,j} = 1 \Big\},  \quad Z = (Z_i)_{i=1}^n, \quad D = (D_i)_{i=1}^n, 
$$ where $\mathcal{A}_n$ is the set of symmetric and unweighted adjacency matrices, $N_i$ denotes the friends of $i$, and $|N_i|$ the degree. Let $Y_i$ denote the $i$'s post-treatment outcome in the experiment. Here, $Z$ can be arbitrary and I impose no restriction on its (joint) distribution.

 With interference, unit $i$'s outcome depends on its own and other units' treatment. 
 In full generality, I can write 
 $
Y_i = \tilde{r}_n(i, D, A, Z, \varepsilon_i)
$ for some unobserved random variables $\varepsilon_i$ capturing uncertainty in potential outcomes, and unknown $\tilde{r}_n(\cdot)$.\footnote{We consider $\varepsilon_i$ as a random variable to capture uncertainty in the realization of the outcomes once the policy discussed in Section \ref{sec:policy} is implemented at scale. It is possible to extend our results if we condition on $\varepsilon_i$ as in \cite{leung2022causal} (and therefore without imposing assumptions on $\varepsilon_i$ other than uniformly bounded outcomes as in \cite{leung2022causal}) only in settings where the treatment probabilities are \textit{known} (see Remark \ref{rem:epsilon}).}

\begin{ass}[Interference] \label{ass:sutnva} For $i \in \{1, \cdots, n\}$, let 
\begin{equation} \label{eqn:Y_s}
Y_i = r\Big(D_i, T_i, Z_i, |N_i|, \varepsilon_i\Big), \quad T_i = g_n\Big(\sum_{k \in N_i} D_k, Z_i, |N_i|\Big), 
\end{equation}  
for some function $r(\cdot)$ unknown to the researcher, and function $g_n(\cdot): \mathbb{Z} \times \mathcal{Z} \times \mathbb{Z} \mapsto \mathcal{T}_n \subseteq \mathbb{Z}$, known to the researcher, with $g_n(0, Z_i, |N_i|) = 0$ almost surely, and unobservables $\varepsilon_i$. 
\end{ass} 

Under Assumption \ref{ass:sutnva}, outcomes depend on (i) the number of first-degree neighbors ($|N_i|$), (ii) the number of first-degree treated neighbors (or a function of this, $T_i$), and (iii) individual's treatment status ($D_i$), observables ($Z_i$), and unobservables ($\varepsilon_i$).
Assumption \ref{ass:sutnva} states that interactions are anonymous \citep{manski2013identification}, and spillovers occur within neighbors.  Heterogeneity occurs through the dependence with $Z_i$ and $|N_i|$.  The model relates to  \cite{leung2019treatment}, and \cite{athey2018exact} provide methods to test anonymous and local interference.

Here, $r(\cdot)$ is unknown and $g_n(\cdot)$ is known and characterizes how individuals depend on neighbors' treatments -- that is, the exposure mapping \citep{aronow2017estimating}; $g_n(0, \cdot) = 0$ is without loss of generality, because $r(\cdot)$ also depends on $(Z_i, |N_i|)$.  The function $g_n$ depends on $n$ because its support $\mathcal{T}_n$ can vary with $n$. For example, $g_n$ can be equal to the number of treated neighbors $T_i = \sum_{k \in N_i} D_k$, and the degree can grow with $n$. This scenario is the most agnostic one because $r$ is unknown and therefore equivalent to $g_n(\cdot)$ being unknown. Alternatively, $g_n(\cdot)$ can be equal to a step function of the share of treated neighbors \citep{sinclair2012detecting}. 
The size of $\mathcal{T}_n$ affects treatments' overlap discussed in Assumption \ref{ass:quasi}. 

\begin{ass}[Unobservables $\varepsilon_i$] \label{ass:ignorability}  For all $i \in \{1, \cdots, n\}$, 
\begin{itemize} 
\item[(A)] 
$\varepsilon_i \Big| A, Z \sim \mathcal{U}_{Z_i, |N_i|}
$  
for unknown distributions $\mathcal{U}_{z, l}, z \in \mathcal{Z}, l \in \mathbb{Z}$;
\item[(B)] $\varepsilon_i \perp (\varepsilon_j)_{j \not \in N_i \cup \{N_k, k \in N_i\}} \Big| A, Z$; 
\item[(C)] $
\mathbb{E}\Big[\sup_{d \in \{0,1\}, t \in \mathbb{Z}} |r(d, t, Z_i, |N_i|, \varepsilon_i)|^{3}\Big|A, Z\Big] \le \Gamma^2$, almost surely, for unknown $\Gamma < \infty$.  
\end{itemize} 
\end{ass}

Condition (A) states that unobservables are identically distributed, conditional on the same individual covariates and number of friends, and conditionally independent of $A$ and other units' characteristics. Condition (A) implies network exogeneity, attained if, for example, two individuals form a link based on observable characteristics and exogenous unobservables. 
Condition (A) guarantees that the individual conditional mean function in Equation \eqref{eqn:estimand} below is the same across units.   
Condition (B) states that unobservables are independent across individuals who do not share a common neighbor (see Example \ref{exmp:one_degree}). Condition (C) is a bounded moment assumption. 

Our method can accommodate scenarios where
(A) and (B) fail. I will \textit{not} assume
Condition (A) in settings where the individual treatment probabilities are either known or estimated parametrically (in Lemma \ref{prop:welfare}, and Theorems \ref{thm:thmmain}, \ref{thm:estimatedoutcome}).  I relax (B) in Section \ref{sec:jj}.

\begin{exmp}[Two-degree dependence] \label{exmp:one_degree} Suppose that each individual is associated with $i.i.d.$ unobservables $\eta_i$ and $Y_i = \tilde{r}\Big(D_i, T_i, Z_i, |N_i|, \eta_i,  \sum_{k \in N_i} \eta_k\Big)$ for some unknown function $\tilde{r}(\cdot)$. Then Assumptions \ref{ass:sutnva} and \ref{ass:ignorability} hold with $\varepsilon_i = \Big(\eta_i,  \sum_{k \in N_i} \eta_k\Big)$. 
 \end{exmp}

\subsection{Sampling and experiment} \label{sec:experiment}

Next, I formalize the sampling mechanism and experiment. 

In the spirit of \cite{abadie2020sampling}, I define $R_i \in \{0,1\}$ a random variable indicating whether individual $i$'s post-treatment outcome is observed by the researchers. Researchers do not necessarily observe the adjacency matrix $A$. However, researchers observe $i$'s relevant characteristics and treatment as well as $i$'s neighbors' characteristics and treatments if $R_i = 1$ (i.e., researchers only observe the friends of the sampled individuals but not necessarily $A$). 
In addition, sampled units and their neighbors (but not necessarily the other units in the population) are assigned treatments in the experiment ($D_i = 1$) with positive probability. 

I formalize these conditions below. Define $R_i^f = 1\Big\{\sum_{k \neq i} A_{i,k} R_k > 0\Big\}$ the indicator of whether individual $i$ has at least one neighbor who is sampled, and $n_e = \sum_{i=1}^n \mathbb{E}[R_i]$ the expected number of sampled individuals. I consider $n_e < n$, and assume that $n_e$ is proportional to $n$ for expositional convenience.\footnote{If $n_e = n^\rho, \rho < 1$ all our results hold if we replace the right-hand side in Assumption \ref{ass:a2} with $\mathcal{O}(n^{(1/2 - \xi)\rho})$.}

\begin{ass}[(Quasi)experiment] \label{ass:quasi} For $i \in \{1, \cdots, n\}$, the following holds:
\begin{itemize} 
\item[(i)] Researchers observe the vector 
\begin{equation} \label{eqn:first}
\Big[R_i \Big(Y_i, Z_i, D_i, N_i, Z_{k \in N_i}, D_{k \in N_i}\Big), R_i\Big]_{i=1}^{n}, \quad R_i \Big| A, Z, (\varepsilon_j)_{j=1}^n \sim_{i.i.d.} \text{Bern}(n_e/n),
\end{equation} 
with $n_e/n = \alpha \in (0,1)$. 
\item[(ii)] $
D_i = f_D\Big(Z_i, R_i, (1 - R_i) R_i^f, \varepsilon_{D_i}\Big)
,\text{ for } \varepsilon_{D_i} | A, Z, (\varepsilon_j)_{j=1}^n, (R_j)_{j=1}^n \sim_{i.i.d.} \mathcal{L}, 
$
for some $f_D(\cdot)$ and distribution $\mathcal{L}$ (known in an experiment and to be estimated in a quasi-experiment); 
\item[(iii)] $P(D_i = 1 | Z_i, R_i = 1), P(D_i = 1 | Z_i, R_i = 0, R_i^f = 1) \in (\gamma, 1 - \gamma)$ almost surely, for some $\gamma \in (0,1)$, and for all $t \in \mathcal{T}_n$, $P\Big(T_i = t | Z_{k \in N_i}, |N_i|, R_{k \in N_i}, R_i = 1\Big) \ge \delta_n$ almost surely, for some $\delta_n \in (0,1)$; 
\end{itemize} 
\end{ass}

Condition (i) states that researchers observe the post-treatment outcomes of sampled units, the covariates and treatment of sampled units, and the covariates and treatments of the friends of the sampled units. I do not assume that $A$ (the connections of the entire target population) is observed, while I assume that relevant information about the friends of the sampled individuals ($R_i = 1$) is observed. 
Condition (i) also postulates 
that the indicators $R_i$ are exogenous with respect to the network $A$, characteristics $Z$ and unobservables $\varepsilon_i$. 

Finally, Condition (i) states that the expected number of sampled individuals $n_e$ is proportional to $n$, which is assumed for expositional convenience. 
We can allow $R_i$ to depend on $Z_i$ (see Remark \ref{rem:ww2}) and $n_e$ not to be proportional to $n$.

 Condition (ii) states the treatment is randomized in the experiment on observables $Z_i$, which can be arbitrary and may also contain network information, and possibly also on the indicator $R_i$. If individuals are not sampled in the experiment ($R_i = 0$), $D_i$ can also depend on whether \textit{at least one} friend is sampled (e.g., researchers collect neighbors' information and then randomize treatments across participants and their neighbors). 

Condition (iii) imposes positive overlap for sampled units and their friends but not necessarily for the remaining units who are not sampled and are not friends of sampled units.
For example, the treatment of those units who do not participate in the experiment and whose friends do not participate in the experiment \textit{can} be equal to the baseline value $D_i = 0$ almost surely, whereas it is randomized with positive probability for the experiment participants and their friends. Here, $\delta_n$ denotes the overlap constant of the neighbors' treatments of the sampled individuals. It depends on $n$, because the support of the exposure mapping $T_i$ may vary with $n$. 
We defer to Section \ref{sec:network} restrictions on $\delta_n$ and on the network.

Figure \ref{fig:network} (left-hand-side panel) presents an illustration. 
In an experiment, Assumption \ref{ass:quasi} entails: randomizing participants $R_i$; collecting the covariates $Z_i$ and their neighbors' covariates $Z_{N_i}$; randomizing treatments among participants and their friends (observed by the researchers); observing the post-treatment outcomes $Y_i$ of the sampled units ($R_i = 1$).

Under Assumptions \ref{ass:ignorability}, and  \ref{ass:quasi} define 
\begin{equation} \label{eqn:estimand}
\small 
\begin{aligned} 
m(d, t, z, l) & = \mathbb{E}\Big[r(d, t, z , l, \varepsilon_i) \Big| Z_i = z, |N_i| = l, T_i = t, D_i = d \Big] \\ 
e(d, t, \mathbf{x}, \mathbf{u},  z, l) & =  P\Big(D_i = d, T_i = t \Big|  Z_{k \in N_i} = \mathbf{x}, R_{k \in N_i} = \mathbf{u}, Z_i = z, R_i = 1,  |N_i| = l \Big)
\end{aligned} 
\end{equation}
the conditional mean and propensity score for sampled units ($R_i = 1$), respectively, where we suppressed the dependence of $e$ with $n$ for expositional convenience.  Note that Assumption \ref{ass:ignorability} (A) guarantees that $m(\cdot)$ does not depend on the index $i$. When the propensity score is known, Assumption \ref{ass:ignorability} (A) is not necessary for our results to hold,  because we can use information about $e(\cdot)$ for identification and estimation.



 \begin{figure}
 \centering
    \begin{tikzpicture}

\coordinate (1) at (-10,1);
\coordinate (2) at (-2,1);
\coordinate (3) at (-2,5);
\coordinate (4) at (-10,5);


  \node[draw, fill = pink, ultra thick, circle] (aaa) at (0, 3.2) {};
  \node[draw, fill = pink, circle] (bbb) at (0, 4.1) {};
  \node[draw, fill = green, circle] (ccc) at (-1, 3.5) {};
 \node[draw, fill =pink, circle] (ddd) at (-0.8, 2.5) {};
  \node[draw, fill =pink,  circle] (eee) at (0.7, 2.5) {};
  \node[draw, fill = pink, circle] (fff) at (1, 3.5) {};
 \node[circle] (hhh) at (0.9, 4.4) {};
 \node[circle] (iii) at (1.7, 2.9) {};
  \node[circle] (lll) at (1.7, 4) {};
   \node[circle] (mmm) at (-0.9, 4.4) {};
   \node[circle] (nnn) at (-1.6, 3) {};
   \node[circle] (ooo) at (-1.4, 2) {};
   \node[circle] (ppp) at (-1.7, 3.9) {};
     \node[circle] (qqq) at (0, 4.8) {};
       \node[ circle] (rrr) at (0, 1.9) {};
       \node[circle] (sss) at (1.2, 1.9) {};

  \node[draw, fill = pink, circle] (aa) at (6, 3.2) {};
  \node[draw, fill = green, circle] (bb) at (6, 4.1) {};
  \node[draw, fill = green, circle] (cc) at (5, 3.5) {};
 \node[draw, fill = green,  circle] (dd) at (5.2, 2.5) {};
  \node[draw, fill = pink, circle] (ee) at (6.7, 2.5) {};
  \node[draw, fill = pink, circle] (ff) at (7, 3.5) {};
 \node[circle] (hh) at (6.9, 4.4) {};
 \node[circle] (ii) at (7.7, 2.9) {};
  \node[circle] (ll) at (7.7, 4) {};
   \node[circle] (mm) at (5.1, 4.4) {};
   \node[circle] (nn) at (4.4, 3) {};
   \node[circle] (oo) at (4.6, 2) {};
   \node[circle] (pp) at (4.3, 3.9) {};
     \node[circle] (qq) at (6, 4.8) {};
       \node[ circle] (rr) at (6, 1.9) {};
       \node[circle] (ss) at (7.2, 1.9) {};

    \draw[-, dashed] (aa) edge (bb) (aa) edge (ee);

    \draw[-, ultra thick] (aaa) edge (bbb) (aaa) edge (ccc)  (aaa) edge (eee);
 
 \draw[-] (bbb) edge (ccc)  (ddd) edge (eee)  (fff) edge (bbb); 
 \draw[-] (fff) edge (hhh) (fff) edge (iii) (eee) edge (iii) (bbb) edge (hhh) (fff) edge (lll); 
\draw[-] (mmm) edge (ccc) (nnn) edge (ccc) (ppp) edge (ccc) (nnn) edge (ddd) (ooo) edge (ddd) (rrr) edge (ddd) (rrr) edge (eee) (sss) edge (eee) (bbb) edge (qqq) (bbb) edge (mmm); 
   
  \draw[-] (aa) edge (bb) (aa) edge (cc)  (aa) edge (ee);
 
 \draw[-] (bb) edge (cc)  (dd) edge (ee)  (ff) edge (bb); 
 \draw[-] (ff) edge (hh) (ff) edge (ii) (ee) edge (ii) (bb) edge (hh) (ff) edge (ll); 
\draw[-] (mm) edge (cc) (nn) edge (cc) (pp) edge (cc) (nn) edge (dd) (oo) edge (dd) (rr) edge (dd) (rr) edge (ee) (ss) edge (ee) (bb) edge (qq) (bb) edge (mm); 
 
  \node (v) at (6, 5) {$\pi(X_i)$};
  \node (v) at (6, 1.5) {$(X_i)_{i=1}^n \subseteq Z$};
  \node (v) at (0, 5) {$D_i | Z_i, R_i, R_i^f \sim \mathcal{P}(Z_i, R_i, R_i^f)$};
  \node (v) at (0, 1.5) {$\Big[(Y_i, Z_i, Z_{\mathcal{N}_i}, D_i, D_{\mathcal{N}_i})R_i, R_i\Big]_{i=1}^n$};
  


    \end{tikzpicture}
 \caption{Example of the experiment (left-hand-side figure) and policy targeting exercise in Section \ref{sec:policy} (right-hand-side figure). Green dots denote treated units, and pink dots denote untreated ones. In the first step, researchers run (or observe data from) an experiment on a (small) subset of individuals, here the black-tick unit. The treatment of such a unit and her friends is randomized with some positive probability, whereas the treatment of the other units can have arbitrary distributions (e.g., equal to the baseline value $D_i = 0$ almost surely if such units are not in the experiment). Researchers observe the vector of outcome, treatment, neighbors, treatments, and covariates of sampled units ($(Y_i, Z_i, Z_{\mathcal{N}_i}, D_i, D_{\mathcal{N}_i})R_i$), as well as the the identity of whom they sample ($R_i$). Researchers then design a treatment allocation $\pi(X_i)$ for the entire population using information $X_i$, a subset of  $Z_i$. 
  }
\label{fig:network}
\end{figure}

\subsection{Policy targeting} \label{sec:policy}

Once the experiment is concluded, a policymaker will design a treatment mechanism with the goal of maximizing average social welfare in the \textit{entire} population $i \in \{1, \cdots, n\}$, with adjacency matrix and covariates $(A, Z)$ as in Figure \ref{fig:network}. Partition $Z_i = \Big[X_i, \tilde{X}_i\Big]$, for two vectors $(X_i, \tilde{X}_i)$, $X_i \in \mathcal{X} \subseteq \mathcal{Z}$. The \textit{policymaker} observes from the \textit{entire} population 
  $$
  X = (X_i)_{i=1}^n, \quad X_i \in \mathcal{X}, 
  $$
  a subset of individuals' characteristics.
  Here, $X_i$ denotes individual information observed by a policymaker for all $n$ units in the population. 
Information $X_i$ can be \textit{arbitrary}. Examples include census data or network statistics \textit{when} observed by the policymaker.\footnote{Although we write $Z_i, |N_i|$ separately for expositional convenience, $Z_i$ (and $X_i$) can also contain the degree and other network statistics if observed by the researcher, given that we impose no assumption on $Z$.}  
   Researchers observe an arbitrary function $b_n(X_1, \cdots, X_n)$ of $X$. For instance, $b_n(\cdot)$ can be a constant function if $X_i$ for \textit{all} $n$ units is only observed by the policymaker but not by the researchers, as in  \cite{KitagawaTetenov_EMCA2018}, or can denote the empirical distribution of $X$ if also observed by the researchers. 
   Researchers design a policy such that: 
\begin{itemize}
\item[(1)] Individuals may be treated differently, depending on observable characteristics; 
\item[(2)] The assignment mechanism must be easy to implement without requiring knowledge of the population network $A$; 
\item[(3)] The assignment mechanism can be subject to (economic or ethical) constraints. 
\end{itemize} 
I therefore consider an \textit{individualized} treatment assignment 
$
\pi : \mathcal{X} \mapsto \{0,1\}, \pi \in \Pi_n(b_n(X)) \subseteq \Pi, 
$
where $\Pi_n(b_n(X))$ denotes the set of constraints on $\pi$, a subset of a given function class $\Pi$. Here, the constraints may also depend on researchers' arbitrary information $b_n(X)$.\footnote{For example, $\Pi_n$ may require $\pi \in \Pi$, \textit{and} the capacity constraint $\frac{1}{n} \sum_{i=1}^n \pi(X_i) \le K$ for a constant $K$. } 
The policy $\pi \in \Pi_n$ satisfies (1), (2), and (3). The policy can be implemented in an online fashion, and it does not require observing the population network. However, because I impose no restrictions on $X_i$, individual covariates \textit{can} contain network statistics if available. 

Finally, note that the individualized treatment rules differs from global treatment rules that depend on the population adjacency matrix $A$. Global treatment rules are more flexible, but require observing the network data of the entire target population and therefore are applicable in contexts complementary to ours. See Remark \ref{rem:comparison} for a comprehensive discussion.


I define utilitarian welfare as the expected outcome once I assign treatments with policy  $\pi(X_i)$ in the \textit{entire} population of $n$ units. Under Assumption \ref{ass:sutnva}, welfare is defined as
\begin{equation} \label{eqn:welfare}
\small 
\begin{aligned}  
W_{A, Z}(\pi)  = \frac{1}{n} \sum_{i=1}^n \mathbb{E}\left[r\Big(\pi(X_i), T_i(\pi), Z_i, |N_i|, \varepsilon_i\Big)\Big|A, Z\right], \quad T_i(\pi) = g_n\Big(\sum_{k \in N_i} \pi(X_k), Z_i,  |N_i| \Big). 
\end{aligned}   
\end{equation}
 The definition of welfare implies no carryovers occur from the previous experimental intervention once we deploy policy $\pi$ on the population.\footnote{In practice, carryovers do not occur if either the policy $\pi$ is deployed sufficiently far in time from the experimental intervention or if the experiment run by researchers has a neglible effect on the entire population. See \cite{athey2018design} for a discussion on the no carryovers assumption.}
I collect the assumptions below.

\begin{ass}[Observable characteristics and targeting]
\label{ass:finite_vc} The researchers observe \\ $\Big[R_i\Big(Y_i, D_i, Z_i, D_{N_i}, Z_{N_i}\Big), R_i\Big]_{i=1}^n$ from an experiment as in Equation \eqref{eqn:first}, and $b_n(X_1, \cdots, X_n)$ from the entire population for some arbitrary function $b_n(\cdot)$, and arbitrary $X_i \in \mathcal{X} \subseteq \mathcal{Z}$. They then constructs a (data-dependent) policy $\hat{\pi}_n: \mathcal{X} \mapsto \{0,1\}, \hat{\pi}_n \in \Pi_n(b_n(X)) \subseteq \Pi$. The policymaker observe $X = (X_i)_{i=1}^n$ from the population, and deploy $\hat{\pi}_n$ on the entire population $i \in \{1, \cdots, n\}$.  Here, $\Pi$ is a class of pointwise measurable functions with finite VC dimension $\mathrm{VC}(\Pi)$.\footnote{The VC dimension denotes the cardinality of the largest set of points that the function $\pi$ can shatter. The VC dimension is a common measure of complexity \citep{devroye2013probabilistic}.} Each $\pi \in \Pi$, generates welfare $W_{A, Z}(\pi)$ in Equation \eqref{eqn:welfare}. 
\end{ass} 

I refer to $\Pi_n(b_n(X))$ as $\Pi_n$. 
Assumption \ref{ass:finite_vc} formalizes the policy targeting exercise and imposes restrictions on the complexity of the function class $\Pi$ as in previous literature \citep[e.g.,][]{KitagawaTetenov_EMCA2018, zhou2018offline}. 
Ideally, one would like to learn 
\begin{equation} \label{eqn:oracle} 
\pi_n^* \in \mathrm{arg} \max_{\pi \in \Pi_n} W_{A, Z}(\pi).
\end{equation}  
However, $\pi_n^*$ depends on $m(\cdot)$ and $A$, both unobserved. I replace the oracle problem in Equation \eqref{eqn:oracle} with its sample analog, and compare the estimated policy to $\pi_n^*$. I discuss identification below and defer estimation to the following section. Define (with $
T_i(\pi)
$ in \eqref{eqn:welfare})  
\begin{equation} 
\label{eqn:I_iii} I_i(\pi) =  1\Big\{T_i(\pi) = T_i, \pi(X_i) = D_i \Big\}, \quad e_i(\pi) = e\Big(\pi(X_i), T_i(\pi), Z_{k \in N_i}, R_{k \in N_i}, Z_i,  |N_i|\Big).
\end{equation} 


\begin{lem}[Identification] \label{prop:welfare} Let Assumptions \ref{ass:sutnva}, \ref{ass:quasi} hold. For any $\pi \in \Pi_n$ 
\begin{equation} \label{eqn:in_sample} 
\begin{aligned} 
W_{A, Z}(\pi) & = 
\frac{1}{n_e} \sum_{i=1}^n  \mathbb{E}\left[R_i Y_i \frac{I_i(\pi)}{e_i(\pi)} \Big| A, Z\right]. 
\end{aligned} 
\end{equation} 
\end{lem}

\begin{proof}[Proof of Lemma \ref{prop:welfare}] The proof is in Appendix \ref{app:ident}. 
\end{proof} 

Lemma \ref{prop:welfare} shows that we can identify welfare using information from the propensity score under exogeneity of $R_i$. It does not impose conditions on $(A, Z)$ or $\varepsilon_i$ (Assumption \ref{ass:ignorability} is not required), other than independence with $(R_i, D_i)$ (Assumption \ref{ass:quasi}). 

Lemma \ref{prop:welfare} identifies welfare effects on the \textit{entire} population of $n$ individuals, conditional on $A$ (and therefore also unconditional on $A$), \textit{without} requiring observing $A$. 
The key intuition is to leverage the randomization induced by the sampling indicators $R_i$ and use their independence with the adjacency matrix $A$ and unobservables $\varepsilon_i$. Incorporating sampling uncertainty for policy targeting (without imposing assumptions on the observables and unobservables) is a contribution of independent interest in the context of policy targeting. 

\begin{rem}[Identification of the propensity score] \label{rem:prop} Here, $e(\cdot)$  can be identified because
\begin{equation} \label{eqn:propest}
\small 
\begin{aligned} 
&P\Big(D_i = d, \sum_{k \in N_i} D_k = t \Big|  Z_{k \in N_i} = \mathbf{x}, R_{k \in N_i} = \mathbf{u}, Z_i = z, R_i = 1, |N_i| = l \Big) \\ &=  
P\Big(D_i = d | Z_i = z, R_i = 1\Big) \sum_{w_1, \cdots, w_{l}: \sum_v w_v = t} \prod_{k = 1}^{l}  P\Big(D_{N_i^{(k)}} = w_k \Big| Z_{N_i^{(k)}} = \mathbf{x}^{(k)}, R_{N_i^{(k)}} = \mathbf{u}^{(k)}, R_i = 1\Big).
\end{aligned} 
\end{equation} 
for  $d \in \{0,1\}, s \in \mathbb{Z}, t \le l$, where $\mathbf{x}^{(k)}$ indicates the $k^{th}$ entry of $\mathbf{x}$, and similarly for $\mathbf{u}^{(k)}$. The expression only depends on marginal treatment probabilities, identified from the experiment.  $e(\cdot)$ can then be written as a sum of probabilities in Equation \eqref{eqn:propest}, for any $g_n(\cdot)$ in Assumption \ref{ass:sutnva}. Also, if the treatments of the participants' neighbors is assigned differently than treatment to participants,  $P(D_i = 1 | Z_i, R_i = 0, R_i^f = 1)$ is identified from the neighbors' assignments. \qed 
\end{rem}

\begin{rem}[Non-reversible treatments] \label{rem:non_reversible}  The policy function class $\Pi_n$ does not depends on the treatments randomized in the experiment. Assumption \ref{ass:finite_vc} rules out policies that force policy-makers \textit{not} to change the treatment status of those units treated in the experiment. Appendix \ref{sec:constrain_pi} extends our results to non-reversible policies, i.e., of the form $\pi(X_i)(1 - D_i) + D_i, \pi \in \Pi_n$ (treatment is one if $D_i = 1$ and is $\pi(X_i)$ otherwise), where the policymaker cannot change the treatment status of individuals treated in the experiment. Our theoretical guarantees (and estimation strategies) also apply to non-reversible treatments.  \qed 
  \label{rem:ww} 
\end{rem}

 \begin{rem}[Different populations] \label{rem:ww2} An interesting scenario is when individuals treated by the policymakers are drawn from a population \textit{different} from the one eligible for the experiment (e.g., we sample individuals from a country to implement the policy in a \textit{different} country). We study this scenario in Section \ref{app:average_bound} and Appendix \ref{sec:different_sample}. 
 \qed 
 \end{rem}


\begin{rem}[Comparison with global treatment rules] \label{rem:comparison} Whenever the network from the entire population $A$ \textit{is observed}, policymakers may consider a global policy $\tilde{\pi}_i(X_i, A)$ that also depends on $A \in \mathcal{A}_n$. This differs from our case, where network statistics can only be included in $X_i$ when observed \citep[e.g., $X_i$ contains measures of centrality as in][]{bloch2017centrality}, and treatments are assigned with policies $\pi(X_i)$ instead of $\tilde{\pi}_i(X_i, A)$. In either case (global or individualized rules), optimization takes into account spillovers for policy design. 

These two approaches are complementary. Individualized assignments considered here do not require collecting network data from the entire population and accommodate settings where the target population is large (and larger than the sample size). However, estimation of individualized rules only use (local) network information available from the experiment.  

Global assignments can be more flexible: a global assignment rule uses information from the \textit{target} population adjacency matrix $A$ to optimize over a large policy space. However, global assignments require observing the population adjacency matrix $A$ and they require that the size of the target population is small (finite) to control the complexity of the policy function class.\footnote{For instance, for a global function class obtained via unions and the intersection of $k_{n}$ half-planes, the VC dimension of the function class is of order $k_{n} \log(k_{n})$ \citep{csikos2019tight}. For a global policy, $k_{n}$ can grow with $n$ requiring a finite target population. 
In the \textit{absence} of policy constraints, an alternative approach is to impose modeling assumptions as in \cite{kitagawa2020should}, different from here, where we allow for policy constraints and semi-parametric identification. } 
These distinctions highlight the complementarity of the two approaches. Global policy rules are best suited in settings where the adjacency matrix $A$ is observed, and the target population is constituted by networks of small (finite) size, as discussed in \cite{ananth}. Individualized rules instead are best suited in settings where network data can be difficult to collect from a (large) target population. 
 \qed 
\end{rem} 

\begin{rem}[Additional extensions] \label{rem:ident} Extending our framework to settings where $R_i$ depends on $Z_i$ is possible. Identification follows similarly, after dividing each summand in Lemma \ref{prop:welfare} by $P(R_i = 1|Z_i)$, assuming $P(R_i = 1 | Z_i) = \alpha(Z_i)n_e/n$, for $\alpha(z) \in (0,1)$. A different extension is when spillovers over compliance occur. This is discussed in Appendix \ref{sec:non_comp}. Finally, a third extension is when higher-order interference occurs. This follows similarly to what is discussed here once we control for (and observe) higher-order neighbors.  \qed 
\end{rem}

\subsection{Network topology and overlap} \label{sec:network}

I conclude the description of the setup with a set of assumptions on the network topology and overlap that control the degree of dependence. Define $\mathcal{N}_n = \max_{i \in \{1, \cdots, n\}} |N_i| + 2$.

\begin{ass}[Maximum degree] \label{ass:a2} Assume $\mathcal{N}_n^{3/2} \log(\mathcal{N}_n)/\delta_n = \mathcal{O}\Big(n^{1/2 - \xi}\Big)$, almost surely for some (unknown) $\xi \in (0, 1/2]$. 
\end{ass} 

Assumption \ref{ass:a2} bounds the ratio of the maximum degree and the overlap constant and trivially holds in networks with bounded degree described below.

\begin{exmp}[Bounded degree] \label{exmp:bounded} Suppose that $\mathcal{N}_n \le c_0$ almost surely for a constant $c_0$ independent of $n$. Then Assumption \ref{ass:a2} holds with $\xi = 1/2$ almost surely. 
\end{exmp}

Example \ref{exmp:bounded} holds for many economic models, for instance, the ones in \cite{de2018identifying}. Economic applications with a bounded degree include the Add Health Study, and \cite{jackson2012social} among others.\footnote{In the Add Health Study researchers elicited up to five names of friends of each sex. The number of reciprocated friends have median one and less than five percent of individuals have more than three of such links \citep[Footnote 7 in][]{de2018identifying}.  In \cite{jackson2012social} fewer than 1 per 1,000 respondents reached the caps of 5 or 8 nominations (Footnote 37, p. 1879).} Assumption \ref{ass:a2} allows for unbounded degree, in which case properties of the estimators in Section \ref{sec:estimation} will depend on $\mathcal{N}_n$ and $\delta_n$.

\begin{exmp}[Unbounded degree] \label{exmp:unbounded} Suppose $\mathcal{N}_n = \mathcal{O}(n^{1/3})$, and for any $n$, \\ $T_i = 1\Big\{\sum_{k \in N_i} D_k/|N_i| > 1/2\Big\}$, such that $P\Big(T_i = 1 | Z_{k \in N_i}, R_{k \in N_i}, |N_i|, R_i = 1\Big) \in (\iota, 1- \iota)$, for some $\iota \in (0,1)$. Then Assumption \ref{ass:a2} holds for $\xi < 1/2$. 
\end{exmp} 

Restrictions on the degree interact with the choice of the exposure mapping $g_n(\cdot)$ and the overlap constant $\delta_n$. I provide two examples below. 

\begin{exmp}[Overlap as a function of the number of treated units] \label{exmp:exposure} Suppose that for arbitrary $\lambda_n$ 
$$
g_n(t, z, l) = 
\begin{cases} 
t & \text{ if } t < \lambda_n \\ 
\lambda_n & \text{ otherwise}. 
\end{cases} 
$$
This specification states that if individuals have less than $\lambda_n$ treated neighbors, spillover effects exhibit arbitrary heterogeneity in the number of treated friends. Spillovers are constant if the number of treated neighbors exceed a certain threshold $\lambda_n$. In this example, the overlap constant is of order $\min\{\gamma^{\lambda_n}, (1 - \gamma)^{\lambda_n}\}$ with $\gamma$ as defined in Assumption \ref{ass:quasi}.  
\end{exmp}

\begin{exmp}[Improving overlap via model restrictions] \label{rem:j} 
Additional restrictions on $g_n(\cdot)$ (and $T_i$) can improve overlap. Suppose that for some ordered $\tau_1, \tau_2, \tau_3$,
\begin{equation} 
r(d, t, z, l, e) = \begin{cases} &\bar{r}_1(d,z,l,e) \text{ if } t/l \le \tau_1 \\ 
& \bar{r}_2(d,z,l,e) \text{ if } \tau_1 < t/l \le \tau_2 \\ 
&  \bar{r}_3(d,z,l,e) \text{ if } \tau_2 < t/l \le \tau_3 
\end{cases} 
\end{equation} 
for some possibly unknown functions $\bar{r}_1, \bar{r}_2, \bar{r}_3$. In this setting, the exposure mapping is a step-function in the share of treated neighbors with a finite support. 
\qed 
\end{exmp}

In summary, Assumption \ref{ass:a2} requires that the overlap constant $\delta_n \rightarrow 0$ at a slower rate than $1/\sqrt{n}$, that can hold under restrictions of either the exposure mapping or on the degree. Section \ref{sec:trimming} presents theoretical results when Assumption \ref{ass:a2} \textit{fails} -- that is, $\delta_n \rightarrow 0$ at a faster rate in $n$, using a trimming strategy.

\subsection{Spillovers in the related literature} \label{sec:connections}

I pause here to compare our framework and assumptions with existing models of spillovers. 

The framework I present most closely connects to the literature on causal inference under interference, including, among others,  \cite{hudgens2008toward}, \cite{manski1993identification}, \cite{aronow2017estimating} and the model in \cite{leung2019treatment} in particular. 
The model in this paper allows for arbitrary heterogeneity in the number of friends, $|N_i|$, observables $Z_i$, and the exposure mapping $T_i$ as a function of the number of treated friends. We can therefore achieve semi-parametric identification of policy effects in the spirit of the literature on (augmented) inverse probability weights \citep[e.g.,][]{tchetgen2012causal, aronow2017estimating}. 

I do not require restrictions on observables $Z_i$, which can be arbitrarily dependent, and on $A$, other than restrictions on the maximum degree. This approach is possible once I explicitly incorporate sampling uncertainty as in \cite{abadie2020sampling} for policy learning. Similar restrictions on the degree are often imposed to obtain concentration of the estimated causal effects \citep[e.g.,][]{savje2017average}. 
Here, the maximum degree restrictions together with the local interference assumption allow me also to control the complexity of the policy function class, characterized by the direct and spillover effects $\Big(\pi(X_i), \sum_{k \in N_i} \pi(X_k)\Big), \pi \in \Pi$. 







I draw connections to the literature on information diffusion and optimal seeding. This literature mostly studies models where informed individuals transmit information to neighbors sequentially over multiple periods \citep{banerjee2013diffusion, banerjee2014gossip, akbarpour2018just, kempe2003maximizing}. These references do not take into account heterogeneity as in this paper (e.g., through $Z_i$), and study centrality measures motivated by the diffusion model considered. This paper studies a static model with heterogeneity, with spillovers occurring through the number of treated friends.

In particular, as noted by \cite{banerjee2013diffusion}, models of information diffusion focus on either what \cite{banerjee2013diffusion} defines as ``information effects" (people become aware of certain opportunities or technologies) or ``endorsement effects" (people's behavior may affect others' behavior), but not necessarily both (similar to what \citealt{manski1993identification} defines exogenous and endogenous spillovers). Once we interpret the outcome $Y_i$ as technology adoption, this paper mostly focuses on information effects through the dependence of the outcome on neighbors' treatments (information). It can accommodate endorsement effects in those settings where the function $r(\cdot)$ captures endorsement effects in a reduced form.\footnote{An example is having two periods $t \in \{1, 2\}$, where the treatment consists of providing information at time $t = 1$ to some individuals. At $t = 1$, outcomes only depend on individual treatments $D_i$, whereas at $t = 2$ outcomes depend on the average number of friends who adopted the technology. Let $Y_{i,1} = D_i \tau + \varepsilon_{i,1}$ the outcome at time $t = 1$, and $Y_{i,2} = f(D_i, Y_{i,1}, \sum_{k \in N_i} Y_{k,1}, |N_i|, \varepsilon_{i,2})$, for some function $f(\cdot)$ and i.i.d. $\varepsilon_{i,1}, \varepsilon_{i,2}$.  This model satisfy our assumptions for $Y_{i,2}$, with $\varepsilon_i = (\sum_{k \in N_i} \varepsilon_{k, 1}, \varepsilon_{i,1}, \varepsilon_{i,2})$ in Assumption \ref{ass:ignorability}.} 

Finally, a further distinction from the literature on seeding \citep[][]{kempe2003maximizing, kitagawa2020should, galeotti2017targeting} is that the current paper focuses on constrained policies,  motivated by the cost of collecting network data, instead of first-best (unconstrained) policies which would require information on the population network.


\vspace{-2mm}

 \section{Network Empirical Welfare Maximization} \label{sec:estimation}

 Next, I introduce our procedure and its properties. I estimate a policy with guarantees valid for finite (possibly large) $n$ and characterize convergence rates as $n,n_e \rightarrow \infty$. Convergence rates are with respect to a sequence of data-generating processes indexed by $n$, each with a \textit{single} network $A \in \mathcal{A}_n$, where I explicitly condition on $A \in \mathcal{A}_n, Z \in \mathcal{Z}^n$ unless otherwise specified. Conditional statements that I provide below do not subsume that $(A, Z)$ are observed. Instead, they establish stronger guarantees than unconditional statements by leveraging the independence of the sampling $R_i$ with the network $A$ and the assumption that the sampled units are drawn from the (larger) target population (see Lemma \ref{lem:expected}).


 \subsection{Known propensity score} \label{sec:prop_score}

Suppose first researchers know the propensity score.  Consider the double robust estimator (AIPW):
\begin{equation} \label{eqn:welf}
\small 
\begin{aligned} 
W_n(\pi, m^c, e) = & \frac{1}{n_e} \sum_{i=1}^n  R_i\left\{ \frac{I_i(\pi)}{e_i(\pi)} \Big(Y_i - m_i^c(\pi)\Big) + m_i^c(\pi)\right\}, 
\end{aligned} 
\end{equation}
where 
$
 m_i^c(\pi) = m^c\Big(\pi(X_i), T_i(\pi),  Z_i, |N_i|\Big).
$
The function $m^c$ denotes an \textit{arbitrary} regression adjustment, possibly different from the population conditional mean function. Note that $m^c$ can be arbitrary. Therefore, it does not require that the conditional mean functions are identical across units (Assumption \ref{ass:ignorability} (A)). The estimated welfare inherits double-robust properties in the spirit of \cite{robins1994estimation}, and \cite{tchetgen2012causal}, \cite{aronow2017estimating}, \cite{liu2019doubly} with spillovers. For known propensity scores and any $m^c$, the estimator is unbiased for $W_{A, Z}(\pi)$  (see Appendix \ref{app:ident}).


\begin{ass}[Regression adjustment: oracle setup] \label{ass:general4} For each $d \in \{0,1\}, t \in \mathcal{T}_n$, let $|m^c(d, t,  Z_i, |N_i|)| < \Gamma,  
$ almost surely, for a finite constant $\Gamma < \infty$, and for $z \in \mathcal{Z}, l \in \mathbb{Z}$, $m^c(d,t,z,l) \perp \Big(Y_i, R_i, D_i\Big)_{i=1}^n \Big| A, Z$.  
\end{ass} 

Assumption \ref{ass:general4} states that the regression adjustment is (i) uniformly bounded and (ii) independent of experiment participants. An example is $m_i^c = 0$, or $m_i^c$ estimated on an independent population. 
The use of $m_i^c(\cdot)$ in this section is not necessary for our results to hold. However, even with a known propensity score, using a regression adjustment can improve the stability of the estimator when poor overlap occurs. Sections \ref{sec:alg1} and \ref{sec:jj} provide details where $m_i^c$ is estimated in-sample. With known propensity score and a parametric regression adjustment (ii) is not necessary, as shown in Section \ref{sec:jj}. Let
$$
\hat{\pi}_{m^c, e} \in \mathrm{arg}\max_{\pi \in \Pi_n} W_n(\pi, m^c, e).
$$

 \begin{thm}[Oracle Regret] \label{thm:thmmain} Let Assumptions \ref{ass:sutnva}, \ref{ass:quasi}, \ref{ass:finite_vc},  \ref{ass:general4}, and (B), (C) in \ref{ass:ignorability} hold. For a universal constant $\bar{C}<\infty$, the following holds almost surely:
    $$
\mathbb{E}\Big[\sup_{\pi \in \Pi_n} W_{A, Z}(\pi) - W_{A, Z}(\hat{\pi}_{m^c, e} ) \Big| A, Z\Big] \le \bar{C} \frac{\Gamma \mathcal{N}_n^{3/2}}{\gamma \delta_n} \sqrt{ \frac{\log(\mathcal{N}_n) \mathrm{VC}(\Pi)}{n_e}}.
    $$
\end{thm}

 \begin{proof}[Proof of Theorem \ref{thm:thmmain}]  
 The proof consists of three steps. First, I extend symmetrization arguments -- widely studied for independent observations \citep[e.g.,][]{devroye2013probabilistic} -- for network data. To obtain symmetrization, I group units into groups of conditionally independent observations. Within each group, I provide bounds in terms of the Rademacher complexity of the function class obtained from the composition of direct and spillover effects (see Definition \ref{defn:rademacher}). As a second step, I bound the Rademacher complexity in each group (i) by deriving an extension of \cite{ledoux2011probability}'s contraction inequality (Lemma \ref{lem:ledoux}), using (ii) Dudley's entropy integral bound \citep[][Theorem 5.22] {wainwright2019high}, and (iii) providing an upper bound on the covering number  of the product of the number of treated neighbors and individual treatment (Lemmas \ref{lem:boundnumber}, \ref{lem:finallemma}).\footnote{See \cite{wainwright2019high} for definitions of covering numbers.}  As the last step, I invoke \cite{brooks1941colouring}'s theorem to control the number of groups containing conditionally independent units.
 
Section \ref{sec:proof} presents a proof sketch, and Appendix \ref{app:regret} the complete proof.  
\end{proof}  

 Theorem \ref{thm:thmmain} provides a non-asymptotic upper bound on the regret, and it is the first result of this type under network interference.
 
 The regret bound depends on the network topology through the maximum degree $\mathcal{N}_n$, the overlap constant $\delta_n$, and the  (expected) sample size $n_e$. The degree affects the regret bound through two channels: (i) dependence between outcomes conditional on the network and covariates and (ii) the complexity of the function class obtained by the composition of direct and spillover effects. For (i), I leverage Assumptions \ref{ass:sutnva}, \ref{ass:quasi} (i, ii), and \ref{ass:ignorability} (B), to show each individual observation is dependent with at most $2 \mathcal{N}_n^2$ many other units. For (ii), I leverage instead Assumptions \ref{ass:sutnva} and \ref{ass:finite_vc}, to bound (ii) as a function of the VC dimension of $\Pi$ and $\mathcal{N}_n$. The bound also depends on $\delta_n$, which can vary with $n$. Intuitively, for larger networks (and larger degrees), the probability that individuals exhibit strict overlap may get smaller, depending on the exposure mapping considered. The bound is independent of $\alpha$ in Equation \eqref{eqn:first}. Theorem \ref{thm:thmmain} does not assume Assumption \ref{ass:ignorability} (A).

The bound shrinks to zero as $n_e$ increases, only if the maximum degree and the overlap constant grows at an appropriate slower rate than the sample size. We formalize this below.

  \begin{cor}[Convergence rate with a possibly unbounded degree] \label{cor:1_main} Let the Assumptions in Theorem \ref{thm:thmmain} hold. Suppose in addition that Assumption \ref{ass:a2} holds. Then 
 $$
\mathbb{E}\Big[\sup_{\pi \in \Pi_n} W_{A, Z}(\pi) - W_{A, Z}(\hat{\pi}_{m^c, e})\Big| A, Z \Big] = \mathcal{O}\Big(n_e^{-\xi} \Big)
 $$
almost surely, for $\xi \in (0, 1/2]$ as defined in Assumption \ref{ass:a2}. 
 \end{cor}

The corollary shows that the regret converges to zero at a rate that depends on the convergence rate of the maximum degree and the number of experiment participants.
 For bounded degree,  the regret scales at rate $1/\sqrt{n_e}$.

\begin{cor}[Example \ref{exmp:bounded} cont'd] Let the Assumptions in Theorem \ref{thm:thmmain} hold, and $\mathcal{N}_n < c_0'$ almost surely, for a constant $c_0'$ independent of $n$. Then almost surely, 
$$
\mathbb{E}\Big[\sup_{\pi \in \Pi_n} W_{A, Z}(\pi) - W_{A, Z}(\hat{\pi}_{m^c, e} )\Big| A, Z \Big] = \mathcal{O}\Big(n_e^{-1/2}\Big). 
$$  
\end{cor}

In the following theorem, I provide a lower bound for any data-dependent policy. Consistently with the previous theorems, I provide the lower bound conditional on $(A,Z)$.

 \begin{thm}[Minimax lower bound on the rescaled regret] \label{prop:minimax} Let $\Pi$ be the class of policies $\pi: \mathcal{X} \mapsto \{0,1\}$, with finite VC dimension $\mathrm{VC}(\Pi)$, $\mathcal{X} = \mathbb{R}^d \subseteq \mathcal{Z}$, for some finite $d < \infty$. Let $\mathcal{P}_n(A,Z)$ the set of conditional distributions $\mathcal{D}_n(A,Z)$ of $(Y_i, D_i, R_i)_{i=1}^n | A, Z$ satisfying Assumptions \ref{ass:sutnva}, \ref{ass:ignorability}, \ref{ass:quasi}. Then for any $g_n(\cdot)$ in Assumption \ref{ass:sutnva}, for any $n_e \ge 16 \mathrm{VC}(\Pi)$, and for any data-dependent $\hat{\pi}_n \in \Pi$, which depends on $\Big[R_i (Y_i, Z_i, Z_{k \in N_i}, D_i, D_{k \in N_i}, N_i), R_i\Big]_{i=1}^n$, 
\begin{equation} \label{eqn:rescaled}
 \small 
 \begin{aligned} 
 \sup_{A \in \mathcal{A}_n^o, Z \in \mathcal{Z}^n} \sup_{\mathcal{D}_n(A,Z) \in \mathcal{P}_n(A,Z)} \frac{\delta_n}{\mathcal{N}_n^{3/2} \log^{1/2}(\mathcal{N}_n) } \mathbb{E}_{\mathcal{D}_n(A,Z)}& \Big[ \Big(\sup_{\pi \in \Pi} W_{A, Z}(\pi) - W_{A, Z}(\hat{\pi}_n)\Big)\Big| A, Z \Big] \\ &\ge  \frac{\exp(-2 \sqrt{2})}{2^{5/2} \log^{1/2}(2)} \sqrt{ \frac{\mathrm{VC}(\Pi)}{n_e}}, 
\end{aligned} 
\end{equation} 
 where $\mathcal{A}_n^o \subset \mathcal{A}_n$ denotes the space of symmetric unweighted adjacency matrices satisfying Assumption \ref{ass:a2}, and $\mathbb{E}_{\mathcal{D}_n}[\cdot]$ denotes the expectation with respect to $\mathcal{D}_n$. 
 \end{thm}

 \begin{proof}[Proof of \ref{prop:minimax}]  The proof follows similar steps of \cite{devroye2013probabilistic, KitagawaTetenov_EMCA2018}, once I construct a sufficiently sparse adjacency matrix for the worst-case lower bound, with two distinctions that, to my knowledge, are novel in the literature: I condition on covariates and consider random sampling indicators.  See Appendix \ref{app:regret} for details. 
 \end{proof} 
 
 Theorem \ref{prop:minimax} provides a worst-case lower bound to any data-dependent policy, holding uniformly for any $n_e \ge 16 \mathrm{VC}(\Pi)$. Similar to lower bounds in the literature \citep{KitagawaTetenov_EMCA2018}, the bound is maximin over the data-generating process, including any adjacency matrix $A$ satisfying Assumption \ref{ass:a2}. However, different from \cite{KitagawaTetenov_EMCA2018}, 
Theorem \ref{prop:minimax} establishes the minimax convergence rate of $\hat{\pi}_{m^c, e}$ for the \textit{rescaled} regret 
\begin{equation} \label{eqn:rescaled}
\small 
\begin{aligned} 
\frac{\delta_n}{\mathcal{N}_n^{3/2} \log^{1/2}(\mathcal{N}_n) } \mathbb{E}_{\mathcal{D}_n(A,Z)} \Big[ \Big(\sup_{\pi \in \Pi} W_{A, Z}(\pi) - W_{A, Z}(\hat{\pi}_n)\Big)\Big| A, Z \Big]
\end{aligned}
\end{equation}  
after we divide by the factor $(\mathcal{N}_n^{3/2} \log(\mathcal{N}_n))/\delta_n$ appearing in Theorem \ref{thm:thmmain}. The rescaling factor differs from lower bounds on the (non-rescaled) regret in the literature, and it is motivated by the  dependence of $\mathcal{N}_n$ with the adjacency matrix and $\delta_n$ with the data-generating process. We discuss  implications for the regret \textit{without} rescaling below. 

\begin{cor} \label{cor:minimax} For any data dependent $\hat{\pi}_n \in \Pi$, satisfying the conditions in Theorem \ref{prop:minimax}, 
$$
\small 
\begin{aligned} 
\sup_{A \in \mathcal{A}_n^o, Z \in \mathcal{Z}^n} \sup_{\mathcal{D}_n(A,Z) \in \mathcal{P}_n(A,Z)} \mathbb{E}_{\mathcal{D}_n(A,Z)} \Big[ \Big(\sup_{\pi \in \Pi} W_{A, Z}(\pi) - W_{A, Z}(\hat{\pi}_n)\Big)\Big| A, Z \Big] \ge  \frac{\exp(-2 \sqrt{2})}{2^{5/2}} \sqrt{ \frac{\mathrm{VC}(\Pi)}{n_e}}.
\end{aligned} 
$$ 
\end{cor} 

Corollary \ref{cor:minimax} follows from the fact that $\delta_n/\mathcal{N}_n^{3/2} \log^{1/2}(\mathcal{N}_n) \le 1/\log^{1/2}(2)$. 
It states that the lower bound for the rescaled regret implies a lower bound for the regret. Therefore, Theorem \ref{prop:minimax} establishes a minimax rate of convergence of $\hat{\pi}$ for the regret \textit{without rescaling} under the additional assumption that $\mathcal{N}_n < c_0$ is uniformly bounded for a constant $c_0 < \infty$.



In summary, the bound in Theorem \ref{thm:thmmain} converges to zero as $n, n_e \rightarrow \infty$, in settings with a sufficiently small degree (see Corollary \ref{cor:1_main}). The bound in Theorem \ref{thm:thmmain} does not converge to zero if the degree $\mathcal{N}_n$ grows at an arbitrary rate with $n$. Therefore our bounds are informative (converge to zero), only in settings with a sufficiently sparse graph. These settings include bounded degree as a special case, but also allows for unbounded degree with rate satisfying Assumption \ref{ass:a2}.  
For example, with an exposure mapping such that $\delta_n \in (\delta, 1 - \delta)$ for a constant $\delta$ independent of $n$ (for instance, the exposure mapping is as in Example \ref{exmp:exposure} with $\lambda_n$ independent of $n$), the bound converge to zero only if $\mathcal{N}_n^{3} \log(\mathcal{N}_n)/n \rightarrow 0$. In addition, the bound in Theorem \ref{thm:thmmain} also provides a minimax rate of convergence of the regret (without rescaling) in settings where the degree is uniformly bounded (but not necessarily otherwise).


\begin{rem}[Expected regret] Theorem \ref{thm:thmmain} provides guarantees on the regret conditional on $(A, Z)$, assuming that the experiment participants are drawn from the target population. Section \ref{app:average_bound} shows that such guarantees are sufficient to also bound the regret with respect to the \textit{expected} welfare (expected over the distribution of $(A, Z)$) if the sample units are drawn from the target population. When sampled units are \textit{not} drawn from the (larger) target population, regret bounds depend on additional terms that characterize the ``cost" of drawing a sample from a population different from the target one (see Section \ref{app:average_bound}).
\qed 
\end{rem} 
 
\subsection{Estimated nuisance functions} \label{sec:alg1}

Next, I derive regret guarantees when estimating the conditional mean $m(\cdot)$ and/or propensity score $e(\cdot)$, as defined in Equation \eqref{eqn:estimand} under Assumptions \ref{ass:ignorability}, and \ref{ass:quasi}. Define $\hat{m}$, and $\hat{e}$ the estimated conditional mean and propensity score as in Algorithm \ref{alg:adaptive} (Appendix \ref{app:algorithm}), $W_n(\pi, \hat{m}, \hat{e})$ as the welfare with the estimated nuisance functions as in Equation \eqref{eqn:crossfit}, and 
\begin{equation} 
 \hat{\pi}_{\hat{m}, \hat{e}} \in \mathrm{arg}\max_{\pi \in \Pi_n} W_n(\pi, \hat{m}, \hat{e}). 
 \end{equation}

I propose a modification of the \textit{cross-fitting} algorithm -- see \cite{chernozhukov2018double}, and \cite{athey2017efficient} in particular -- here studied in the context of interference. I describe the algorithm in Algorithm \ref{alg:adaptive} and provide a sketch in Algorithm \ref{alg:sketch}.  

First, I find the smallest partition of sampled individuals such that two individuals assigned to the same group are neither friends nor share a common friend. This information is available under the sampling mechanism in Section \ref{sec:experiment}, because researchers observe the set of friends of each sampled individual. 
The solution to this problem is obtained by solving a sequence of mixed-integer linear programs. Each program fixes the number of groups (starting from one). For a given number of groups, it checks whether a feasible partition 
exists. If no feasible partition exists, it increases by one the number of groups and iterates. 

Once I obtain such groups, I estimate the conditional mean function using standard cross-fitting within each group of individuals as in \cite{athey2017efficient}. Specifically, I partition each group $g$ into $K$ equally sized folds; for individual $i$ in group $g$, fold $k$, I estimate her conditional mean function using information from all units in each fold in group $g$ except fold $k$. I repeat the same algorithm for the propensity score, where I first estimate the individual treatment probability and then aggregate such probabilities as in Remark \ref{rem:prop}. 
Algorithm \ref{alg:adaptive} presents the details and Algorithm \ref{alg:sketch} a summary.

As in \cite{athey2017efficient}, the regret bound is increasing in the number of folds, while the estimation error of the nuisance functions is decreasing in the number of folds (see Appendix \ref{app:estimated_1}). Therefore, we must choose  a sufficiently large $K$ to control the estimation error of the nuisance functions. However, the choice of $K$ must also guarantee that each fold contains a non-negligible proportion of observations. In practice, I recommend $K$ between five and ten.

  \begin{algorithm} [!h]   \caption{Sketch of Network Cross-Fitting (see Algorithm \ref{alg:adaptive} for details)}\label{alg:sketch}
    \begin{algorithmic}[1]
    \State Partition sampled individuals: 
    \begin{algsubstates}
     \State Fix $K = 1$
     \State Check whether a feasible partition of sampled individuals with $K$ groups exists. The partition must be such that two individuals in the same group are neither friends nor share a common friend. 
     \State If such a partition does not exist, set $K = K + 1$ and iterate. 
     \end{algsubstates}
    \State For each $i$, estimate the conditional mean function and propensity score for individual $i$, $\hat{m}^{(i)}, \hat{e}^{(i)}$ via cross-fitting using the units in $i$'s group returned by the partition in 1.
    Define 
  \begin{equation} \label{eqn:m_hati}
        \small 
        \begin{aligned} 
        \hat{m}_i(\pi) = \hat{m}^{(i)}\Big(\pi(X_i), T_i(\pi), Z_i, |N_i|\Big), \quad \hat{e}_i(\pi) = \hat{e}^{(i)}\Big(\pi(X_i), T_i(\pi), Z_{k \in N_i}, R_{k \in N_i}, Z_i, R_i, |N_i|\Big)
        \end{aligned} 
        \end{equation} 
     and 
\begin{equation} \label{eqn:crossfit}
W_n(\pi, \hat{m}, \hat{e}) = \frac{1}{n_e} \sum_{i=1}^n R_i \left\{\frac{I_i(\pi)}{\hat{e}_i(\pi)}\Big(Y_i - \hat{m}_i(\pi)\Big) - \hat{m}_i(\pi) \right\}.  
\end{equation}  
\Return $W_n(\pi, \hat{m}, \hat{e})$. 
         \end{algorithmic}
\end{algorithm}

 To my knowledge, Algorithm \ref{alg:adaptive} is novel to the literature on interference. Its main innovation with respect to existing cross-fitting methods is the partitioning approach (Part 1 in Algorithm \ref{alg:sketch}), here required due to interference. For settings where the network presents approximately independent components (e.g., regions), I also present a computational relaxation in Algorithm \ref{alg:adaptive2}. Algorithm \ref{alg:adaptive2} constructs subgraphs of the network \textit{recursively} to minimize the number of individuals with shared friends between different subgraphs. It estimates nuisance functions for unit $i$ using information from units in the subgraphs different from the one of unit $i$. With multiple disconnected regions, Algorithm \ref{alg:adaptive2} estimates the nuisance functions using information from all regions except the one containing $i$. See Appendix \ref{app:algorithm} for details.

To study properties of the algorithm, I assume that the estimated nuisance functions satisfy the same bounded and overlap conditions as their population counterparts \citep[this can be relaxed by assuming uniform convergence as in][]{athey2017efficient}.

\begin{ass}[Estimated nuisances] \label{ass:bounded_nuisance} Assume that for each $d \in \{0,1\}, t \in \mathcal{T}_n,  i \in \{1, \dots, n\}$, and $\hat{m}^{(i)}(\cdot), \hat{e}^{(i)}(\cdot)$ as in Algorithm \ref{alg:adaptive}, $|\hat{m}^{(i)}(d, t, Z_i, |N_i|)| < \Gamma$ almost surely, for a finite constant $\Gamma$ and $\hat{e}^{(i)}(d,t, Z_{k \in N_i}, R_{k \in N_i}, Z_i, R_i, |N_i|) \in (\gamma \delta_n, 1 - \gamma \delta_n)$, almost surely, for $\gamma, \delta_n$ as defined in Assumption \ref{ass:quasi}.    
\end{ass}


   

The rate of convergence here also depends on the product of the mean-squared error of the estimated conditional mean function and propensity score, averaged over the population covariates and number of neighbors: 
\begin{equation}
\small 
\begin{aligned} 
&\mathcal{R}_n(A, Z) =   \frac{1}{n} \sum_{i=1}^n \mathbb{E}\left[\sup_{d,t} \Big( \hat{m}^{(i)}(d, t, Z_i, |N_i|) - m(d, t, Z_i, |N_i|) \Big)^2\Big| A, Z, R_i = 1\right]\\
& \mathcal{B}_n(A, Z) =  
\frac{1}{n} \sum_{i=1}^n \mathbb{E}\left[\sup_{d,t}   \Big( \frac{1}{\hat{e}^{(i)}(d, t,  Z_{k \in N_i}, R_{k \in N_i},  Z_i, |N_i|)} - \frac{1}{e(d, t, Z_{k \in N_i},  R_{k \in N_i}, Z_i, |N_i|)} \Big)^2 \Big| A, Z, R_i = 1\right],  
\end{aligned} 
\end{equation} 
where 
$\hat{m}^{(i)}, \hat{e}^{(i)}$ are the estimated functions for unit $i$, as defined in Algorithms \ref{alg:sketch}, \ref{alg:adaptive}.

\begin{thm} \label{thm:dr} Let Assumptions \ref{ass:sutnva}, \ref{ass:ignorability}, \ref{ass:quasi}, \ref{ass:finite_vc}, \ref{ass:a2}, \ref{ass:bounded_nuisance} hold. Suppose that $\hat{m}, \hat{e}$ are estimated as in Algorithm \ref{alg:adaptive}. Then 
    $$
\mathbb{E}\Big[\sup_{\pi \in \Pi_n} W_{A, Z}(\pi) - W_{A, Z}(\hat{\pi}_{\hat{m}, \hat{e}})\Big| A, Z\Big]= \mathcal{O}\Big(n_e^{-\xi} + \sqrt{\mathcal{R}_n(A, Z) \times \mathcal{B}_n(A, Z)} \Big) .
    $$
almost surely,    for $\xi \in (0, \frac{1}{2}]$ as defined in Assumption \ref{ass:a2}. 
\end{thm} 

\begin{proof}[Proof of Theorem \ref{thm:dr}] 
The proof leverages the network cross-fitting argument (Algorithm \ref{alg:adaptive}) combined with similar techniques used to derive Theorem \ref{thm:thmmain}. The rate $n_e^{-\xi}$ follows from Assumption \ref{ass:a2}. See Appendix \ref{app:estimated_1} for the complete derivation. 
\end{proof}

Theorem \ref{thm:dr} states that the regret bound depends on two components. The first component depends on the convergence rate of the maximum degree, overlap constant, and experiment size, similar to what was discussed in the presence of a known propensity score (e.g., Corollary \ref{cor:1_main}). For a bounded degree as in Example \ref{exmp:bounded}, $\xi = 1/2$, and $\xi < 1/2$ otherwise. The second component depends on the estimation error of the nuisance functions, and in particular, it depends on the \textit{product} of their convergence rates, in the same spirit of standard conditions in the $i.i.d.$ setting \citep[e.g.,][]{farrell2015robust}.

\begin{rem}[Convergence rate of nuisance functions]
 \label{rem:jj} 
 Appendix \ref{sec:lasso} shows that using Algorithm \ref{alg:adaptive}, $\sqrt{\mathcal{R}_n(A, Z) \times \mathcal{B}_n(A, Z)} = \mathcal{O}(\mathcal{N}_n^2 n_e^{-(\zeta_m + \zeta_e)}/\delta_n)$, where $n_e^{-2\zeta_m}$, and $n_e^{-2\zeta_e}/\delta_n^2$ are the rate of convergence of the mean squared error of the conditional mean and propensity score, respectively, on a sample of \textit{independent observations}.  As a result, whenever $\mathcal{N}_n^{1/2} n_e^{-(\zeta_m + \zeta_e)} = n_e^{-1/2}$ (e.g., $n_e^{-\zeta_m} = n_e^{-\zeta_e} = \mathcal{N}_n^{-1/4} n_e^{-1/4}$), it follows that $\sqrt{\mathcal{R}_n(A, Z) \times \mathcal{B}_n(A, Z)} = \mathcal{O}(n_e^{-\xi})$. Convergence rates for the estimation error of order $\mathcal{N}_n^{1/2} n^{-(\zeta_m + \zeta_e)} = n_e^{-1/2}$ imply that the estimation error of the nuisance functions does \textit{not} affect the rate of the regret bound in Theorem \ref{thm:thmmain} in the absence of estimation error. Appendix \ref{sec:lasso} presents formal results. 
\qed 
\end{rem}

\subsection{Optimization} \label{sec:optimization} 

Next, I discuss the optimization procedure. For simplicity, consider the most agnostic case where $T_i = \sum_{k \in N_i} D_k$ denotes the sum of treated neighbors. Similar reasoning applies to $T_i$ being a known function of the sum of treated neighbors.
Define the estimated effect of assigning to unit $i$ treatment $d$, after treating $t$ neighbors: 
\begin{equation}
\small 
\begin{aligned} 
&q_i(d, t) =  \left\{ \frac{1\{\sum_{k \in N_i} D_k = t, D_i = d \}}{e\Big(d, t,   Z_{k \in N_i}, R_{k \in N_i}, Z_i, |N_i|\Big)}\Big(Y_i - m^c\Big(d, t, Z_i, |N_i|\Big)\Big) + m^c\Big(d, t,  Z_i, |N_i|\Big)\right\},
\end{aligned} 
\end{equation} 
where I omit the dependence of $q_i(\cdot)$ with $m^c$ and $e$ for the sake of brevity. Second, let 
$
B_i(\pi, h) = 1\Big\{\sum_{k \in N_i} \pi(X_k) = h \Big\} 
$
be the indicator of whether $h$ neighbors of individual $i$ have been treated under policy $\pi$. We have the following: 
\begin{equation}
\small  
\begin{aligned} 
\sum_{h = 0}^{|N_i|}\Big\{\Big(q_i(1, h) - q_i(0,h)\Big)\pi(X_i) B_i(\pi, h) + B_i(\pi, h) q_i(0,h)\Big\} = q_i\Big(\pi(X_i),\sum_{k \in N_i} \pi(X_k)\Big).
\end{aligned} 
\end{equation} 

Namely, each element in the sum is weighted by the indicator $B_i(\pi,h)$, and only one of these indicators is equal to one. 
I can then define variables $p_i, p_i = \pi(X_i), \pi \in \Pi_n$ that denote the treatment assignment of each unit $i$ either sampled $(R_i = 1)$ or friend of a sampled unit $(R_i^f = 1)$. 
For example, for 
$
\pi(X_i) = 1\{X_i^\top \beta \ge 0\}, \beta \in \mathcal{B}, 
$ \citep{florios2008exact}, 
$$
\small 
\begin{aligned} 
\frac{X_i^\top \beta}{|C_i|} < p_i \le \frac{X_i^\top \beta}{|C_i|} + 1, \quad C_i > \sup_{\beta \in \mathcal{B}} |X_i^\top \beta|, \quad  p_i \in \{0,1\}, 
\end{aligned} 
$$
where $p_i$ is equal to one if $X_i^\top \beta$ is positive, and zero otherwise.
The key intuition is to introduce additional variables to write $B_i(\pi, h)$ using mixed-integer linear constraints. Define
$$
t_{i,h,1} = 1\left\{\sum_{k \in N_i} p_k \ge h\right\}, \quad  
t_{i,h,2} = 1\left\{\sum_{k \in N_i} p_k \le h\right\}, \quad h \in  \{0, \cdots,|N_i|\}.
$$
It follows that $t_{i,h,1} + t_{i,h,2} - 1 = B_i(\pi, h)$, and that such variables admit a mixed-integer linear program characterization. Formally, the optimization program is  
\begin{equation} \label{eqn:myopt} 
\small 
\begin{aligned} 
&\max_{\{u_{i, h}\}, \{p_i\}, \{t_{i,1,h}, t_{i,2,h}\}}  \sum_{i = 1}^n \sum_{h = 0}^{|N_i|}R_i \Big\{\Big(q_i(1, h) - q_i(0,h)\Big)u_{i,h}  + q_i(0,h)(t_{i,h,1} + t_{i,h,2} - 1) \Big\}
\end{aligned} 
\end{equation} 
under the following constraints:
\begin{equation} \label{eqn:constraints} 
\small 
\begin{aligned}
&(A) \quad p_i = \pi(X_i), \quad \pi \in \Pi_n, \quad \forall i: R_i = 1 \text{ or } R_i^f = 1  \\
&(B) \quad \frac{p_i + t_{i,h,1} + t_{i,h,2}}{3} - 1 < u_{i,h} \le  \frac{p_i + t_{i,h,1} + t_{i,h,2}}{3}, u_{i,h} \in \{0,1\} \quad \forall h \in \{0, \cdots , |N_i|\}, \forall i: R_i = 1 \\ 
&(C) \quad  \frac{(\sum_k A_{i,k} p_k - h)}{ |N_i| + 1} < t_{i,h,1} \le  \frac{(\sum_k A_{i,k} p_k - h)}{ |N_i|+1} + 1,  t_{i,h,1} \in \{0,1\}, \quad \hspace{-2.6mm} \forall h \in \{0, \cdots , |N_i|\}, \forall i: R_i =1 \\
&(D) \quad  \frac{( h -\sum_k A_{i,k} p_k)}{  |N_i| + 1} < t_{i,h,2} \le \frac{(h - \sum_k A_{i,k} p_k)}{  |N_i| + 1} + 1,  t_{i,h,2} \in \{0,1\},\hspace{-2.6mm} \quad  \forall h \in  \{0, \cdots ,|N_i|\}, \forall i: R_i = 1.
\end{aligned} 
\end{equation} 
The first constraint can be replaced by methods discussed in previous literature, such as maximum scores \citep{florios2008exact}. By contrast, the additional constraints are due to interference. In practice, including additional (superfluous) constraints stabilizes the optimization problem. These are $\sum_h (t_{i,h,1} + t_{i,h,2} - 1) = 1$ for each $i$ and $\sum_i \sum_h u_{i,h} = \sum_i p_i$.
Whenever units have no neighbors, the objective function is proportional to the one discussed in \cite{KitagawaTetenov_EMCA2018} under no interference. Therefore, the formulation generalizes the MILP formulation to the case of interference. 

\begin{thm} \label{prop:optimization} Let $T_i = \sum_{k \in N_i} D_k$. Then
$
\hat{\pi} \in \text{argmax}_{\pi \in \Pi_n} W_n(\pi, m^c, e),
$
if and only if it maximizes Equation \eqref{eqn:myopt} with constraints in Equation \eqref{eqn:constraints}. 
\end{thm} 

The proof of Theorem \ref{prop:optimization} follows directly from the argument in the current section.


 \subsection{Derivation of Theorem \ref{thm:thmmain}: main steps} \label{sec:proof}

This section includes a sketch of the proof of Theorem \ref{thm:thmmain}, whereas Appendix \ref{app:regret} presents formal definitions and derivations. Readers not interested in the proof of Theorem \ref{thm:thmmain} can skip to Section \ref{sec:extensions} (or \ref{sec:application}). For brevity, in the argument below, I further assume $Y_i \in [-\Gamma', \Gamma']$ for a finite constant $\Gamma' < \infty$; that is, the outcome is uniformly bounded. Appendix \ref{app:regret} presents derivations for unbounded outcomes. 
Because $\Pi_n \subseteq \Pi$, it follows that 
 \begin{equation} \label{eqn:help3_main_text}
 \small 
 \begin{aligned} 
 \mathbb{E}\left[\sup_{\pi \in \Pi_n} W_{A, Z}(\pi) - W_{A, Z}(\hat{\pi}_{m^c, e} )\Big| A, Z \right] & \le 2 \mathbb{E}\left[\sup_{\pi \in \Pi_n}\Big|W_n(\pi, m^c, e) - W_{A, Z}(\pi)\Big| | A, Z\right] \\ 
 & \le 2 \mathbb{E}\left[\sup_{\pi \in \Pi}\Big|W_n(\pi, m^c, e) - W_{A, Z}(\pi)\Big| | A, Z\right],
 \end{aligned} 
 \end{equation}  
 our focus will be bounding the right-hand side of Equation \eqref{eqn:help3_main_text}. Define 
$$
\small 
\begin{aligned} 
Q_i(\pi, A, Z) = R_i \left[\frac{I_i(\pi)}{e_i(\pi)}\Big(Y_i - m_i^c(\pi)\Big) + m_i^c(\pi)\right],  
\end{aligned}
$$  
where the dependence with $e, m^c$ is suppressed for convenience. Define $\mathcal{Q}_n(\pi, A, Z)$ as the \textit{joint} distribution, of $Q_i$, namely  $\Big(Q_i(\pi, A, Z)\Big)_{i=1}^n \Big| A, Z \sim \mathcal{Q}_n(\pi, A, Z)$, for given $\pi, A, Z$. 

Define $(\sigma_i)_{i=1}^n$ $i.i.d.$ Rademacher random variables independent of observables and unobservables ($P(\sigma_i = 1) = P(\sigma_i = -1) = 1/2$) and $\mathbb{E}_\sigma[\cdot]$ denotes the expectation only with respect to $(\sigma_i)_{i=1}^n$, conditional on observables and unobservables. 
By Lemma \ref{prop:welfare} $\mathbb{E}[W_n(\pi)| A,Z] = W_{A, Z}(\pi)$ for all $\pi \in \Pi$.  


\paragraph{Symmetrization with network data} Next, I extend the symmetrization argument \citep[e.g., Lemma 6.4.2 in][]{vershynin2018high} to the context of this paper. 
Define \\ $\Big(Q_i'(\pi, A, Z)\Big)_{i=1}^n \Big| A, Z \sim \mathcal{Q}_n(\pi, A, Z)$, an independent copy of $\Big(Q_i(\pi, A, Z)\Big)_{i=1}^n$, conditional on $(A, Z)$. It  follows 
 \begin{equation} 
 \label{eqn:helper1_text}
 \small 
 \begin{aligned} 
\eqref{eqn:help3_main_text} & \le \mathbb{E}\left[\sup_{\pi \in \Pi}\Big|\frac{1}{n_e} \sum_{i=1}^n \Big[Q_i(\pi, A, Z) - Q_i'(\pi, A, Z)\Big]\Big| | A, Z \right] \quad (\because \text{Jensen's inequality}). 
 \end{aligned} 
  \end{equation} 
Ideally, using standard symmetrization arguments, I would like to bound the right-hand side in Equation \eqref{eqn:helper1_text}. Unfortunately, this is not possible because of dependence. I instead partition observations into groups of conditionally independent random variables. I then obtain bounds that depend on the number of such groups. 
Let $A^2$ be the adjacency matrix obtained by connecting neighbors and two-degree neighbors under $A$. Let $\chi_n(A^2)$ be the smallest number of groups such that each group does not contain two units that either are neighbors or share a common neighbor under $A$, and $\mathcal{C}_n^2 = \{\mathcal{C}_n^2(g)\}_{g=1}^{\chi_n(A^2)}, \mathcal{C}_n^2(g) \subseteq \{1, \cdots, n\}$, the smallest set of such groups.  Then
 \begin{equation} \label{eqn:helper5_text}
 \small 
 \begin{aligned} 
 & \mathbb{E}\left[\sup_{\pi \in \Pi}\Big|\frac{1}{n_e} \sum_{i=1}^n \Big[Q_i(\pi, A, Z) - Q_i'(\pi, A, Z)\Big]\Big|  | A, Z\right] \quad (\because \text{ triangular inequality}) \\ &\le 
 \sum_{g \in \{1, \cdots, \chi_n(A^2)\}} \underbrace{\mathbb{E}\left[\sup_{\pi \in \Pi}\Big|\frac{1}{n_e} \sum_{i \in \mathcal{C}_n^2(g)} \Big[Q_i(\pi, A, Z) - Q_i'(\pi, A, Z)\Big]\Big|  | A, Z\right]}_{(II)}. 
\end{aligned} 
 \end{equation}  
 Note that $Q_i$ equals zero if $R_i = 0$. Therefore, under Assumption \ref{ass:quasi} (ii), it follows that $Q_i$ can be written as a function of $\Big[R_i\Big(\varepsilon_i, R_i, \varepsilon_{D_i}, R_i^f, R_{j \in N_i}, R_{j \in N_i}^f, \varepsilon_{D_{j \in N_i}}, Z_i, |N_i|, Z_{k \in N_i}\Big)\Big]$, where $R_i^f = 1\{\sum_k A_{i,k} R_k > 0\}$. For each $j \in N_i$, $R_{j}^f$ equals one almost surely conditional on $R_i = 1$. $R_i^f$ is instead a deterministic function of $R_{j \in N_i}$. As a result, because $Q_i = 0$ if $R_i = 0$ almost surely, one can write $Q_i$ only as a function of $\Big[R_i\Big(\varepsilon_i, R_i, \varepsilon_{D_i}, R_{j \in N_i}, \varepsilon_{D_{j \in N_i}}, Z_i, |N_i|, Z_{k \in N_i}\Big)\Big]$, its dependence with $R_{j \in N_i}^f$ can be dropped.
 
 Under the distributional assumptions of each of these components, it follows that $Q_i$ are jointly independent if they are not neighbors and do not share a common neighbor conditional on $A, Z$.\footnote{In particular, we leverage here Assumption \ref{ass:sutnva} (interference is local); Assumption \ref{ass:quasi} (ii) (treatments are conditionally independent); Assumption \ref{ass:ignorability} (B) (unobservables are conditionally independent if two individuals do not share a common neighbor). I relax Assumption \ref{ass:ignorability} (B) in Section  \ref{sec:extensions}.} Because $Q_i, Q_i' | A, Z$ have the same marginal distribution by construction,  
 $$
 \small 
 \begin{aligned} 
 (II) \le 2 \mathbb{E}\Big[\underbrace{\mathbb{E}_\sigma\Big[ \sup_{\pi \in \Pi} \Big| \frac{1}{n_e} \sum_{i \in \mathcal{C}_n^2(g)} \sigma_i Q_i(\pi, A, Z)\Big| \Big]}_{(III)} \Big| A, Z \Big].
 \end{aligned}
 $$

  \paragraph{Bound on the function class complexity} I control $(III)$ with Lemma \ref{lem:finallemma}. The idea of the lemma is the following. First, note that here $Q_i(\pi, \cdot)$ depends on $\pi$ through $\Big(\pi(X_i), \sum_{k \in N_i} \pi(X_k)\Big)$. 
  I show that $Q_i(\pi, A, Z)$ is Lipschitz in $\Big(\sum_{k \in N_i} \pi(X_k)\Big)$ with the Lipschitz contant proportional to $\frac{\Gamma'}{\gamma \delta_n}$. I then leverage extensions of the Ledoux-Talagrand contraction inequality \citep[Lemma \ref{lem:ledoux}, which extends Theorem 4.12 in][]{ledoux2011probability} to show 
\begin{equation} \label{eqn:final_main}
\small 
\begin{aligned} 
 \mathbb{E}_\sigma\left[\sup_{\pi \in \Pi}\Big|\frac{1}{n_e} \sum_{i \in \mathcal{C}_n^2(g)} \sigma_i Q_i(\pi, A, Z)\Big| \right] \le \frac{\bar{C} \Gamma'}{\gamma \delta_n}  \mathbb{E}_\sigma\left[\sup_{\pi \in \Pi}\Big|\frac{1}{n_e} \sum_{i \in \mathcal{C}_n^2(g)} R_i \sigma_i \Big(\sum_{k \in N_i} \pi(X_k)\Big) \pi(X_i)\Big| \right]
 \end{aligned} 
\end{equation} 
for a universal constant $\bar{C} < \infty$. Using Theorem 5.22 in \cite{wainwright2019high}, I can bound the right-hand side in Equation \eqref{eqn:final_main}, by an integral of the covering number of a function class obtained from $ \Big(\sum_{k \in N_i} \pi(x_k)\Big) \pi(x_i), \pi \in \Pi$ -- which we can bound by a function of the maximum degree and the VC dimension of $\Pi$ (Lemma \ref{lem:boundnumber}) -- and $\frac{\sqrt{\sum_{i=1}^n R_i 1\{i \in \mathcal{C}_n^2(g)\}}}{n_e}$. 

\paragraph{Conclusions} Collecting terms, for a universal constant $\bar{C} < \infty$, I show 
$$
\small 
\begin{aligned} 
\eqref{eqn:help3_main_text} & \le \bar{C} \times  \sum_{g=1}^{\chi_n(A^2)} \times \frac{\Gamma'}{\gamma \delta_n}  \times \sqrt{\log(\mathcal{N}_n) \mathcal{N}_n \mathrm{VC}(\Pi)} \times \mathbb{E}\left[\frac{\sqrt{\sum_{i=1}^n R_i 1\{i \in \mathcal{C}_n^2(g)\}}}{n_e}\Big| A, Z\right] \\ 
& \le \bar{C} \times  \sqrt{\chi_n(A^2)} \times \frac{\Gamma'}{\gamma \delta_n}  \times \sqrt{\log(\mathcal{N}_n) \mathcal{N}_n \mathrm{VC}(\Pi)} \times \mathbb{E}\left[\frac{\sqrt{\sum_{i=1}^n R_i}}{n_e}\Big| A, Z\right] \quad (\because \text{concavity of } \sqrt{x}). 
\end{aligned} 
$$
The first term $\sqrt{\chi_n(A^2)}$ captures the dependence structure. By \cite{brooks1941colouring}'s theorem, $\chi_n(A^2) \le 2 \mathcal{N}_n^2$ (see Lemma \ref{lem:boundnumber}). The second term captures Lipschitz-continuity of the objective function and depends on the overlap $1/\delta_n$. The third term captures the complexity of the function class of interest, increasing in the maximum degree. The last term captures concentration in the sample size. Using Jensen's inequality, $\mathbb{E}\Big[\frac{\sqrt{\sum_{i=1}^n R_i}}{n_e}\Big] \le 1/n_e^{1/2}$.  In Theorem \ref{thm:thmmain}, $\Gamma$ replaces $\Gamma'$ under bounded moments, instead of bounded outcomes. 

\begin{rem}[Independence of sampling indicators] \label{rem:sampling_local_dependent} My results  extend to settings where sampling indicators are locally dependent. For instance, if indicators are dependent between two-degree neighbors, the proof above follows verbatim, because the sampling indicators in the set $\mathcal{C}_n^2(g), g \in \{1, \cdots, \chi(A_n^2)\}$ are independent. \qed 
\end{rem} 

\begin{rem}[Regret conditional on $\varepsilon_i$] \label{rem:epsilon} For known propensity score and uniformly bounded outcome, the proof technique follows verbatim conditional on $\varepsilon_i$, once I define welfare as $\frac{1}{n} \sum_{i=1}^n r\Big(\pi(X_i), \sum_{k \in N_i} \pi(X_k), Z_i, |N_i|, \varepsilon_i\Big)$, conditional on $(\varepsilon_i)_{i=1}^n$, as in a design-based framework \citep[e.g.][]{leung2021network}. In particular, we can invoke verbatim the symmetrization argument in Equation  \eqref{eqn:helper1_text} and follow the same steps, providing stronger guarantees that hold conditional on $(\varepsilon_i)_{i=1}^n$ (without assumptions on $(\varepsilon_i)_{i=1}^n$). However, with an \textit{unknown} propensity score, convergence rates of the estimators in Section \ref{sec:alg1} depend on the \textit{distribution} of $\varepsilon_i$: regret guarantees can only be obtained \textit{in expectation}, after integrating welfare over $\varepsilon_i$ as in \cite{KitagawaTetenov_EMCA2018}, \cite{athey2017efficient}.  \qed 
\end{rem} 
 
 \section{Main extensions} \label{sec:extensions}

I discuss here trimming with poor overlap, higher-order dependence, different target and sample units, and non-reversible treatments. Appendix \ref{app:ext2} contains additional extensions. 

\subsection{Trimming to control overlap} \label{sec:trimming}

In this subsection, I provide regret bounds whenever a few units may present a large degree. I consider the setting where $T_i = \sum_{k \in N_i} D_k$. To guarantee overlap, I introduce the following trimming estimator: 
\begin{equation} \label{eqn:trimming} 
W_n^{tr}(\pi, m^c, e; \kappa_n) = \frac{1}{n} \sum_{i=1}^n  R_i \left\{ \frac{I_i(\pi)}{e_i(\pi)} \Big(Y_i - m_i^c(\pi)\Big) 1\Big\{|N_i| \le \log_\gamma(\kappa_n)\Big\} + m_i^c(\pi)\right\}, 
\end{equation} 
with $e_i(\pi), m_i^c(\pi), I_i(\pi)$ as in Equation \eqref{eqn:welf}. Here, $\log_\gamma(\kappa_n)$ defines the trimming constant, as the logarithm in scale $\gamma$ of a user-specific $\kappa_n$ (with $\gamma$ in Assumption \ref{ass:quasi}). 

The trimming estimator builds on the following idea: it excludes the direct effect on the largely connected nodes (with more than $\log_\gamma(\kappa_n)$ neighbors) but keeps information from the spillovers that such nodes generate. This is because nodes with most connections are those for which overlap restrictions are more likely to fail. Define
  $$
 \hat{\pi}_{\kappa_n}^{tr} \in \mathrm{arg} \max_{\pi \in \Pi_n} W_n^{tr}(\pi, m^c, e; \kappa_n), \quad P_n\Big(|N_i| \ge \log_\gamma(\kappa_n)\Big) = \frac{1}{n} \sum_{i=1}^n 1\Big\{|N_i| \ge \log_\gamma(\kappa_n)\Big\}.
  $$

\begin{thm} \label{thm:trimming}    Suppose that $P_n\Big(|N_i| \ge \log_\gamma(\kappa_n)\Big) < c$, for a constant $c < 1$. Let $T_i = \sum_{k \in N_i} D_k$, and let Assumptions \ref{ass:sutnva}, \ref{ass:ignorability}, \ref{ass:quasi}, \ref{ass:finite_vc}, \ref{ass:general4} hold. Then  
    $$
\mathbb{E}\Big[\sup_{\pi \in \Pi_n} W_{A, Z}(\pi) - W_{A, Z}(\hat{\pi}_{\kappa_n}^{tr} )\Big| A, Z \Big] = \mathcal{O}\left(\frac{\mathcal{N}_n^{3/2}}{\kappa_n} \sqrt{\frac{\log(\mathcal{N}_n) \mathrm{VC}(\Pi)}{n_e}} + P_n\Big(|N_i| \ge \log_\gamma(\kappa_n)\Big)\right). 
    $$
\end{thm}

\begin{proof}[Proof of Theorem \ref{thm:trimming}]
See Appendix \ref{app:regret}.
\end{proof} 
Theorem \ref{thm:trimming} shows we can improve the regret bound for a suitable choice of $\kappa_n$ under restrictions on the degree distribution. For instance, suppose $\sqrt{n}$-many individuals have a degree that can grow in $n$, whereas all other units have a degree bounded by at most $\log_\gamma(\kappa)$, for a constant $\kappa$ independent of $n$. In this case, $P_n(|N_i| \ge \log_\gamma(\kappa)) = \mathcal{O}(\sqrt{\frac{\alpha}{n_e}})$, and the regret is of order $\mathcal{O}\Big(\frac{\mathcal{N}_n^{3/2}}{\kappa} \sqrt{\frac{\log(\mathcal{N}_n) \mathrm{VC}(\Pi)}{n_e}}\Big)$, independent of $\delta_n$. 
Theorem \ref{thm:trimming} illustrates how information can be leveraged from the \textit{degree distribution} to improve convergence rates.

\subsection{Regret with higher-order dependence} \label{sec:jj}

Next, I characterize regret bounds in settings where individuals can depend on friends up to the degree of order $M$, where $M$ is a finite number and unknown. To simplify exposition, I assume the outcome is uniformly bounded.

\begin{ass}[Higher-order dependence and bounded outcome] \label{ass:higher_order} Suppose that for some unknown $M \ge 2$, (A) $\varepsilon_i \perp (\varepsilon_j)_{j \not \in \cup_{k=1}^{M} N_{i, k}} \Big| A, Z$, where $N_{i,k}$ denotes the set of connection of $i$ of degree $k$. Suppose in addition that  (B) $Y_i \in [-\Gamma', \Gamma'],$ for a positive constant $\Gamma' < \infty$. 
\end{ass} 

Under Assumption \ref{ass:higher_order}, unobservables can depend on individuals of at most degree $M$. Suppose $M$ is unknown and researchers do not have information from higher-order neighbors.
Define $m^c: \{0,1\} \times \mathbb{Z} \times \mathcal{Z} \times \mathbb{Z} \mapsto [-\Gamma', \Gamma']$ for some finite $\Gamma' < \infty$, $e^c(\cdot; |N_i|): \mathcal{Z}^{|N_i|} \times \{0,1\}^{|N_i|} \times \mathcal{Z} \mapsto (\gamma \delta_n, 1 - \gamma \delta_n)$, the \textit{pseudo-true} conditional mean function and propensity score, and $\hat{m}, \hat{e}$ their corresponding estimators constructed arbitrarly (e.g., pooling information from all sampled units). Let  
\begin{equation}  \label{eqn:helper_text}
\small 
\begin{aligned}
&\tilde{\mathcal{R}}_n(A, Z) = \frac{1}{n} \sum_{i=1}^n  \mathbb{E}\left[\sup_{d, t}   \Big( \hat{m}(d, t, Z_i, |N_i|) - m^c(d, t,  Z_i, |N_i|)\Big )^2 | A, Z\right]. 
\\ & \tilde{\mathcal{B}}_n(A, Z) =  \frac{1}{n} \sum_{i=1}^n   \mathbb{E}\left[\sup_{d, t}  \Big(\frac{1}{e^c(d, t,Z_{k \in N_i}, R_{k \in N_i},  Z_i)} - \frac{1}{\hat{e}(d, t,Z_{k \in N_i}, R_{k \in N_i},  Z_i)} \Big )^2 | A, Z\right]
\end{aligned}  
\end{equation} 
denote the mean-squared errors of the estimators obtained from all sampled units, averaged over the population covariates and number of neighbors.  Different from Theorem \ref{thm:dr}, we do not need to condition on $R_i = 1$ in Equation \eqref{eqn:helper_text} because no cross-fitting is used, and the estimated nuisance function is independent of $i$'s index.

\begin{thm} \label{thm:estimatedoutcome} Let Assumptions \ref{ass:sutnva}, \ref{ass:quasi} hold, and Condition (C) in \ref{ass:ignorability}, Assumptions \ref{ass:finite_vc}, \ref{ass:a2}, \ref{ass:general4}, \ref{ass:bounded_nuisance}, \ref{ass:higher_order} hold.  Assume either (or both) (i) $e^c(\cdot) = e(\cdot)$, or (ii) Assumption \ref{ass:ignorability} (A) holds and $m^c = m$.  
  Then, for $M \ge 2$, $\xi \in (0, 1/2]$ as in Assumption \ref{ass:a2}:
    $$
    \small 
    \begin{aligned} 
\mathbb{E}\Big[\sup_{\pi \in \Pi_n} W_{A, Z}(\pi) - W_{A, Z}(\hat{\pi}_{\hat{m}, \hat{e}} )\Big| A, Z\Big] =   \mathcal{O}\left(M\mathcal{N}_n^{M/2 - 1} n_e^{-\xi} + \frac{1}{\delta_n}  \sqrt{\max\Big\{\tilde{\mathcal{R}}_n(A, Z), \tilde{\mathcal{B}}_n(A, Z)\Big\}}\right).
\end{aligned} 
    $$
\end{thm} 

\begin{proof}[Proof of Theorem \ref{thm:estimatedoutcome}] See Appendix \ref{app:1_first} 
\end{proof} 
 Theorem \ref{thm:estimatedoutcome} provides a uniform bound on the regret, and it is double robust to correct specification of the conditional mean and the propensity score. The theorem's result depends on the convergence rate of $\hat{e}$ and $\hat{m}$ to their $pseudo$-true value. For parametric estimators of the conditional mean and the propensity score and bounded degree, the regret bounds scale at rate $1/\sqrt{n_e}$, divided by the overlap parameter.  For general machine-learning estimators, the rate can be slower than the parametric one, reflecting the ``cost'' of the lack of knowledge of the degree of dependence $M$. Here, $\mathcal{N}_n^{M/2 - 1}$ captures higher-order dependence. Theorem \ref{thm:estimatedoutcome} does not require that Assumption \ref{ass:ignorability} (A) holds in settings with a correctly specified propensity score, assuming $\hat{m}^c$ converges to \textit{some} pseudo-true value $m^c$.

\subsection{Expected regret with a different target population} \label{app:average_bound}

This subsection compares regret guarantees when units are either drawn from the (larger) target population as described in Section \ref{sec:identification}, or units are drawn from a \textit{different} population from the target population.  Following \cite{KitagawaTetenov_EMCA2018}, and to simplify exposition in this subsection, we  
consider a policy function class $\Pi_n = \Pi$ where $\Pi$ is not data dependent.\footnote{We assume that $\Pi_n = \Pi$ not to define the joint distribution of $(X, A', Z')$ in the definition below. } Consider a population with $n$ individuals, connected under adjacency matrix $A'$ and with covariates matrix $Z'$. For given $(A', Z')$, welfare is defined as 
\begin{equation} \label{eqn:welfare_target} 
W_{A', Z'}(\pi) = \frac{1}{n} \sum_{i=1}^n m\Big(\pi(X_i), \sum_k A_{i,k}' \pi(X_k'), Z_i', \sum_k A_{i,k}'\Big), \quad X_i' \subseteq Z_i'.
\end{equation} 
Consider two notions of regret, the \textit{conditional} and \textit{expected} regret, defined respectively as 
\begin{equation} \label{eqn:expected_regret} 
\begin{aligned} 
\mathcal{R}_{\Pi, A' , Z'}^{\mathrm{cond}} & = \mathbb{E}\left[\sup_{\pi \in \Pi} W_{A', Z'}(\pi) - W_{A', Z'}(\hat{\pi}_{m^c, e}) \Big| A', Z'\right], \\ \mathcal{R}_{\Pi}^{\mathrm{exp}} & = \sup_{\pi \in \Pi} \mathbb{E}\Big[W_{A',Z'}(\pi)\Big] - \mathbb{E}\Big[W_{A',Z'}(\hat{\pi}_{m^c, e})\Big].
\end{aligned}
\end{equation}  
The conditional regret is a function of the target population adjacency matrix and covariates $Z'$, whereas the expected regret takes expectation over $(A', Z')$. The expected regret is (implicitly) a function of the joint distribution of $(A', Z', A, Z)$, since it integrates over the distribution of $(A', Z')$ and $\hat{\pi}$ estimated on the sampled units.

 \begin{figure}[!ht]
 \centering
    \begin{tikzpicture}[auto,
 block/.style ={rectangle, draw=blue, thick, fill=blue!20, text width=5em,align=center, rounded corners, minimum height=2em},
 block1/.style ={rectangle, draw=blue, thick, fill=blue!20, text width=5em,align=center, rounded corners, minimum height=2em},
 line/.style ={draw, thick, -latex',shorten >=2pt},
 cloud/.style ={draw=red, thick, ellipse,fill=red!20,
 minimum height=1em}]

\coordinate (1) at (-1.7,0.8);
\coordinate (2) at (9.5,0.8);
\coordinate (3) at (9.5,5.5);
\coordinate (4) at (-1.7,5.5);

\coordinate (11) at (3.1,0.6);
\coordinate (12) at (14,0.6);
\coordinate (13) at (14,5.7);
\coordinate (14) at (3.1,5.7);

\draw[thick] (1)--(2)--(3)--(4)-- cycle;
\draw[dotted, thick] (11)--(12)--(13)--(14)-- cycle;

  \node[draw, fill = pink, circle] (aaa) at (0, 3.2) {};
  \node[draw, fill = green, circle] (bbb) at (0, 4.1) {};
  \node[draw, fill = green, circle] (ccc) at (-1, 3.5) {};
 \node[draw, fill =green, circle] (ddd) at (-0.8, 2.5) {};
  \node[draw, fill =pink,  circle] (eee) at (0.7, 2.5) {};
  \node[draw, fill = pink, circle] (fff) at (1, 3.5) {};
 \node[circle] (hhh) at (0.9, 4.4) {};
 \node[circle] (iii) at (1.7, 2.9) {};
  \node[circle] (lll) at (1.7, 4) {};
   \node[circle] (mmm) at (-0.9, 4.4) {};
   \node[circle] (nnn) at (-1.6, 3) {};
   \node[circle] (ooo) at (-1.4, 2) {};
   \node[circle] (ppp) at (-1.7, 3.9) {};
     \node[circle] (qqq) at (0, 4.8) {};
       \node[ circle] (rrr) at (0, 1.9) {};
       \node[circle] (sss) at (1.2, 1.9) {};

  \node[draw, fill = pink, ultra thick, circle] (aa) at (6, 3.2) {};
  \node[draw, fill = pink, circle] (bb) at (6, 4.1) {};
  \node[draw, fill = green, circle] (cc) at (5, 3.5) {};
 \node[draw, fill = pink,  circle] (dd) at (5.2, 2.5) {};
  \node[draw, fill = pink, circle] (ee) at (6.7, 2.5) {};
  \node[draw, fill = pink, circle] (ff) at (7, 3.5) {};
 \node[circle] (hh) at (6.9, 4.4) {};
 \node[circle] (ii) at (7.7, 2.9) {};
  \node[circle] (ll) at (7.7, 4) {};
   \node[circle] (mm) at (5.1, 4.4) {};
   \node[circle] (nn) at (4.4, 3) {};
   \node[circle] (oo) at (4.6, 2) {};
   \node[circle] (pp) at (4.3, 3.9) {};
     \node[circle] (qq) at (6, 4.8) {};
       \node[ circle] (rr) at (6, 1.9) {};
       \node[circle] (ss) at (7.2, 1.9) {};

  \node[draw, fill = green, circle] (a) at (12, 3.2) {};
  \node[draw, fill = green, circle] (b) at (12, 4.1) {};
  \node[draw, fill = pink, circle] (c) at (11, 3.5) {};
 \node[draw, fill = green,  circle] (d) at (11.2, 2.5) {};
  \node[draw, fill = pink, circle] (e) at (12.7, 2.5) {};
  \node[draw, fill = green, circle] (f) at (13, 3.5) {};
 \node[circle] (h) at (12.9, 4.4) {};
 \node[circle] (i) at (13.7, 2.9) {};
  \node[circle] (l) at (13.7, 4) {};
   \node[circle] (m) at (11.1, 4.4) {};
   \node[circle] (n) at (10.4, 3) {};
   \node[circle] (o) at (10.6, 2) {};
   \node[circle] (p) at (10.3, 3.9) {};
     \node[circle] (q) at (12, 4.8) {};
       \node[ circle] (r) at (12, 1.9) {};
       \node[circle] (s) at (13.2, 1.9) {};
 
    \draw[-, dashed] (aa) edge (bb) (aa) edge (ee);

    \draw[-] (aaa) edge (bbb) (aaa) edge (ccc)  (aaa) edge (eee);
 
 \draw[-] (bbb) edge (ccc)  (ddd) edge (eee)  (fff) edge (bbb); 
 \draw[-] (fff) edge (hhh) (fff) edge (iii) (eee) edge (iii) (bbb) edge (hhh) (fff) edge (lll); 
\draw[-] (mmm) edge (ccc) (nnn) edge (ccc) (ppp) edge (ccc) (nnn) edge (ddd) (ooo) edge (ddd) (rrr) edge (ddd) (rrr) edge (eee) (sss) edge (eee) (bbb) edge (qqq) (bbb) edge (mmm); 
   
  \draw[-] (a) edge (b)   (a) edge (e);
 
 \draw[-] (b) edge (c) ; 
 \draw[-] (f) edge (h) (f) edge (i) (e) edge (i) (b) edge (h) (f) edge (l); 
\draw[-] (m) edge (c) (n) edge (c) (p) edge (c) (n) edge (d) (o) edge (d) (r) edge (d)  (s) edge (e) (b) edge (q) (b) edge (m); 
 
 \draw[-, ultra thick] (aa) edge (bb) (aa) edge (cc)  (aa) edge (ee);
 
 \draw[-] (bb) edge (cc)  (dd) edge (ee)  (ff) edge (bb); 
 \draw[-] (ff) edge (hh) (ff) edge (ii) (ee) edge (ii) (bb) edge (hh) (ff) edge (ll); 
\draw[-] (mm) edge (cc) (nn) edge (cc) (pp) edge (cc) (nn) edge (dd) (oo) edge (dd) (rr) edge (dd) (rr) edge (ee) (ss) edge (ee) (bb) edge (qq) (bb) edge (mm);

  \node (v) at (0, 5) {$\pi(X_i)$};
  \node (v) at (12, 5) {$\pi(X_i')$};
  \node (v) at (0, 1.5) {$(X_i)_{i=1}^n \subseteq Z$};
   \node (v) at (12, 1.5) {$(X_i')_{i=1}^n \subseteq Z'$};
  \node (v) at (6, 5) {$D_i | Z_i, R_i, R_i^f \sim \mathcal{P}(Z_i, R_i, R_i^f)$};
  \node (v) at (6, 1.5) {$\Big[(Y_i, Z_i, Z_{\mathcal{N}_i}, D_i, D_{\mathcal{N}_i})R_i, R_i\Big]_{i=1}^n$};
  \node (v) at (0.45, 0.3) {Sample from target pop};
  \node (v) at (11.5, 0.3) {Sample \textit{not} from target pop};


    \end{tikzpicture}
 \caption{Example of the experiment (picture at the center) and policy targeting exercise when the sample is drawn from the target population as in Section \ref{sec:policy} (left-hand side) or the sample is \textit{not} drawn from the target population (right-hand side). Green dots denote treated units, and pink dots denote untreated ones. The experiment runs as described in Section \ref{sec:identification}.  Researchers observe the vector of outcome, treatment, neighbors, treatments, and covariates of sampled units ($(Y_i, Z_i, Z_{\mathcal{N}_i}, D_i, D_{\mathcal{N}_i})R_i$), as well as the the identity of whom they sample ($R_i$). When the experiment participants are drawn from the target population, researchers then design a treatment allocation $\pi(X_i)$ for the entire population using information $X_i$, a subset of  $Z_i$ available to policymakers for all $n$ units. When instead the target population is different from the population from which the sample is drawn, policymakers only observe covariates $(X_i')_{i=1}^n$ from the target sample, and the experiment did not use a sample drawn from the target population. 
  }
\label{fig:network2}
\end{figure}

When the target population differs from the population from which we sample experiment participants, we can only hope to control the expected, but not the conditional regret. When instead the target population is the one from which we sample the experiment participants, we can control both notions of regret as shown in the following lemma. 
\begin{lem}[Expected and conditional regret] \label{lem:expected} Suppose that $(A', Z') = (A, Z)$ almost surely, i.e., for any realization of $(A, Z)$, experiment participants are always drawn from the (larger) target population as in Section \ref{sec:identification}. Then 
$$
\mathcal{R}_{\Pi}^{\mathrm{exp}} \le \mathbb{E}\left[\mathcal{R}_{\Pi, A , Z}^{\mathrm{cond}}\right], 
$$ 
where $\mathcal{R}_{\Pi, A , Z}^{\mathrm{cond}}$ is bounded as in Theorem \ref{thm:thmmain} for $\Pi_n = \Pi$. 
\end{lem} 

Lemma \ref{lem:expected} shows that the regret guarantees in Section \ref{sec:estimation} are valid bounds on the expected (and conditional) regret. The proof of Lemma \ref{lem:expected} follows directly from Jensen's inequality and the law of iterated expectations. The main assumption of Lemma \ref{lem:expected} is that the sampled units are drawn from the (larger) target population, which is the main case of interest in this paper. This is a common feature in applications where researchers sample (small groups of) individuals at random from a large region or country \citep[e.g.,][]{cai2015social, egger2019general}, and are interested in scaling the policy up in such a region or country.

Suppose, however, we are interested in implementing the policy on a population different from the one from which we have drawn our sample (e.g., in a different country). In the following theorem, we study guarantees of the proposed procedure for this setting.

\begin{thm}[Sampled units not drawn from the target population] \label{thm:expected} Suppose that the conditions in Theorem \ref{thm:thmmain} hold, with $(A', Z') \perp \Big[A, Z, (Y_i, R_i, D_i)_{i=1}^n\Big]$. For a universal constant $\bar{C} < \infty$, 
$$
\small 
\begin{aligned} 
\mathcal{R}_{\Pi}^{\mathrm{exp}} \le & \bar{C} \frac{\Gamma \mathbb{E}_A\left[\mathcal{N}_n^{3/2} \log^{1/2}(\mathcal{N}_n)\right]}{\gamma \delta_n} \sqrt{\frac{\mathrm{VC}(\Pi)}{n_e}} + 2 \mathbb{E}_{A, Z}\left[\sup_{\pi \in \Pi} \Big|W_{A,Z}(\pi) - \mathbb{E}_{A', Z'}[W_{A', Z'}(\pi)] \Big|\right], 
\end{aligned}
$$    
where $\mathbb{E}_{A, Z}[\cdot]$ is the expectation operator with respect to the distribution of $(A, Z)$.  
\end{thm}

The proof is in Appendix \ref{proof:expected}. 
Theorem \ref{thm:expected} provides a bound on the expected (instead of conditional) regret, allowing the sampled units to be drawn from a population different from the target population. The bound depends on two components. The first mimics the component in Theorem \ref{thm:thmmain} and depends on the expected maximum degree and the expected size of the sampled population $n_e$. The second component instead captures the discrepancy between the population from which the sample is drawn $(A, Z)$ and the target population. 

Suppose that $(A, Z), (A', Z')$ have the same distribution. It follows 
\begin{equation} \label{eqn:second_component} 
\small 
\begin{aligned} 
& \mathbb{E}\left[\sup_{\pi \in \Pi} \Big|W_{A,Z}(\pi) - \mathbb{E}_{A', Z'}[W_{A', Z'}(\pi)] \Big|\right] \\ &= \mathbb{E}\left[\sup_{\pi \in \Pi} \Big|\frac{1}{n} \sum_{i=1}^n m(\pi(X_i), \sum_k A_{i,k} \pi(X_k), Z_i, |\mathcal{N}_i|) - \mathbb{E}\Big[m(\pi(X_i), \sum_k A_{i,k} \pi(X_k), Z_i, |\mathcal{N}_i|)\Big] \Big|\right], 
\end{aligned} 
\end{equation} 
which is \textit{independent} of the sample size $n_e$.  Equation \eqref{eqn:second_component} depends on how fast the conditional mean functions of \textit{all} units $n$ concentrate around their expectation uniformly over $\Pi$. Equation \eqref{eqn:second_component} captures the expected ``cost" of targeting treatments on a population \textit{different} from the one from which the sample was drawn.


\begin{rem}[Trade-offs of collecting network data] In settings where the target population 
is different from the population from which the sample is drawn, it is possible to obtain \textit{faster} regret bounds if researchers \textit{observe} network data from the entire target population. I show this in Appendix \ref{sec:different_sample}, where regret guarantees do not depend on the additional component $\mathbb{E}_{A, Z}\left[\sup_{\pi \in \Pi} \Big|W_{A,Z}(\pi) - \mathbb{E}_{A', Z'}[W_{A', Z'}(\pi)] \Big|\right]$. 
Therefore, Appendix \ref{sec:different_sample}, together with Theorem \ref{thm:expected}, illustrates trade-offs between collecting and not collecting network data from the target sample when sampled units are not drawn from the target population. \qed 
\end{rem}

\section{Empirical application} \label{sec:application}

I now illustrate the proposed method using data originating from \cite{cai2015social}. The authors study the effect of an information session on farmers' weather insurance adoption. Individuals are grouped into 185 addresses (villages) grouped into approximately 50 larger areas. According to the authors, ``All rice-producing
households were invited to one of the sessions, and almost 90\% of them attended. Consequently, this provided us (the authors) with a census of the population of these 185 villages. In total, 5,335 households were surveyed" \citep{cai2015social}. 
Before conducting the experiment, researchers collected network data by asking each individual to indicate at most five friends (who can be in the same or different village). On average, $50\%$ of the connections of sampled units have a different village. More than  $90\%$ of the connections are within the same area.  

In this application, I use information collected from those units for which information about their post-treatment outcome and their friend's identity is available; in total, 4511, a subset of the population. 
The experiment consists of two rounds of information sessions three days apart, each round containing two types of information sessions (simple and intensive). 
Households are randomized to each round and within each round to each type of information session. By using time variation over the two rounds, \cite{cai2015social} show the existence of significant neighbors' spillover effects of an intensive information session on second-round participants' outcomes and no endogenous spillover effects, consistently with the model presented in this paper. I defer a discussion on how the model and assumptions of this paper connect to \cite{cai2015social} to Section \ref{sec:assumptions_app}.

\vspace{-1.5mm}

\subsection{Experimental setup and estimation} 

In the experiment, ``the effect of social networks on insurance take-up is identified by looking at whether second round participants are more likely to buy insurance if they have more friends who were invited to first round intensive sessions"  \citep{cai2015social}. 
Specifically, each round consists of two sessions held simultaneously. In the first round, households are assigned to either a 20-minute session during which researchers offer details about the insurance contract only (control arm, ``simple" information session) or a 45-minute session that also provides details about the expected benefits of insurance (treatment arm, ``intensive" information session). 
In the second round, farmers are assigned similarly to either intensive or simple information sessions. Treatment denotes whether individuals were assigned to an intensive information session (either in the first or second round), whereas, by design, spillovers occurs from the first to second round, as described in \cite{cai2015social}.\footnote{For estimation, I follow \cite{cai2015social} and consider the general network matrix where spillovers only occur from individuals participating in the first information session to individuals in the second session. When evaluating the out-of-sample performance of the policy, I use the original ``general network" as an adjacency matrix because out-of-sample evaluations may not have the sequential structure of the experiment (i.e., some individuals may be treated and asked to make purchase decisions some time after treatment occurs, possibly generating spillovers also on the treated units participating in the same information session).}  Researchers also considered additional arms where they provided information about purchase decisions of other participants (``More info" in Figure \ref{fig:clusters}). Here, I follow the main analysis in \cite{cai2015social} (Table 2), and focus on providing information on insurance benefits only. 

\begin{figure}
\centering
\begin{tikzpicture}[
node distance = 6mm and 6mm, 
  start chain = going right,
    mw/.style = {minimum width=#1},
  list/.style = {rectangle split, rectangle split parts=1,
                 rectangle split horizontal, draw,
                 align=center,
                 text width=14mm, 
                 minimum height=18mm, 
                 inner sep=5mm, on chain}, 
          > = stealth, 
          ]

  \node[list] (A) at (11, 0) {First round};
  \node[list] (B) at (7, -3) {Simple session };
  \node[list] (C) at (10, -3) {Intensive session };
  
    \node[list] (D) at (17, 0) {Second round};
  \node[list] (E) at (14, -3) {Simple session};
  \node[list] (F) at (17, -3) {Intensive session};
  \node[list] (G) at (20, -3) {More Info};
  

  
  \path[->, -triangle 90] (A) edge node [text width=2.5cm,midway,above] {+ 3 days}  (D);
    \draw[->, -triangle 90] (A) edge (B) (A) edge (C) (D) edge (E) (D) edge (F) (D) edge (G);
    
\end{tikzpicture}
\caption{Design in \cite{cai2015social} with household-level treatment randomization. Participants are assigned at random to first and second rounds, and within each round, to different information sessions. Simple session denotes the control arm, where researchers provided information about the insurance contract only. Intensive session is the main treatment arm, where individuals are also provided with information about the benefits of insurance. ``More info" contains additional arms with information about purchase decisions, omitted in our analysis and \cite{cai2015social}'s main analysis. Purchase decisions were made at the end of each information session. 
 } \label{fig:clusters}
\end{figure}


I follow \cite{cai2015social} in the model specification. I estimate a model using all first-round participants and those second-round participants either in the control arm or in the main (intensive) treatment arm.\footnote{Namely, I follow Column (2)-(5) in Table 2 in \cite{cai2015social}. As discussed in \cite{cai2015social}, I can drop observations in the ``More info" treatment arms for estimating the conditional mean function because individuals in the second-round of information sessions do not generate spillover effects by design. } I estimate $\hat{m}$ using the linear probability model for the outcome as in \cite{cai2015social} (Table 2, Col (4)), controlling for area fixed effects, a large set of covariates, the average number of treated neighbors, individual treatment, and the interaction between individual and neighbors' treatments.
The model in \cite{cai2015social} assumes homogenous treatment effects across covariates and villages. Here, I also allow for some heterogeneity in covariates and control for interaction terms of the rice area, a coefficient capturing risk aversion and education with individual and neighbors' treatments. 
Following \cite{cai2015social}, I consider the ``general network" as the main network, that is, the raw network data obtained from surveys where an individual generates spillover effects on $i$ if she was indicated by $i$ as a friend.
I then construct welfare using a \textit{doubly-robust} estimator, with ten-fold cross-fitting as in Algorithm \ref{alg:adaptive2}. The conditional mean is estimated via lasso with a small penalty ($e^{-12}$) to increase the stability of the estimator. The individual propensity score is estimated as in Remark \ref{rem:prop} via a penalized logistic regression with a similar small penalty and $5\%$ trimming.

\subsection{Policy evaluation} 

I ``simulate" the following environment: researchers collect information from villages in the first fifteen areas. They estimate the policy to treat individuals in the remaining villages. In the remaining villages, I assume the policymaker does not have access to the network information but only observes the farmer's education, risk aversion, and rice area. I then compute welfare effects \textit{out-of-sample} on the villages outside the training set (first 15 areas). I repeat the same process via three-fold cross-fitting: I use the second fifteen areas as a training set and the remaining areas as a test set; similarly, I use the last group of areas as a training set and the first thirty areas as a test set. Finally, I compute the average out-of-sample improvements over the three out-of-sample evaluations.
 The out-of-sample evaluation uses the double-robust score, estimated out-of-sample. This exercise mimics settings where participants are sampled from a random subset of villages, and the treatment assigned to the experiment participants cannot be changed after the experiment (see Remark \ref{rem:ww}). In this exercise, I sample areas instead of villages to guarantee that the welfare estimates are independent of the training set, a desirable property for out-of-sample comparisons. 

I contrast to the empirical welfare-maximization method that ignores welfare effects in \cite{athey2017efficient, KitagawaTetenov_EMCA2018} and uses the \textit{same} policy and models of the proposed procedure for both the propensity score and conditional mean function (including that the conditional mean function controls for spillovers). 

As a first exercise, I consider \textit{simple} policies that use information from transformations of two of the three covariates: education, rice area, and a coefficient capturing risk aversion. I compute simple classification trees obtained for all possible two-out-of-three combinations of such variables. The tree finds one optimal split over the first (continuous) variable. The split for the second variable is constrained to be at the population median value. This policy is simple to compute and communicate because it assigns treatments based on a few possible sub-groups. I study out-of-sample improvements while varying the treatment cost as $1\%, 3\%, 5\%$ of the insurance take-up benefit. These costs are comparable to the direct treatment effect that we would estimate once observations from all villages as in Table 2, Col 2 in \cite{cai2015social} are pooled (approximately equal to $3\%$).  Table \ref{tab:all} provides welfare comparisons. We observe welfare improvements up to approximately thirty percentage points and positive effects uniformly across the specifications. These economically significant improvements are obtained despite the network not being observable in the target sample.

As a second exercise, I consider a more complex policy consisting of a maximum score that controls for education, rice area and risk aversion as follows:
\begin{equation} \label{eqn:ms}
\small 
\begin{aligned} 
\pi(X_i) = 1\Big\{\beta_0 + \text{Rice area} \times \beta_1 + \text{Risk aversion } \times \beta_2 + \text{Education} \times \beta_3 > 0\Big\}. 
\end{aligned}
\end{equation} 
The parameters are estimated using the mixed-integer linear program in Section \ref{sec:optimization}. Table \ref{fig:coefficients} reports the average \textit{out-of-sample} welfare improvement estimated via three-fold cross-fitting. It shows out-of-sample welfare improvements up to nine percentage points. This result illustrates the benefits of the procedure for more complex policy functions as well. 

The cross-fitting procedure returns three policies estimated on independent samples. 
To investigate the properties of the estimated policy, Table \ref{fig:coefficients} reports the coefficients of the estimated policy (NEWM) leading to the largest out-of-sample welfare.  
The policy treats individuals who are more risk-averse, less educated, and with a smaller rice area. I contrast this policy with the one that ignores network effects (EWM). The two policies are substantially different when treating individuals with larger rice areas and risk aversion. This difference highlights the importance of taking into account spillover effects for policy targeting because different subgroups should be treated differently with spillover effects.




\begin{table}[!ht] \centering 
  \caption{\textit{Out-of-sample} welfare improvement for a classification tree upon empirical welfare-maximization targeting rule in \cite{athey2017efficient} that does not account for network effects in the design of the policy. Different columns denote different $X$ variables considered for the design of the policy. Here $C$ denotes the cost of the treatment. The policy is a classification tree that allows for the first covariate to be continuous and finds the best split over the first covariate, whereas the second covariate is whether such a variable is above or below its median value or missing.  } 
  \label{tab:all} 
\begin{tabular}{@{\extracolsep{5pt}} cccc} 
\\[-1.8ex]\hline 
\hline \\[-1.8ex] 
 & Educ \& Rice-ar & Educ \& Risk-av & Rice-ar \& Risk-av \\ 
\hline \\[-1.8ex] 
$C = 1\%$ & 0.146 & 0.084 & 0.289 \\ 
$C = 3\%$ & 0.159 & 0.093 & 0.201 \\ 
$C = 5\%$ & 0.093 & 0.111 & 0.143 \\ 
\hline \\[-1.8ex] 
\end{tabular} 
\end{table}

\begin{table}[!ht] \centering 
  \caption{Estimated coefficients for $\pi(X) = 1\{X^\top \beta + \beta_0 > 0\}$, as a function of the rice area of the farmer, a coefficient capturing risk aversion and education. NEWM denotes the proposed method and EWM the double-robust empirical welfare-maximization procedure that ignores network effects. Coefficients are normalized by $\beta_0$, with estimated $\beta_0 = 1$ for both NEWM and EWM. 
The right-hand-side panel reports the average out-of-sample improvement of the NEWM method over policies that ignore network effects, estimated via three folds cross-fitting. $C$ denotes the cost of treatment. The left-hand-side panel reports the estimated coefficients of the policy with the largest out-of-sample welfare for $C = 5\%$. } 
  \label{fig:coefficients} 
  \scalebox{0.95}{
\begin{tabular}{@{\extracolsep{5pt}} cccc|ccc} 
\\[-1.8ex]\hline 
\hline \\[-1.8ex] 
 &  Rice Area & Risk Aversion & Educ &  & Welfare Improvement  &  \\ 
    &  &  &  &  $C = 1\%$ & $3\%$ & $5\%$   \\ 
\hline \\[-1.8ex] 
 NEWM &  -0.068 & 0.395 & -0.397 &   0.074 & 0.085 & 0.093 \\ 
 EWM &  -0.003 & -0.041 & -0.473 \\
\hline \\[-1.8ex] 
\end{tabular} 
}
\end{table}

\subsection{Assumptions and applicability of the method} \label{sec:assumptions_app}

This section concludes with a review of the assumptions required by the proposed procedure and their applicability in the context of the chosen application. Assumption \ref{ass:sutnva} states that interference occurs through the neighbors' treatment assignments. In the context of our application, treatments denote (intensive) information sessions. This paper assumes potential outcomes are (possibly heterogeneous) functions of the number of informed neighbors. As a result, the model is best suited when information effects, as opposed to endorsement effects (i.e., effects driven by neighbors' purchase decisions), occur.
 This restriction is consistent with findings in \cite{cai2015social}, who, by leveraging the sequential structure of the experiment, illustrate information effects and lack of endorsement effects. Quoting \cite{cai2015social}'s abstract: ``By varying the information
available about peers' decisions and randomizing default options, we show that the
network effect is driven by the diffusion of insurance knowledge rather than the purchase
decisions." Insurance knowledge denotes the treatments, and purchase decisions are the outcomes of interest, consistent with our model.

A second restriction this paper imposes is that the maximum degree is sufficiently smaller than the sample size (Assumption \ref{ass:a2}). This restriction avoids overfitting and controls the complexity of the function class of interest. Following the specification in \cite{cai2015social}, here individuals generate spillovers on those people indicated as friends, at most five of them by the design of the survey in \cite{cai2015social}.
Therefore, we interpret our analysis as imposing a restriction on the exposure mapping $g_n(\cdot)$: only the five ``closest" friends (i.e., friends indicated in the survey) generate spillover effects, whereas if there are other friends not indicated in the survey, these generate no or negligible spillovers. This assumption is mantained in \cite{cai2015social}, who state: ``The drawback of
this specification is that the network characterization may be incomplete.
This concern is mitigated by the experience of the pilot test in two villages,
where most farmers named four or five friends (82\% five, 14\% four, and 4\% others) when the number was not limited."  However, it is important to acknowledge that this is an assumption, and future research should explore the sensitivity of the estimated policy to misspecification of the exposure mapping \citep[e.g.,][]{savje2023causal}.


The model specification of the conditional mean function in \cite{cai2015social} imposes a lack of heterogeneity in unobserved network statistics. However, because we augment the estimated conditional mean with the doubly robust score, the estimators also allow for arbitrary network heterogeneity, even if such heterogeneity is not captured in the estimated conditional mean function. The reader may refer to Lemma \ref{prop:welfare} and Theorem \ref{thm:thmmain} for details. 


Finally, the sampling in \cite{cai2015social} guarantees that the welfare estimated using information from participants is an unbiased estimator of welfare once the policy is deployed \textit{at scale} in rural China.  The main reason is that \cite{cai2015social} independently sample 185 small villages in rural China, and, among such, they randomize treatments at the individual level \citep[see Page 7 in][]{cai2015social}. 
This sampling induces local dependence within small villages, which is possible to accommodate in our framework (see Remark \ref{rem:sampling_local_dependent}).

\vspace{-3mm}

  \section{Conclusions} \label{sec:conclusion} 
  
This paper introduced a method for estimating treatment rules under network interference. It considers constrained environments, and accommodates policy functions that do not necessarily depend on network information. The proposed methodology is valid for a large class of networks and does not impose restrictions on covariates. I cast the optimization problem into a mixed-integer linear program and derive guarantees on the policy regret.

The proposed method assumes anonymous and exogenous interactions. Future research can address the case of endogenous interactions by explicitly modeling the endogenous component, or considering weak dependence structures as in \cite{leung2022causal}.

This paper estimates welfare-maximizing policies when the network information on the target sample is not observed by directly maximizing the empirical welfare. Extending our method by incorporating partial information on the population network is an interesting future direction. Combining the high-dimensional estimator of the network as in \cite{alidaee2020recovering} with the empirical welfare-maximization procedure is a possible approach.  

Finally, the literature on influence maximization has often relied on structural models, whereas the literature on treatment choice has focused on semiparametric estimation.
This paper opens new questions about the trade-off between structural assumptions and model-robust estimation of policy functions. Exploring this trade-off remains an open question.

\begin{appendices} 

\section{Practical guide} \label{app:algorithm}

This section provides details on the implementation. 
Algorithm \ref{alg:all} presents a summary. The method is implemented in the R package {\tt NetworkTargeting} available on the author's website.

  \begin{algorithm} [!h]   \caption{Network Empirical Welfare Maximization}\label{alg:all}
    \begin{algorithmic}[1]
    \State Sample individuals in a (quasi)experiment at random from the population of interest (see Remark \ref{rem:ident} for stratified sampling). 
    \State For each sampled individual $(R_i = 1)$ and their friends ($R_i^f = 1$) in the experiment randomize treatment assignments as in Assumption \ref{ass:quasi} (treatments do not need to be randomized among the remaining units in the population). 
    \State Collect information $\Big[R_i\Big(Y_i, D_i, T_i, N_i, Z_i, Z_{k \in N_i}\Big), R_i\Big]_{i=1}^n$, denoting sampling indicators $(R_i = 1)$, post treatment outcome $Y_i$, treatment assignment $D_i$, neighbors' treatments $T_i$, \textit{arbitrary} individual and neighbors' observable characteristics $Z_i, Z_{k \in N_i}$. 
    \State Run Algorithm \ref{alg:adaptive} to estimate $\hat{m}, \hat{e}$ the conditional mean and propensity scores for sampled units $(R_i = 1)$ as defined in Equation \eqref{eqn:estimand}. 
    \State Run the optimization algorithm in Section \ref{sec:optimization} to estimate $\hat{\pi}$ using (arbitrary) individual level information $X_i \subseteq Z_i$. 
    \State Implement $\hat{\pi}$ on the population of interest by collecting individual-level information $(X_i)_{i=1}^n$ for all units in the population. 
    \end{algorithmic} 
\end{algorithm}

\subsection{Cross-fitting: exact solution} 

The cross-fitting algorithm is described in Algorithm \ref{alg:adaptive}. It solves a \textit{sequence} of mixed-integer linear programs of the form
\begin{equation} \label{eqn:coloring} 
     \small 
     \begin{aligned} 
    (K^*, G^*) = & \mathrm{arg} \min_{K \in \mathbb{Z}, G \in \{0,1\}^{n \times K} }  \quad  K \quad 
     \text{such that } & \sum_{k=1}^K \sum_{j = 1}^n R_i R_j 1\{j \not \in \mathcal{I}_i\} G_{j,k} G_{i,k} = 0  \\ 
    & & \sum_{k=1}^K G_{i,k} = 1, \quad \forall i \in \{1, \cdots, n\}, 
     \end{aligned} 
  \end{equation}
where $\mathcal{I}_i$ is defined in Equation \eqref{eqn:I_i} as the set of sampled units who are \textit{not} friends or share a common friend with $i$. Each program consists of finding a feasible solution to the constraints in Equation \eqref{eqn:coloring} for given $K$. The program finds the smallest number of groups $K^*$ and groups partition $G^*$ such that two \textit{sampled} individuals who are friends or share a common friend are not in the same group.  Here, $G_{i,k}^* = 1$ if $i$ is assigned to group $k$.

To estimate the conditional mean, the algorithm performs cross-fitting with $J$ folds within each group, as in standard cross-fitting algorithms \citep{athey2017efficient}. If some of these groups are small (with fewer than $J \check{P}$ units, for some small finite $\check{P}$), Algorithm \ref{alg:adaptive} does not use information from such groups. Here, $\check{P}$ is a small constant and denotes the minimum number of observations such that the estimator is well-defined (e.g., the effective degrees of freedom for linear regression).\footnote{The presence of groups with a few units does not affect our results in Theorem \ref{thm:dr}, because these results are directly expressed in terms of average convergence rates of the nuisance functions (see Appendix \ref{app:estimated_1}). It also does not affect the characterization of the convergence rate in Remark \ref{rem:jj}, and Appendix \ref{sec:lasso}. Intuitively, because $K^* \le 2\mathcal{N}_n^2$ by \cite{brooks1941colouring}'s theorem, the contribution of groups with few observations to the average estimation error is at most $\mathcal{O}(\mathcal{N}_n^2/n_e)$. See Appendix \ref{sec:lasso} for details. } The propensity score is estimated using a similar approach. To estimate $\hat{e}^{(i)}$, researchers can also use information about the \textit{treatments} of the neighbors of sampled units ($R_i = 1$) who have not been sampled, as described in Algorithm \ref{alg:adaptive}.

To gain further intuition on each step, observe that the proposed partition guarantees that the outcomes of two individuals in the same group are independent conditional on $(A, Z)$. Therefore, within each group, we can then apply a standard cross-fitting algorithm. The construction of such groups and the intuition behind the cross-fitting approach is a novel contribution of this paper. 

  \begin{algorithm} [!h]   \caption{Network Cross-Fitting: Exact Optimization}\label{alg:adaptive}
    \begin{algorithmic}[1]
    \Require $\Big[R_i\Big(Y_i, D_i, T_i, N_i, Z_i, Z_{k \in N_i}\Big), R_i\Big]_{i=1}^n$, finite $\check{P}$, finite $J$. 
    \State For each $i \in \{1, \cdots, n\}$ construct 
     \begin{equation} \label{eqn:I_i}
      \small 
\begin{aligned}      
     \mathcal{I}_i & = \Big\{j \in \{1, \cdots, n\} \setminus \{i\}: R_j = 1 \text{ and } j \not \in N_i, N_i \cap N_j = \emptyset \Big\}.
     \end{aligned} 
     \end{equation} 
     \State Solve Equation \eqref{eqn:coloring} and return $K^*, G^*$. 
    \For{$k \in \{1, \cdots, K^*\}$} 
     \begin{algsubstates}
     \State Partition units $\{i: R_i G_{i,k}^* = 1\}$, to $J$ folds $(F_k^j)_{j=1}^J$, equally sized up-to one element. Define $F_{k}^{j(i)}$ the fold containing unit $i$.
     \State For $i$ such that $G_{i,k}^* R_i = 1$ construct the estimator $\hat{m}^{(i)}(\cdot)$ of $m(\cdot)$, using $(Y_v, D_v, D_{k \in N_v}, Z_v, N_v)$ from units $v$ in $(F_{k}^j)_{j=1}^J \setminus F_{k}^{j(i)}$. Let $\hat{m}^{(i)}(\cdot) = 0$ if $\sum_i G_{i,k}^* R_i \le J\check{P}$. 
     \end{algsubstates}
    \EndFor
    \State Repeat for the propensity score: for $i$ such that $G_{i,k}^* R_i = 1$ estimate the \textit{individual} conditional treatment probabilities using $(D_v, Z_v, R_v, (D_k (1 - R_k), R_k, Z_k)_{k \in N_v})$ from units $v$ in folds $(F_{k}^j)_{j=1}^J \setminus F_{k}^{j(i)}$. Aggregate such probabilities to construct an estimator of $e(\cdot)$ for unit $i$, $\hat{e}^{(i)}(\cdot)$ as in Remark \ref{rem:prop}. Let $1/\hat{e}^{(i)}(\cdot) = 0$ if $\sum_i G_{i,k}^* R_i \le J \check{P}$. 
        \State Define $\hat{m}_i(\pi), \hat{e}_i(\pi)$ as in Equation \eqref{eqn:m_hati} and $W_n(\pi, \hat{m}, \hat{e})$ as in Equation \eqref{eqn:crossfit}. 
\State \Return $W_n(\pi, \hat{m}, \hat{e})$. 
         \end{algorithmic}
\end{algorithm}

\subsection{(Approximate) network cross-fitting with subgraphs}

Algorithm \ref{alg:adaptive2} presents a relaxation of network cross-fitting. 
It fixes $K$, and creates $K$ groups \textit{recursively}. Each iteration, it constructs two groups to \textit{maximize} the number of individuals who are friends or share a common friend and are assigned to the \textit{same} group. It then repeats the same optimization within each group until we obtain $K$ groups in total. The algorithm constructs subgraphs by solving recursively max-cut optimization problems (see Algorithm \ref{alg:recursive}). For each unit $i$, Algorithm \ref{alg:adaptive2} then estimates the conditional mean function using all groups \textit{except} the group assigned to unit $i$. To estimate the propensity score, I construct subgraphs where I maximize the number of individuals who are neighbors (but not necessarily neighbors of neighbors) in each subgraph.\footnote{The reason is that, due to the independence of treatments in Assumption \ref{ass:quasi} (ii), the estimated propensity score is independent of unit $i$'s outcome if it is estimated using information from treatments different from $(D_i, D_{k \in N_i})$.} The slackness parameter $s$ in Algorithm \ref{alg:recursive} guarantees subgraphs have approximately the same number of units up to $s$ units (e.g., five or ten). 

The rationale is the following. If the network presents $K$ completely independent and equally sized clusters,  the algorithm will recover such clusters. In this case, unit $i$'s prediction would use information from clusters except the one containing $i$; the predicted value for unit $i$ would be independent of $i$'s outcome, avoiding overfitting. 
 The algorithm approximates this setup by constructing subgraphs that minimize the number of connections between such subgraphs.\footnote{Although optimization for clusterings with networks goes beyond the scope of this paper, we note that \cite{leung2021network} presents an extensive discussion where clusters are not independent. 
}  I recommend choosing $K$ by leveraging prior knowledge of the data, such as using the number of villages or regions. For example, in the empirical application, units present almost all the connections within same \textit{large areas} with 47 total areas; therefore, any $K \le 47$ (e.g., $K = 10$) guarantees independent subgraphs. Also, note that the effective sample size only shrinks by a factor $(K-1)/K = \mathcal{O}(1)$.


  \begin{algorithm} [!h]   \caption{Network Cross-Fitting: Approximate Optimization}\label{alg:adaptive2}
    \begin{algorithmic}[1]
    \Require $\Big[R_i\Big(Y_i, D_i, T_i, N_i, Z_i, Z_{k \in N_i}\Big), R_i\Big]_{i=1}^n$, slackness parameter $s$, $K$ folds.
    \State Assign individuals into $K$ folds by running Recursive Opt in Algorithm \ref{alg:recursive} with $\tilde{n} = n$, and slackness $s$. 
     \State For $i: R_i = 1$, construct $\hat{m}^{(i)}(\cdot)$, the estimator of $m(\cdot)$ for unit $i$, using data in all except $i$'s fold.
    \State Repeat for the propensity score: run Algorithm \ref{alg:recursive} with $\mathcal{H}_i = \{j \in \{1, \cdots, n\}: j \not \in N_i, R_j + \sum_k A_{j,k} R_k > 0\}$ in lieu of $\mathcal{I}_i$. For each unit $i$, construct $\hat{e}^{(i)}(\cdot)$, the estimator of $e(\cdot)$ for unit $i$ by: (i) estimating individual treatment probabilities with units in all folds except the one containing $i$; (ii) aggregating such probabilites as in Remark \ref{rem:prop}. 
        \State Construct $\hat{e}^{(i)}, \hat{m}^{(i)}$ and $W_n(\pi, \hat{m}, \hat{e})$ as in Equation \eqref{eqn:crossfit}. 
\Return $W_n(\pi, \hat{m}, \hat{e})$. 
         \end{algorithmic}
\end{algorithm}

  \begin{algorithm} [!h]   \caption{Recursive Opt}\label{alg:recursive}
    \begin{algorithmic}[1]
    \Require input size $\tilde{n}$, $(R_i, \mathcal{I}_i)_{i=1}^{\tilde{n}}$, with $\mathcal{I}_i$ as in Equation \eqref{eqn:I_i}, slackness parameter $s$, $K$
    \State Solve 
    $$
    \small 
    \begin{aligned} 
    \vspace{-10mm}
 \hspace{-5mm}  G^* \in \mathrm{arg} \min_{G \in \{0,1\}^{\tilde{n} \times \tilde{n}}} \sum_{i=1}^{\tilde{n}} \sum_{j \neq i}^{\tilde{n}} G_i (1 - G_j) 1\{j \in \mathcal{I}_i\} & R_i R_j  \quad  G_i \in \{0,1\}, i \in \{1, \cdots, \tilde{n}\},  \\ 
  & \frac{1}{n} \sum_{i=1}^n G_i \in \left[\frac{1}{2\tilde{n}}\sum_{i=1}^{\tilde{n}} R_i - s/\tilde{n}, \frac{1}{2\tilde{n}} \sum_{i=1}^{\tilde{n}} R_i + s/\tilde{n}\right].
   \end{aligned} 
    $$ 
    \If{$K = 2$}
    \State 
    \Return $G^*$.
    \Else 
    \State 
    \Return 
    $$
    \small 
\begin{aligned}     
    \hspace{-5mm} \left[G^*, \text{Recursive Opt}\left(\sum_{i=1}^{\tilde{n}} G_i^*, (R_i, \mathcal{I}_i)_{G_i^* = 1}, S', \frac{K}{2}\right),\text{Recursive Opt}\left(\tilde{n} - \sum_{i=1}^{\tilde{n}} G_i^*, (R_i, \mathcal{I}_i)_{G_i^* = 0}, S', \frac{K}{2}\right)\right].
    \end{aligned}
    $$  
    \EndIf 
         \end{algorithmic}
\end{algorithm}

Appendix \ref{app:ext2} contains additional extensions, Appendix \ref{app:numerics} a numerical study , and Appendix \ref{app:derivations_all} derivations. Appendix \ref{app:algorithm} at the end of the main text contains the algorithms. 

\vspace{-2mm}

\onehalfspacing

 \bibliography{my_bib2}
\bibliographystyle{chicago}

 \section{Additional extensions} \label{app:ext2}

\vspace{-1mm}


\subsection{Estimation error of nuisance functions with Algorithm \ref{alg:adaptive}} \label{sec:lasso}

This section examines the estimation error $\sqrt{\mathcal{R}_n(A, Z) \times \mathcal{B}_n(A, Z)}$ in Theorem \ref{thm:dr}. 
Consider estimating $m(\cdot)$ with Algorithm \ref{alg:adaptive}. Algorithm \ref{alg:adaptive} first partitions the units into $K^*$ groups. Within each group, it constructs $J$ equally sized folds. 
For two units $(i,v)$, define $\phi_v^m(i) \in \{0,1\}$ with $\phi_v^m(i) = 1$ if \textit{all} of the following conditions hold unit $v$ is sampled ($R_v = 1$); $v$ is in the \textit{same} partition $k \in \{1, \cdots, K^*\}$ of $i$; and $v$ is in any fold except the one containing unit $i$.\footnote{Following Algorithm \ref{alg:adaptive}'s definitions, $\phi_v^m(i) = 1\{v \in (F_{k}^j)_{j=1}^J \setminus F_{k}^{j(i)}, k \text{ such that } i \in \cup_j F_k^j\}$.} The effective sample size for estimation of $\hat{m}^{(i)}$ is $\sum_{v = 1}^n R_v \phi_v^m(i)$ because, Algorithm \ref{alg:adaptive} uses sampled units not in the same fold of $i$, but in its same partition $k$. Define $\phi_v^e(i) \in \{0,1\}$,  with $\phi_v^e(i) = 1$ if \textit{all} of the following conditions hold: (a) unit $v$ is sampled or, if not sampled, one of its friends is sampled ($R_v = 1$ or $(1 - R_v)R_v^f = 1$); (b) $v$ is in the \textit{same} partition $k \in \{1, \cdots, K^*\}$ of $i$; and (c) $v$ is in any fold except the one containing unit $i$, once we run Algorithm \ref{alg:adaptive} to estimate $e(\cdot)$. Let $m \in \mathcal{M}, e \in \mathcal{E}$, for function classes $\mathcal{M}, \mathcal{E}$, and assume
\begin{equation} \label{eqn:rate} 
\small 
\begin{aligned} 
\mathcal{R}_n(A, Z) = \mathcal{O}\Big(\frac{1}{n} \sum_{i=1}^n C_{\mathcal{M}} \mathbb{E}\Big[\Big(1 + \sum_{v=1}^n R_v \phi_v^m(i)  \Big)^{-2\zeta_m}\Big|R_i = 1, A, Z\Big]\Big) \\
\mathcal{B}_n(A, Z) = \mathcal{O}\Big(\frac{1}{n} \sum_{i=1}^n \frac{1}{\delta_n^2} C_{\mathcal{E}} \mathbb{E}\Big[\Big(1 + \sum_{v=1}^n R_v \phi_v^e(i) \Big)^{-2\zeta_e}\Big| R_i = 1, A, Z\Big]\Big)
\end{aligned} 
\end{equation}  
for some $1/2 \ge \zeta_m, \zeta_e > 0$, and $C_{\mathcal{M}}, C_{\mathcal{E}}$ capturing the complexity of the function class. 
Here, $\zeta_m$ characterizes the convergence rate of the conditional mean function on a sample of \textit{independent} units (by Algorithm \ref{alg:adaptive}), with $\Big(1 + \sum_{v=1}^n R_v \phi_v^m(i)  \Big)$ denoting the effective sample size to estimate $\hat{m}_i$. Similarly, $\zeta_e$ for the propensity score. I rescale the rates for the propensity score by $1/\delta_n^2$ because the propensity score is bounded from zero by $\delta_n$.
Equation \eqref{eqn:rate} also captures the contribution to the estimation error of those units $i$ belonging to groups with a few (finite number of) observations (see Algorithm \ref{alg:adaptive}).\footnote{For those units $i$ with a finite number of observations in their partition $k$, $\sum_{v=1}^n R_v \phi_v^m(i) = \mathcal{O}(1)$, and $\mathcal{O}\left(\mathbb{E}\Big[\Big(1 + \sum_{v=1}^n R_v \phi_v^m(i)  \Big)^{-2\zeta_m}\Big|A, Z, R_i = 1\Big]\right)$ is bounded away from (does not converge to) zero for $i$.}

\begin{prop} \label{prop:rate_estimator} Suppose the conditions in Theorem \ref{thm:dr} and Equation \eqref{eqn:rate} hold, and $n_e = \alpha n, \alpha \in (0,1)$. Then 
$
\sqrt{\mathcal{R}_n(A, Z) \times \mathcal{B}_n(A, Z)} = \mathcal{O}\left(\frac{\mathcal{N}_n^2 C_{\mathcal{M}}^{1/2} C_{\mathcal{E}}^{1/2}}{\delta_n n_e^{\zeta_m + \zeta_e}}\right). 
$
In addition, if \\ 
$\mathcal{N}_n^{1/2} C_{\mathcal{M}}^{1/2} C_{\mathcal{E}}^{1/2}/n_e^{ \zeta_m + \zeta_e} = \mathcal{O}\left(n_e^{-1/2}\right)$, then
$
\mathbb{E}\Big[\sup_{\pi \in \Pi_n} W_{A, Z}(\pi) - W_{A, Z}(\hat{\pi}_{\hat{m}, \hat{e}})\Big| A, Z\Big]= \mathcal{O}\left(n_e^{-\xi}\right).
$ 
\end{prop}
 
See Appendix \ref{app:lasso} for the proof. 
Proposition \ref{prop:rate_estimator} characterizes the rate of the estimation error. Here, $ \mathcal{N}_n^{1/2} C_{\mathcal{M}}^{1/2} C_{\mathcal{E}}^{1/2}/n_e^{\zeta_m + \zeta_e} = \mathcal{O}\left(n_e^{-1/2}\right)$ holds for a large class of estimators under conditions on the maximum degree. An example is lasso. Under fixed sparsity, bounded regression matrix, and regularities in \cite{negahban2012unified}, $\zeta_m =  1/2$, $\mathcal{C}_{\mathcal{M}} = \log(p)$, where $p$ is the dimension of the regression matrix. To attain $\mathcal{N}_n^{1/2} C_{\mathcal{M}}^{1/2} C_{\mathcal{E}}^{1/2}/n_e^{ \zeta_m + \zeta_e} = \mathcal{O}\left(n_e^{-1/2}\right)$, we only need that $\zeta_e$ for the propensity score is such that $\mathcal{N}_n^{1/2} \mathcal{C}_{\mathcal{E}}^{1/2} \log^{1/2}(p)/n_e^{\zeta_e} = \mathcal{O}\left(1\right)$.

\vspace{-1mm}

\subsection{Welfare with spillovers on non-compliance} \label{sec:non_comp}

 Consider the setting where spillovers also occur over individuals' compliance. Namely, let $D_i \in \{0,1\}$ denote the assigned treatment and $S_i 
\in \{0,1\}$ denote the selected treatment from individual $i$. I model non-compliance as follows: 
\begin{equation} \label{eqn:comp} 
\small 
\begin{aligned} 
Y_i = r\Big(S_i, \sum_{k \in N_i} S_k, Z_i, |N_i|, \varepsilon_i \Big), \quad S_i = h_{\theta}\Big(D_i, \sum_{k \in N_i} D_k, Z_i, |N_i|, \nu_i \Big). 
\end{aligned} 
\end{equation} 
I let $\nu_i$ be exogenous unobservables, independent from $\varepsilon_i$ (see Proposition \ref{thm:selection}), and $(r(\cdot), \theta)$ unknown, with $\theta$ denoting the set of parameters indexing $h$. 
Similarly to what discussed in Section \ref{sec:identification}, let 
$
W_{A, Z}(\pi)= \frac{1}{n} \sum_{i=1}^n \mathbb{E}\left[Y_i \Big| A, Z, \Big\{D_i = \pi(X_i)\Big\}_{i=1}^n \right]
$ be the welfare under $\pi$.

\begin{prop}[Identification] \label{thm:selection}
Let Equation \eqref{eqn:comp} hold with 
$
 \varepsilon_i \perp \Big((\nu_j)_{j=1}^n, (\varepsilon_{D_j})_{j=1}^n \Big) \Big| A, Z$, \\ $\nu_i \Big| A, Z, (\varepsilon_{D_j})_{j=1}^n \sim_{i.i.d.} \mathcal{P}_\nu.   
$  
Let $P_\theta(S_i = 1|\cdot)$ denotes the conditional probability of selection into treatment indexed by the parameters $\theta$. 
For each $i \in \{1,\cdots, n\}$,
$$
\small 
\begin{aligned} 
& \mathbb{E}\left[Y_i \Big| A, Z, \Big\{D_i = \pi(X_i)\Big\}_{i=1}^n \right] = \sum_{d \in \{0,1\}, s \in \{0, \cdots, |N_i|\}} \mathbb{E}\Big[Y_i \Big| Z_i, |N_i|, S_i = d, \sum_{k \in N_i} S_k = s\Big] \times H_i(d,s, \pi),  \\ 
& H_i(d,s, \pi)  = P_{\theta}\Big(S_i = d\Big|  Z_i, |N_i|, V_i(\pi)\Big)  \sum_{u_1, \cdots, u_{l}: \sum_v u_v = s} \prod_{k = 1}^{|N_i|}  P_{\theta}\Big(S_{N_i^{(k)}} = u_k \Big| Z_{N_i^{(k)}}, |N_{N_i^{(k)}}|, V_{N_i^{(k)}}(\pi) \Big), 
\end{aligned}   
$$
where $V_i(\pi) = \Big\{D_i = \pi(X_i), \sum_{k \in N_i} D_k = \sum_{k \in N_i} \pi(X_k), Z_i, Z_{k \in N_i}\Big\}$. 
\end{prop}

See Appendix \ref{sec:ee} for the proof. 
Proposition \ref{thm:selection} is an identification result.  The welfare effect of an incentive $\pi$ depends on conditional means and $H_i(\cdot)$. Here $H_i(\cdot)$ denotes the conditional probability of selecting into treatment, conditional on the individual and neighbors' incentives.  Its expression only depends on the individual probability of selected treatments $P_{\theta}(S_i = 1|\cdot)$, conditional on individual's and neighbors' treatment assignments. Interestingly, $H_i(\cdot)$ also depends on the treatment assigned to the second-degree neighbors; therefore, information from second-degree neighbors is required for identification. 
Literature on non compliance includes \cite{kang2016peer}, \cite{vazquez2020causal}.  These references do not study welfare maximization. This motivates a different identification strategy here. 

\vspace{-1mm} 

\subsection{Reweighting with known and different target population}  \label{sec:different_sample}

Here, we study settings where the target population differs from the population from which the sample is drawn \textit{and} the adjacency matrix of the target population is \textit{known}. 

 Consider a population with $n$ individuals, connected under adjacency matrix $A'$ and with covariates matrix $Z'$, and $(A', Z')$ are \textit{observed} by the researcher. Welfare is as in Equation \eqref{eqn:welfare_target}.  
Define $\mathcal{S}_n(A, Z)$ as the empirical support of $Z_i, Z_{k \in N_i}, |N_i|$ for given adjacency matrix $(A, Z)$, and similarly $\mathcal{S}_n(A', Z')$ for $A', Z'$. 
$|\mathcal{S}_n(A, Z)| \le n$ by construction. Define 
$
L(z, \mathbf{x}, l) = \frac{1}{n} \sum_{i=1}^n 1\Big\{Z_i = z, Z_{k \in N_i} = \mathbf{x}, \sum_k A_{i,k} = l \Big\}, L'(z, \mathbf{x}, l)  = \frac{1}{n} \sum_{i=1}^n 1\Big\{Z_i' = z, Z_{k \in N_i'}' = \mathbf{x}, \sum_k A_{i,k}' = l \Big\}, 
$ 
the number of units in each population with individual covariates $z$, neighbors' observables $\mathbf{x}$, and number of friends $l$. 
Estimate the empirical welfare as 
$$
\small 
\begin{aligned} 
\tilde{W}_n(\pi, m^c, e) = \frac{1}{n_e} \sum_{i=1}^n R_i \frac{L'\Big(Z_i, Z_{k \in N_i}, |N_i|\Big)}{L\Big(Z_i, Z_{k \in N_i}, |N_i|\Big)} \left\{ \frac{I_i(\pi)}{e_i(\pi)} \Big(Y_i - m_i^c(\pi)\Big) + m_i^c(\pi)\right\}. 
\end{aligned}
$$  
Here, the empirical welfare reweights observations by the ratio of the empirical distributions in the target population and the sampled units. Importantly, the functions $L(\cdot), L'(\cdot)$ must be observed by the researcher. $L(\cdot)$ is observed under the sampling assumptions in Section \ref{sec:identification}, whereas observing $L'(\cdot)$ assumes that researcher observe $(A', Z')$ from the target population.  

\begin{prop} \label{prop:different2} Suppose the conditions in Theorem \ref{thm:thmmain} hold conditional also on $(A', Z')$, and $\mathcal{S}_n(A',Z') \subseteq \mathcal{S}_n(A, Z)$ almost surely. Let $\hat{\pi}^t \in \mathrm{arg} \max_{\pi \in \Pi_n} \tilde{W}_n(\pi, m^c, e)$. Then, for a universal constant $\bar{C} < \infty$, 
$
\mathbb{E}\Big[\sup_{\pi \in \Pi_n} W_{A', Z'}(\pi) - W_{A',Z'}(\hat{\pi}^t) \Big| A, Z, A', Z'\Big] \le \frac{\bar{C} \Gamma \bar{L}_{A, Z, n} \mathcal{N}_n^{3/2}}{\gamma \delta_n}\sqrt{\frac{\log(\mathcal{N}_n) \mathrm{VC}(\Pi)}{n_e}}, 
$ 
where $\bar{L}_{A, Z, n} =\max_{(Z_i, Z_{k \in N_i}, |N_i|) \in \mathcal{S}_n(A,Z)} L'\Big(Z_i, Z_{k \in N_i}, |N_i|\Big) \Big/ L\Big(Z_i, Z_{k \in N_i}, |N_i|\Big)$. 
\end{prop} 
 
See Appendix \ref{app:main_last}  for a proof. 
Proposition \ref{prop:different2} shows that regret bounds depend on the largest ratio between the empirical distribution on the target and sampled units over the empirical support of the individuals, and neighbors' covariates and of degree. An important assumption is that the support $\mathcal{S}_n(A', Z')$ is contained in the support $\mathcal{S}_n(A, Z)$.

 \subsection{Constraints on $\Pi_n$ that depend on $D$} \label{sec:constrain_pi}
 
 Following Remark \ref{rem:ww}, in this subsection, I discuss a policy-function class 
 \begin{equation} \label{eqn:policy2} 
 \small 
 \begin{aligned} 
 \tilde{\Pi}_n = \Big\{\tilde{\pi}: \mathcal{X} \times \{0,1\} \mapsto \{0,1\}, \tilde{\pi}(x,d) = \pi(x)(1 - d) + d, \pi \in \Pi_n \Big\}, 
 \end{aligned} 
 \end{equation} 
  for $\Pi$ with finite VC dimension. Here $\tilde{\pi}(D_i, X_i)$ is one almost surely if the treatment in the experiment is one $(D_i = 1)$. I define $e, m^c$ as in Equation \eqref{eqn:welf}, here functions of $\tilde{\pi}$.

  \begin{prop} \label{prop:3} Let Assumptions \ref{ass:sutnva}, \ref{ass:ignorability}, \ref{ass:quasi}, \ref{ass:finite_vc}, and \ref{ass:general4} hold. Consider a policy class $\tilde{\pi}(X_i, D_i), \tilde{\pi} \in \tilde{\Pi}_n$, with $\tilde{\pi}_{m^c, e}^* \in \mathrm{arg} \max_{\tilde{\pi} \in \tilde{\Pi}_n} W_n(\tilde{\pi}, m^c, e)$. For a universal constant $\bar{C}<\infty$,
    $$
    \small 
    \begin{aligned} 
\mathbb{E}\Big[\sup_{\pi \in \tilde{\Pi}_n} W_{A, Z}(\pi) - W_{A, Z}(\tilde{\pi}_{m^c, e}^* ) \Big| A, Z\Big] \le \bar{C} \frac{\Gamma \mathcal{N}_n^{3/2}}{\gamma \delta_n} \sqrt{ \frac{\log(\mathcal{N}_n) \mathrm{VC}(\Pi)}{n_e}}.
\end{aligned} 
    $$ 
  \end{prop}

 See Appendix \ref{app:more} for the proof. Proposition \ref{prop:3} extends our results for policies constrained to always assign treatments to the treated individuals in the experiment.



\section{A numerical study} \label{app:numerics}

I simulate data as
$
Y_i = \frac{1}{\max 1, |N_i|} \Big( X_i \beta_1 + X_i \beta_2 D_i + \mu\Big) \sum_{k \in N_i} D_k  + X_i \beta_3 D_i+ \varepsilon_i, \varepsilon_i = \frac{\eta_i + \sum_{k \in N_i} \eta_k}{\sqrt{2 (|N_i| + 1)}},  
$
with $\eta_i \sim_{i.i.d.} \mathcal{N}(0,1)$. I simulate covariates as $X_{i} \in [-1, 1]^4$, with each entry drawn independently and uniformly between $[-1,1]$. I draw $\beta_3   \in \{-1.5, 1.5\}$ with equal probabilities. I consider five versions of NEWM described in the caption of Table \ref{tab:summaries}. 

I compare NEWM to methods that ignore network effects from \cite{KitagawaTetenov_EMCA2018, athey2017efficient}. Each method uses a policy function of the form 
 $\pi(X_i) = 1\Big\{X_{i,1} \phi_1 + X_{i,2} \phi_2 + \phi_3 \ge 0\Big\},$ estimated via MILP. First, I consider a geometric network formation of the form
$
A_{i,j} = 1\Big\{|X_{i,2} - X_{j,2}|/2 + |X_{i,4} - X_{j,4}|/2 \le \sqrt{4/2.75n}\Big\}.
$
 In the second set of simulations, I generate Barabasi-Albert networks. I draw $n/5$ edges uniformly according to Erd\H{o}s-Rényi graph with probabilities $10/n$, and second, I draw sequentially connections of the new nodes to the existing ones with probability equal to the average number of connections of the existing nodes. I simulate over $200$ data sets with $n_e = n$, and evaluate the performance out-of-sample over $1000$ networks, drawn from the same distribution. Results are in Table \ref{tab:summaries}.  For $n$ sufficiently large $(n = 200)$, the five specifications of NEWM yield comparable results. NEWM outperforms methods that ignore spillovers across all specifications.

\begin{table*}[!ht]\centering
\caption{Out-of-sample median \textit{welfare} over $200$ replications. DR is the method in \cite{athey2017efficient} with estimated balancing score and EWM PS is the method in \cite{KitagawaTetenov_EMCA2018} with known balancing score. NEWM\_out1 is NEWM with a correctly specified outcome model, and NEWM\_out2 its equivalent with approximate network cross-fitting. NEWM\_dr1 is the doubly robust equivalent controlling for the number of treated neighbors, and NEWM\_dr2, NEWM\_dr3 control for a binned version of the number of treated neighbors as in Remark \ref{rem:j}, with and without approximate network cross-fitting.  GE denotes the geometric network, and AB the Albert-Barabasi. 
}    \label{tab:summaries} 
\ra{1.3}
\scalebox{0.9}{
\begin{tabular}{@{}lrrrcrrrcrrrcrrrcrrr@{}}\toprule
Welfare & \multicolumn{2}{c}{$n = 50$} & \ & \multicolumn{2}{c}{$n = 70$} & \ & \multicolumn{2}{c}{$n = 100$} & \ & \multicolumn{2}{c}{$n = 150$} & \ & \multicolumn{2}{c}{$n = 200$} \\
\cmidrule{2-3} \cmidrule{5-6}  \cmidrule{8-9}   \cmidrule{11-12}  \cmidrule{14-15} 
&GE& AB& &GE& AB &&GE& AB  &&GE& AB &&GE& AB  \\ \midrule
DR & 1.49 & 0.94 && 1.49 & 1.08 && 1.38 & 1.05 && 1.53& 0.95 && 1.42& 0.95 \\
EWM PS  &1.21 &0.93 && 1.23 &0.92&& 1.32& 0.93 && 1.38& 0.90&& 1.29& 0.95 \\
NEWM\_out1 & 1.74& 1.31 & &1.87 & 1.38 & & 1.93 & 1.37 && 1.91 & 1.40 && 2.00& 1.39 \\
NEWM\_out2 & $1.77$ &1.34& & 1.87 & 1.41 && 1.91 & 1.37 &&1.95 & 1.38&& 1.98& 1.39 \\
NEWM\_dr1 & 1.78 & 1.22 && 1.89 & 1.33 &&1.89 & 1.37 && 1.94 & 1.28 && 1.95& 1.33 \\
NEWM\_dr2 & $1.69$ &1.21& & 1.83 & 1.36 && 1.84 & 1.33 &&1.82 & 1.31&& 1.94& 1.38 \\
NEWM\_dr3 & $1.45$ &1.15& & 1.75 & 1.25 && 1.79 & 1.28 &&1.81 & 1.28&& 1.88& 1.35 \\
\bottomrule
\end{tabular}
}
\end{table*}

\vspace{-7mm}

 \section{Derivations} \label{app:derivations_all}
 
 \vspace{-2mm} 
 
 \subsection{Notation} 
 
\begin{defn}[Proper Cover] \label{defn:cover}
Given an adjacency matrix $A \in \mathcal{A}_n$, with $n$ rows and columns, a family $\mathcal{C}_n = \{\mathcal{C}_n(g)\}$ of  disjoint subsets $\mathcal{C}_n(1), \mathcal{C}_n(2), \cdots$ of $\{1, \cdots, n\}$ is a proper cover of $A$ if $\cup_g \mathcal{C}_n(g) = \{1, \cdots, n\}$ and $\mathcal{C}_n(g) \subseteq \{1, \cdots, n\}$ consists of units such that for any pair of elements $\{i, k \in \mathcal{C}_n(g), k \neq i\}$, $A_{i,k} = 0$.  \qed 
\end{defn} 
\begin{defn}[Chromatic number] \label{defn:chromatic} 
The chromatic number $\chi_n(A)$, denotes the size of the smallest proper cover of $A$. \qed 
\end{defn}

\begin{defn} \label{defn:matrix} 
For a given matrix $A \in \mathcal{A}_n$, I define $A^2 \in \mathcal{A}_n$ the adjacency matrix such that $A_{i,j} = 1$ if $(i,j)$ are either neighbors or they share at least a common neighbor. Similarly $A^M(A)$ is the adjacency matrix obtained after connecting units sharing common neighbors up to $M^{th}$ degree; $N_{i,M}$ is the set of neighbors of individual $i$ for an adjacency matrix $A^M$. \qed 
\end{defn}

The proper cover of $A_n^2$ is defined as $\mathcal{C}_n^2 = \{\mathcal{C}_n^2(g)\}_{g = 1}^{\chi(A^2)}$ with chromatic number $\chi(A_n^2)$. Similarly $\mathcal{C}_n^M = \{\mathcal{C}_n^M(g)\}_{g=1}^{\chi(A^M)}$ with chromatic number $\chi_n(A_n^M)$ is the proper cover of $A_n^M$. For a given set $\mathcal{C}_n^M(g)$, I denote $|\mathcal{C}_n^M(g)|$ the number of elements in such a set.

I will refer to $\chi(A)$ as $\chi_n(A_n)$ whenever clear from the context. Let 
$$
\small 
\begin{aligned} 
e_i^c(\pi) = e^c\Big(\pi(X_i), T_i(\pi), Z_{k \in N_i}, R_{k \in N_i}, Z_i, |N_i|\Big), \quad m_i^c(\pi) = m^c\Big(\pi(X_i), T_i(\pi), Z_i, |N_i|\Big), 
\end{aligned} 
$$ for given functions $e^c, m^c$, and $I_i(\pi) = 1\{T_i(\pi) = T_i, \pi(X_i) = D_i\}$, similarly to Equation \eqref{eqn:I_iii}. In the presence of estimation error, define $\hat{e}_i(\pi), \hat{m}_i(\pi)$ their corresponding estimators.

Following \cite{devroye2013probabilistic}'s notation, for $x_1^n = (x_1, ..., x_n)$ being arbitrary points in $\mathcal{X}^n$, for a function class $\mathcal{F}$, with $f \in \mathcal{F}$,  $f : \mathcal{X} \mapsto \mathbb{R}$, let
$
\mathcal{F}(x_1^n) = \left\{f(x_1), ..., f(x_n): f \in \mathcal{F} \right\}.
$

\begin{defn} For a class of functions $\mathcal{F}$, with $f : \mathcal{X} \mapsto \mathbb{R}$, $\forall f \in \mathcal{F}$ and $n$ data points $x_1, ..., x_n \in \mathcal{X}$ define the $l_q$-covering number
$
\mathcal{M}_q\Big(\eta, \mathcal{F}(x_1^n)\Big)
$
to be the cardinality of the smallest cover $\{s_1, ..., s_N\}$, with $s_j \in \mathbb{R}^n$, such that for each $f \in \mathcal{F}$, there exist an $s_j \in \{s_1, ..., s_N\}$ such that 
$
(\frac{1}{n} \sum_{i=1}^n |f(x_i) - s_j^{(i)}|^q)^{1/q} < \eta. 
$
For $\bar{F}$ the envelope of $\mathcal{F}$, define the Dudley's integral as $\int_0^{2 \bar{F}} \sqrt{\log\Big(\mathcal{M}_1(\eta, \mathcal{F}(x_1^n))\Big)} d\eta$. \qed 
\end{defn}

For random variables $X = (X_1, ..., X_n)$, denote $\mathbb{E}_{X}[.]$ the expectation with respect to $X$, conditional on the  other variables inside the expectation operator. 
\begin{defn} \label{defn:rademacher}
Let $X_1, ..., X_n$ be arbitrary random variables. Let $\sigma = \{\sigma_i\}_{i=1}^n$ be $i.i.d$ Rademacher random variables ($P(\sigma_i = -1) = P(\sigma_i = 1) = 1/2$), independent of $X_1, ..., X_n$. The empirical Rademacher complexity is
$
\mathcal{R}_n(\mathcal{F}) = \mathbb{E}_{\sigma}\Big[\sup_{f \in \mathcal{F}}| \frac{1}{n} \sum_{i=1}^n \sigma_i f(X_i)| \Big| X_1, ..., X_n\Big].
$
\end{defn}

\vspace{-4mm}

\subsection{Theorems} \label{app:regret} 

I discuss the theorems first. Appendix  \ref{sec:lem} presents the lemmas used for these theorems.

The first theorem controls the supremum of the empirical process of interest with respect to $\Pi \supseteq \Pi_n$ as in Assumption \ref{ass:finite_vc}. Theorem \ref{thm:thm2} imposes the same assumptions as Theorem \ref{thm:thmmain}, except that unobservables can be locally dependent up to the $M^{th}$ degree.

\begin{thm} \label{thm:thm2} Let Assumptions \ref{ass:sutnva}, \ref{ass:ignorability} (C),  \ref{ass:quasi}, \ref{ass:finite_vc},  \ref{ass:general4}, \ref{ass:higher_order} (A) hold. Consider functions $m^c(\cdot), e^c(\cdot)$ such that for all $d \in \{0,1\}, t \in \mathcal{T}_n$ $m^c(d, t, Z_i, |N_i|) \in [-\Gamma, \Gamma]$,  for a finite constant $\Gamma$,  and $e^c(d, t, Z_{k \in N_i}, R_{k \in N_i}, Z_i,  |N_i|) \in (\gamma \delta_n, 1 - \gamma \delta_n)$ almost surely. Suppose that either (or both) (i) $e^c = e$, or (ii) also Assumption \ref{ass:ignorability} (A) hold and $m^c = m$. Then for any $n \ge 1, M \ge 2$, and a universal constant $\bar{C}<\infty$
    \begin{equation} \label{eqn:thm1} 
 \small    
    \begin{aligned} 
& \mathbb{E}\Big[\sup_{\pi \in \Pi} |W_n(\pi, m^c, e^c) - W_{A, Z}(\pi)| \Big| A, Z  \Big]  \le \bar{C} \frac{\Gamma}{\gamma \delta_n} \sqrt{\frac{M \mathcal{N}_n^{M + 1} \log(\mathcal{N}_n) \mathrm{VC}(\Pi)}{n_e}}.
\end{aligned} 
    \end{equation}

\end{thm}

\begin{proof}[Proof of Theorem \ref{thm:thm2}] 

I organize the proof as follows. First, I derive a symmetrization argument to bound the supremum of the empirical process in Equation \eqref{eqn:thm1} with the Rademacher complexity of direct and spillover effects. Second, I bound the Rademacher complexity using Lemmas \ref{lem:finallemma}, \ref{lem:final_lemmab}. Section \ref{sec:proof} provides a proof sketch. Define  
$$
\small 
\begin{aligned}  
Q_i(\pi, A, Z) = R_i \left[\frac{I_i(\pi)}{e_i^c(\pi)}\Big(Y_i - m_i^c(\pi)\Big) + m_i^c(\pi)\right],  
\end{aligned} 
$$  
where I suppressed the dependence with $e^c, m^c$. Define $\mathcal{Q}_n(\pi, A, Z)$ the distribution such that $\Big(Q_i(\pi, A, Z)\Big)_{i=1}^n \Big| A, Z \sim \mathcal{Q}_n(\pi, A, Z)$. Define $(\sigma_i)_{i=1}^n$ $i.i.d.$ Rademacher random variables independent of observables and unobservables. Finally, let $\Big(Q_i'(\pi, A, Z)\Big)_{i=1}^n \Big| A, Z \sim \mathcal{Q}_n(\pi, A, Z)$, an independent copy of $\Big(Q_i(\pi, A, Z)\Big)_{i=1}^n$, conditional on $(A, Z)$. Note that $Q_i(\pi, A, Z)$ depends on $\pi$ through $\Big(\pi(X_i), \sum_{k \in N_i} \pi(X_k)\Big)$ by Assumption \ref{ass:sutnva}.

\vspace{-2mm} 

\paragraph{Conditional expectation} 
By definition of $Q_i'$, 
\begin{equation} \label{eqn:helper_main_proof} 
\small 
\begin{aligned} 
\mathbb{E}[W_n(\pi, e^c, m^c)| A,Z] = \frac{1}{n} \sum_{i=1}^n \mathbb{E}[Q_i(\pi, e^c, m^c)| A,Z] =  \frac{1}{n} \sum_{i=1}^n \mathbb{E}[Q_i'(\pi, e^c, m^c)| A,Z].  
\end{aligned} 
\end{equation} 
It follows: 
 \begin{equation} 
 \label{eqn:helper1_app}
 \small 
 \begin{aligned} 
& \mathbb{E}\Big[\sup_{\pi \in \Pi} |W_n(\pi, m^c, e^c) - W_{A, Z}(\pi)| \Big| A, Z  \Big] \quad \\ & =  \mathbb{E}\Big[\sup_{\pi \in \Pi} |W_n(\pi, m^c, e^c) - \mathbb{E}[W_n(\pi, m^c, e^c) | A, Z]| \Big| A, Z  \Big] \quad &(\because \text{Lemma  \ref{lem:doublerobust}})\\ & = \mathbb{E}\Big[\sup_{\pi \in \Pi}\Big|\frac{1}{n_e} \sum_{i=1}^n \Big[Q_i(\pi, A, Z) - \mathbb{E}[Q_i'(\pi, A, Z)| A, Z] \Big]\Big| | A, Z \Big] \quad    &(\because \text{Eq. \eqref{eqn:helper_main_proof}})  \\ 
&= \mathbb{E}\Big[\sup_{\pi \in \Pi}\Big|\frac{1}{n_e} \sum_{i=1}^n \mathbb{E}_{Q'}\Big[Q_i(\pi, A, Z) - Q_i'(\pi, A, Z) \Big| A, Z\Big]\Big| | A, Z \Big] \quad  &(\because  (Q_i')_{i=1}^n \perp (Q_i)_{i=1}^n | A,Z) \\
&\le \mathbb{E}\Big[\sup_{\pi \in \Pi}\Big|\frac{1}{n_e} \sum_{i=1}^n \Big[Q_i(\pi, A, Z) - Q_i'(\pi, A, Z) \Big]\Big| | A, Z \Big] \quad &(\because \text{Jensen's inequality}). 
 \end{aligned} 
  \end{equation} 
  The second to last equality takes the expectation with respect to $Q'$ (given $Q, A, Z$).

\paragraph{Symmetrization and proper cover} 
Recall now Definitions \ref{defn:cover}, \ref{defn:chromatic}, \ref{defn:matrix}. Construct an adjacency matrix $A^M$ with neighbors connected up to the $M^{th}$ degree, with smallest proper cover $\mathcal{C}_n^M = \{\mathcal{C}_n(j)\}_{g=1}^{\chi(A^M)}, \mathcal{C}_n^M(g) \subseteq \{1, \cdots, n\}, \cup_g  \mathcal{C}_n^M(g) = \{1, \cdots, n\}$, and chromatic number $\chi(A^M)$. Note that such a cover always exists.\footnote{For example, in a fully connected network, the chromatic number is $n$, where each group only contains one unit, while in a network with no connection, the chromatic number is one. The size of such cover (chromatic number) will affect the bound in the statement of the theorem via the maximum degree.} By the triangular inequality
 \begin{equation} \label{eqn:helper5_text}
 \small 
 \begin{aligned} 
 & \mathbb{E}\Big[\sup_{\pi \in \Pi}\Big|\frac{1}{n_e} \sum_{i=1}^n \Big[Q_i(\pi, A, Z) - Q_i'(\pi, A, Z)\Big]\Big|  | A, Z\Big]  \\ &\le 
 \sum_{g \in \{1, \cdots, \chi(A^M)\}} \underbrace{\mathbb{E}\Big[\sup_{\pi \in \Pi}\Big|\frac{1}{n_e} \sum_{i \in \mathcal{C}_n^M(g)} \Big[Q_i(\pi, A, Z) - Q_i'(\pi, A, Z)\Big]\Big|  | A, Z\Big]}_{:=II(g)}. 
\end{aligned} 
 \end{equation}  
 
 Observe first that $\mathbb{E}[Q_i(\pi, A, Z) - Q_i'(\pi, A, Z) | A, Z] = 0$ since $Q_i, Q_i'$ have the same distribution. Also, if $R_i = 0$, then $Q_i = 0$. Therefore, by Assumption \ref{ass:sutnva}, and Assumption \ref{ass:quasi} (ii), for a given $\pi$, $Q_i(\pi, A, Z)$ is a deterministic function of 
 $
R_i \Big(R_{k \in N_i}, \varepsilon_{D_i}, \varepsilon_{D_{k \in N_i}}, Z_{k \in N_i}, Z_i, \varepsilon_i, R_{k \in N_i}^f\Big). 
$
Also, note that if $R_i = 1$, then $R_{k}^f = 1$, for $k \in N_i$ almost surely. Therefore, $Q_i$ can be written as a deterministic function of $
\Big(R_i, R_{k \in N_i}, \varepsilon_{D_i}, \varepsilon_{D_{k \in N_i}}, Z_{k \in N_i}, Z_i, \varepsilon_i\Big)
$
only, where we can drop its dependence with $R_{k \in N_i}^f$. 
The following holds. 
\begin{itemize} 
\item By Assumption \ref{ass:quasi} (ii), $\varepsilon_{D_i}$ are $i.i.d.$ and exogenous with respect to $(A,Z, \varepsilon)$; 
\item By Assumption \ref{ass:quasi} (i) $R_i$ are $i.i.d.$ and exogenous;
\item Under Assumption \ref{ass:higher_order} (A), $\varepsilon_i | A, Z$ are independent for individuals who are not neighbors up to degree $M \ge 2$.
\end{itemize} 
As a result, it directly follows that conditional on $A, Z$, for any $M \ge 2$, 
\begin{equation} \label{eqn:yeah}
\small 
\begin{aligned} 
& \Big(R_i, R_{k \in N_i}, \varepsilon_{D_i}, \varepsilon_{D_{k \in N_i}}, Z_{k \in N_i}, Z_i, \varepsilon_i\Big) \perp \Big(R_j, R_{k \in N_j}, \varepsilon_{D_j}, \varepsilon_{D_{k \in N_j}}, Z_{k \in N_j}, Z_j, \varepsilon_j\Big)_{j \not \in \cup_{k=1}^{M} N_{i,k}} | A, Z .   
\end{aligned} 
\end{equation} 
Equation \eqref{eqn:yeah} implies that $Q_i(\pi, A, Z) \perp (Q_j(\pi, A, Z))_{j \not \in \cup_{k=1}^{M} N_{i,k}} | A, Z$. 
Since $(Q_i)_{i=1}^n, (Q_i')_{i=1}^n|A,Z$ have the same \textit{joint} distribution and are independent, we also have 
\begin{equation} \label{eqn:ii}
\small 
\begin{aligned} 
\Big(Q_i(\pi, A, Z) - Q_i'(\pi, A, Z)\Big) \perp \Big(Q_j(\pi, A, Z) - Q_j'(\pi, A, Z)\Big)_{j \not \in \cup_{k=1}^{M} N_{i,k}} | A, Z. 
\end{aligned} 
\end{equation} 
Note that  $(Q_i)_{i \in \mathcal{C}_n^M(g)} =_d  (Q_i')_{i \in \mathcal{C}_n^M(g)} | A,Z$ and are independent (since $\mathcal{C}_n^M$ is deterministic conditional on $A$). Therefore, for each group $\mathcal{C}_n^M(g)$, by Equation \eqref{eqn:ii}, for $i \in \mathcal{C}_n^M(g)$
$$
\small 
\begin{aligned} 
\Big(Q_i(\pi, A, Z) - Q_i'(\pi, A,Z)\Big) \perp 
\Big(Q_j(\pi, A, Z) - Q_j'(\pi, A,Z)\Big)_{j \neq i, j \in \mathcal{C}_n^M(g)} | A, Z .
\end{aligned} 
$$ 
We can then bound $II(g)$ in Equation \eqref{eqn:helper5_text} as follows 
$$
\small 
\begin{aligned} 
II(g)
& = \mathbb{E}\Big[\sup_{\pi \in \Pi}\Big|\frac{1}{n_e} \sum_{i \in \mathcal{C}_n^M(g)} \sigma_i \Big[Q_i(\pi, A, Z) - Q_i'(\pi, A, Z)\Big]\Big|  | A, Z\Big] \\
&\le  \mathbb{E}\Big[\sup_{\pi \in \Pi}\Big|\frac{1}{n_e} \sum_{i \in \mathcal{C}_n^M(g)} \sigma_i Q_i(\pi, A, Z) \Big|  | A, Z\Big] + \mathbb{E}\Big[\sup_{\pi \in \Pi}\Big|\frac{1}{n_e} \sum_{i \in \mathcal{C}_n^M(g)} \sigma_i Q_i'(\pi, A, Z) \Big|  | A, Z\Big] \\
&= 2\mathbb{E}\Big[\sup_{\pi \in \Pi}\Big|\frac{1}{n_e} \sum_{i \in \mathcal{C}_n^M(g)} \sigma_i Q_i(\pi, A, Z) \Big|  | A, Z\Big]. 
\end{aligned} 
$$ 
The first equality follows from independence of $Q_i - Q_i' | A, Z$ within the subset $\mathcal{C}_n^M(g)$, and the fact that $Q_i, Q_i'$ have the same distribution. The second inequality follows from the triangular inequality and $Q_i, Q_i'$ having the same joint distribution given $A, Z$. 
\vspace{-3mm} 
\paragraph{Bound on the Rademacher complexity} The following holds
\begin{equation} \label{eqn:elements_call}
\small 
\begin{aligned} 
& \mathbb{E}\Big[\sup_{\pi \in \Pi}\Big|\frac{1}{n_e} \sum_{i \in \mathcal{C}_n^M(g)} \sigma_i Q_i(\pi, A, Z) \Big|  | A, Z\Big] \le \mathbb{E}\Big[\underbrace{\mathbb{E}_{Y, \sigma}\Big[\sup_{\pi \in \Pi}\Big|\frac{1}{n_e} \sum_{i \in \mathcal{C}_n^M(g)} \sigma_i R_i \frac{I_i(\pi)}{e_i^c(\pi)} Y_i \Big| \Big]}_{:=i(g)} | A, Z\Big] \\ &+ \mathbb{E}\Big[\underbrace{\mathbb{E}_\sigma\Big[\sup_{\pi \in \Pi}\Big|\frac{1}{n_e} \sum_{i \in \mathcal{C}_n^M(g)} \sigma_i R_i \frac{I_i(\pi)}{e_i^c(\pi)} m_i^c(\pi) \Big| \Big]}_{:=ii(g)}  | A, Z\Big]  +
\mathbb{E}\Big[\underbrace{\mathbb{E}_\sigma\Big[\sup_{\pi \in \Pi}\Big|\frac{1}{n_e} \sum_{i \in \mathcal{C}_n^M(g)} \sigma_i R_i m_i^c(\pi) \Big| \Big]}_{:=iii(g)}  | A, Z\Big] , 
\end{aligned} 
\end{equation} 
where $\mathbb{E}_{Y,\sigma}[\cdot]$ denotes the conditional expectation with respect to $(Y, \sigma)$ only, given all other observables and unobservables, and similarly $\mathbb{E}_\sigma[\cdot]$, with respect to $\sigma$ only. 
Let  $\bar{C} < \infty$ be a universal constant. 
I invoke Lemma \ref{lem:final_lemmab} for each element in Equation \eqref{eqn:elements_call} as follows. 
\begin{itemize}
\item I invoke Lemma \ref{lem:final_lemmab} for $i(g)$ with $Y_i$ in lieu of $\Omega_i$ in the statement of Lemma \ref{lem:final_lemmab}, with third moment bounded by $\Gamma^2$ by Assumption \ref{ass:ignorability} (C); and $\frac{I_i(\pi)}{e_i^c(\pi)}$ in lieu of $g_i(\cdot)$ in Lemma \ref{lem:final_lemmab}, with upper bound $U_n = 1/(\gamma \delta_n)$ ($U_n$ as in the statement of Lemma \ref{lem:final_lemmab}) by Assumption \ref{ass:quasi} (iii). Since we sum over elements $R_i 1\{i \in \mathcal{C}_n^M(g)\} = 1$, by Lemma \ref{lem:final_lemmab} 
$$
\small 
\begin{aligned} 
i(g) \le \bar{C} \frac{\Gamma}{n_e \gamma \delta_n} \sqrt{\mathrm{VC}(\Pi) \mathcal{N}_n \sum_{i=1}^n R_i 1\{i \in \mathcal{C}_n^M(g)\} \log(\mathcal{N}_n)}.
\end{aligned} 
$$ 
\item I invoke Lemma \ref{lem:final_lemmab} for $ii(g)$ where we have $\frac{I_i(\pi)}{e_i^c(\pi)} m_i(\pi)$ in lieu of $g_i(\cdot)$ in the statement of Lemma \ref{lem:final_lemmab}, with constant $U_n = \Gamma/(\gamma \delta_n)$ by Assumption \ref{ass:ignorability} (C) and Assumption \ref{ass:quasi} (iii), and $\Omega_i = 1$ in the statement of Lemma \ref{lem:final_lemmab}. Therefore, 
$$
\small 
\begin{aligned} 
ii(g) \le \bar{C} \frac{\Gamma}{n_e \gamma \delta_n} \sqrt{\mathrm{VC}(\Pi) \mathcal{N}_n \sum_{i=1}^n R_i 1\{i \in \mathcal{C}_n^M(g)\} \log(\mathcal{N}_n)}.
\end{aligned} 
$$ 
\item I invoke Lemma \ref{lem:final_lemmab} for $iii(g)$ where we have $m_i(\pi)$ in lieu of $g_i(\cdot)$ with constant $U_n = \Gamma$, and $\Omega_i = 1$ in the statement of Lemma \ref{lem:final_lemmab}. Therefore,  
$$
\small 
\begin{aligned} 
iii(g) \le \bar{C} \frac{\Gamma}{n_e} \sqrt{\mathrm{VC}(\Pi) \mathcal{N}_n \sum_{i=1}^n R_i 1\{i \in \mathcal{C}_n^M(g)\} \log(\mathcal{N}_n)}. 
\end{aligned} 
$$ 
\end{itemize} 
\vspace{-6mm}
\paragraph{Summing the terms} Collecting the terms together, I obtain
$$
\small 
\begin{aligned} 
\eqref{eqn:helper5_text} \le \sum_{g \in \{1, \cdots, \chi(A^M)\}} \mathbb{E}\Big[\frac{\Gamma}{n_e \gamma \delta_n} \sqrt{\mathcal{N}_n \log(\mathcal{N}_n) \mathrm{VC}(\Pi) \sum_{i=1}^n R_i 1\{i \in \mathcal{C}_n^M(g)\}}\Big| A, Z\Big], 
\end{aligned} 
$$ 
 where the expectation is taken with respect to $R = (R_1, \cdots, R_n)$. I write 
\begin{equation} \label{eqn:aii2} 
 \small 
 \begin{aligned} 
 & \sum_{g \in \{1, \cdots, \chi(A^M)\}} \mathbb{E}\Big[\frac{\Gamma}{n_e \gamma \delta_n} \sqrt{\mathcal{N}_n \log(\mathcal{N}_n) \mathrm{VC}(\Pi) \sum_{i=1}^n R_i 1\{i \in \mathcal{C}_n^M(g)\}}\Big| A, Z\Big] \\ &\le 
 \sum_{g \in \{1, \cdots, \chi(A^M)\}} \frac{\Gamma}{n_e \gamma \delta_n} \sqrt{\mathcal{N}_n \log(\mathcal{N}_n) \mathrm{VC}(\Pi) \sum_{i=1}^n \mathbb{E}[R_i|A,Z] 1\{i \in \mathcal{C}_n^M(g)\}} \quad &(\because \text{ Jensen's inequality}) \\
 &=\sum_{g \in \{1, \cdots, \chi(A^M)\}} \frac{\Gamma}{n_e \gamma \delta_n} \sqrt{\mathcal{N}_n \log(\mathcal{N}_n) \mathrm{VC}(\Pi) n_e |\mathcal{C}_n^M(g)|/n} \quad &(\because \mathbb{E}[R_i |A,Z] = n_e/n).  \end{aligned} 
 \end{equation} 
 We have 
 \begin{equation} \label{eqn:aii3} 
 \small 
 \begin{aligned} 
 \eqref{eqn:aii2}
 &\le \chi(A^M)  \frac{\Gamma}{n_e \gamma \delta_n} \sqrt{\mathcal{N}_n \log(\mathcal{N}_n) \mathrm{VC}(\Pi) n_e \frac{1}{\chi(A^M)} \sum_{g \in \{1, \cdots, \chi(A^M)\}} |\mathcal{C}_n^M(g)|/n}  \quad (\because \text{concave } \sqrt{x}) \\
 &= \chi(A^M)  \frac{\Gamma}{n_e \gamma \delta_n} \sqrt{\mathcal{N}_n \log(\mathcal{N}_n) \mathrm{VC}(\Pi) n_e \frac{1}{\chi(A^M)} } =   \frac{\Gamma}{\gamma \delta_n} \sqrt{\frac{\chi(A^M) \mathcal{N}_n \log(\mathcal{N}_n) \mathrm{VC}(\Pi) }{n_e}}. 
 \end{aligned} 
 \end{equation} 
In the first inequality  in \eqref{eqn:aii3} I divided and multiplied by $\chi(A^M)$ and used concavity of the square-root function. In the second equality I used the fact that $\{\mathcal{C}_n^M(g)\}$ contain disjoint sets, with $\sum_g |\mathcal{C}_n^M(g)|= n$. By Lemma \ref{lem:degreebound} $\chi(A^M) \le M \mathcal{N}_n^M$, completing the proof. 
 \end{proof}

 \vspace{-4mm} 
 
	\subsubsection{Theorem \ref{thm:thmmain} and Theorem \ref{thm:estimatedoutcome}} \label{app:1_first}
	
	I state these two theorems as corollaries of Theorem \ref{thm:thm2}. 
\begin{cor} \label{cor:1} Theorem \ref{thm:thmmain} holds.  
\end{cor} 
 \begin{proof} Following \cite{KitagawaTetenov_EMCA2018}, 
\begin{equation} \label{eqn:kk1}
 \small
 \begin{aligned} 
 & \mathbb{E}\Big[\sup_{\pi \in \Pi_n} W_{A,Z}(\pi) - W_{A,Z}(\hat{\pi}_{m^c,e})  \Big| A, Z\Big] \\ &= \mathbb{E}\Big[\sup_{\pi \in \Pi_n} W_{A,Z}(\pi) - W_n(\hat{\pi}_{m^c,e}, m^c, e) +  W_n(\hat{\pi}_{m^c,e}, m^c, e) -  W_{A,Z}(\hat{\pi}_{m^c,e})  \Big|A, Z\Big] \\ &\le 
 \mathbb{E}\Big[\sup_{\pi \in \Pi_n} W_{A,Z}(\pi) - W_n(\pi, m^c, e) +  W_n(\hat{\pi}_{m^c,e}, m^c, e) -  W_{A,Z}(\hat{\pi}_{m^c,e^c})  \Big|A, Z\Big].   \end{aligned} 
\end{equation} 
 We have $\eqref{eqn:kk1} \le \mathbb{E}\Big[2\sup_{\pi \in \Pi_n} |W_{A,Z}(\pi) - W_n(\pi, m^c, e)|  \Big|A, Z\Big]  \le \mathbb{E}\Big[2\sup_{\pi \in \Pi} |W_{A,Z}(\pi) - W_n(\pi, m^c, e)|  \Big|A, Z\Big]$ $(\because \Pi_n \subseteq \Pi)$. The proof completes by Theorem \ref{thm:thm2}, with $M = 2$. 
 \end{proof} 

\begin{cor} Theorem \ref{thm:estimatedoutcome} holds.  
\end{cor} 
 \begin{proof} \renewcommand{\qedsymbol}{}Following the argument of Corollary \ref{cor:1}, and using the fact that $\Pi_n \subseteq \Pi$, it follows 
 	$$
 	\small
 	\begin{aligned}
  \mathbb{E}\Big[\sup_{\pi \in \Pi_n} W_{A,Z}(\pi) - W_{A,Z}(\hat{\pi}_{\hat{m}, \hat{e}})\Big| A, Z \Big] &\le 2\underbrace{\mathbb{E}\Big[\sup_{\pi \in \Pi} |W_n(\pi, m^c, e^c) - W_{A,Z}(\pi)|\Big| A, Z\Big]}_{(I)}  \\ &\quad + 2 \underbrace{\mathbb{E}\Big[\sup_{\pi \in \Pi} |W_n(\pi,\hat{m}, \hat{e})- W_n(\pi,m^c, e^c)|\Big| A, Z\Big]}_{(II)}  .
 	\end{aligned} 
 	$$
 	Term $(I)$ is bounded by Theorem \ref{thm:thm2}. 
 	 I now study $(II)$. In particular, $(II)$ is equal to
\begin{equation} \label{eqn:kjh}
\small
 	\begin{aligned} 
 	&\mathbb{E}\Big[\sup_{\pi \in \Pi}\Big|\frac{1}{n_e} \sum_{i=1}^n R_i \frac{I_i(\pi)}{\hat{e}_i(\pi)}\Big(Y_i - \hat{m}_i(\pi) \Big) + \frac{1}{n_e} \sum_{i=1}^n R_i \Big(\hat{m}_i(\pi) -  m_i^c(\pi)\Big) - 
 	\frac{1}{n_e} \sum_{i=1}^n R_i \frac{I_i(\pi)}{e_i^c(\pi)}\Big(Y_i - m_i^c(\pi)\Big)\Big| | A, Z\Big] \\ 
 	&\le \mathbb{E}\Big[\sup_{\pi \in \Pi}\Big|\frac{1}{n_e} \sum_{i=1}^n R_i \frac{I_i(\pi)}{\hat{e}_i(\pi)}\Big(Y_i + m_i^c(\pi) - m_i^c(\pi) - \hat{m}_i(\pi) \Big) - 
 	\frac{1}{n_e} \sum_{i=1}^n R_i \frac{I_i(\pi)}{e_i^c(\pi)}\Big(Y_i - m_i^c(\pi)\Big)\Big| | A, Z\Big] \\ 
 	& \quad + \mathbb{E}\Big[\sup_{\pi \in \Pi} |\frac{1}{n_e} \sum_{i=1}^n R_i \Big(\hat{m}_i(\pi) -  m_i^c(\pi)\Big)| \Big| A, Z\Big] \\ 
 	&\le  \mathbb{E}\Big[\sup_{\pi \in \Pi} |\frac{1}{n_e} \sum_{i=1}^n R_i \Big(\hat{m}_i(\pi) -  m_i^c(\pi)\Big)| \Big| A, Z\Big] + \mathbb{E}\Big[\sup_{\pi \in \Pi} |\frac{1}{n_e} \sum_{i=1}^n R_i (\frac{I_i(\pi)}{e_i^c(\pi)} - \frac{I_i(\pi)}{\hat{e}_i(\pi)}) \Big(Y_i -  m_i^c(\pi)\Big) |  \Big| A, Z\Big] \\ 
&  \quad	+ \mathbb{E}\Big[\sup_{\pi \in \Pi} |\frac{1}{n_e} \sum_{i=1}^n R_i \frac{I_i(\pi)}{\hat{e}_i(\pi)} \Big(\hat{m}_i(\pi) -  m_i^c(\pi)\Big)| \Big| A, Z\Big].
 	\end{aligned} 
 	\end{equation} 
 	I inspect each term in Equation \eqref{eqn:kjh}. Since $R_i \in \{0,1\}$
 	$$
 	\small 
 	\begin{aligned} 
 	\mathbb{E}\Big[\sup_{\pi \in \Pi} |\frac{1}{n_e} \sum_{i=1}^nR_i \Big(\hat{m}_i(\pi) -  m_i^c(\pi)\Big)| \Big| A, Z\Big] \le \mathbb{E}\Big[ \frac{1}{n_e} \sum_{i=1}^n \sup_{d, s} R_i |\hat{m}(d,s,Z_i, |N_i|) -  m^c(d,s,Z_i, |N_i|)| \Big| A, Z\Big]. 
 	\end{aligned} 
 	$$ 
 	By Cauchy-Schwarz inequality and the triangular inequality 
 	$$
 	\small 
 	\begin{aligned} 
 	&  \mathbb{E}\Big[ \frac{1}{n_e} \sum_{i=1}^n \sup_{d, s} R_i |\hat{m}(d,s,Z_i, |N_i|) -  m^c(d,s,Z_i, |N_i|)| \Big| A, Z\Big] \\ &\le  \sqrt{\frac{1}{n_e} \sum_{i=1}^n \mathbb{E}[R_i^2]} \sqrt{\mathbb{E}\Big[ \frac{1}{n_e} \sum_{i=1}^n \sup_{d, s}  |\hat{m}(d,s,Z_i, |N_i|) -  m^c(d,s,Z_i, |N_i|)| ^2\Big| A, Z\Big]} \\
 	 &= \sqrt{\mathbb{E}\Big[ \frac{1}{n} \sum_{i=1}^n \sup_{d, s}  |\hat{m}(d,s,Z_i, |N_i|) -  m^c(d,s,Z_i, |N_i|)| ^2\Big| A, Z\Big]} \quad (\because \mathbb{E}[R_i^2] = \mathbb{E}[R_i] = n_e/n).
 	 \end{aligned} 
 	$$ 
 	 For the second term we have (let $e_i^c(d,t) = e^c(d,t,Z_{k \in N_i}, R_{k \in N_i},  |N_i|)$ and similarly for $\hat{e}_i(d,t)$)
 	$$
 	\small 
 	\begin{aligned} 
 	& \mathbb{E}\Big[\sup_{\pi \in \Pi} |\frac{1}{n_e} \sum_{i=1}^n R_i (\frac{I_i(\pi)}{e_i^c(\pi)} - \frac{I_i(\pi)}{\hat{e}_i(\pi)}) \Big(Y_i -  m_i^c(\pi)\Big) |  \Big| A, Z\Big] \le 2 \Gamma' \mathbb{E}\Big[\sup_{\pi \in \Pi} \frac{1}{n_e} \sum_{i=1}^n R_i |(\frac{I_i(\pi)}{e_i^c(\pi)} - \frac{I_i(\pi)}{\hat{e}_i(\pi)})  |  \Big| A, Z\Big] \\
 	&\le 2 \Gamma' \mathbb{E}\Big[ \frac{1}{n_e} \sum_{i=1}^n R_i \sup_{d,t} |(\frac{1}{e_i^c(d,t)} - \frac{1}{\hat{e}_i(d, t)})  |  \Big| A, Z\Big] \le 2 \Gamma' \sqrt{\mathbb{E}\Big[ \frac{1}{n} \sum_{i=1}^n  \sup_{d,t} |(\frac{1}{e_i^c(d,t)} - \frac{1}{\hat{e}_i(d, t)})  |^2  \Big| A, Z\Big]}  
 	\end{aligned} 
 	$$ 
 	where in the first inequality I used the fact that $Y_i, m^c$ are uniformly bounded and in the last inequality I used Cauchy-Schwarz.  
 For the third term in \eqref{eqn:kjh}, it follows similarly
 	$$
 	\small 
 	\begin{aligned} 
 	& \mathbb{E}\Big[\sup_{\pi \in \Pi} |\frac{1}{n_e} \sum_{i=1}^n R_i \frac{I_i(\pi)}{\hat{e}_i(\pi)} \Big(\hat{m}_i(\pi) -  m_i^c(\pi)\Big)| \Big| A, Z\Big]  \le \frac{1}{\gamma \delta_n} \mathbb{E}\Big[\sup_{\pi \in \Pi} \frac{1}{n_e} \sum_{i=1}^n R_i | \Big(\hat{m}_i(\pi) -  m_i^c(\pi)\Big)| \Big| A, Z\Big] \\
 	&\le \frac{1}{\gamma \delta_n} \sqrt{\mathbb{E}\Big[ \frac{1}{n} \sum_{i=1}^n \sup_{d,t} | \Big(\hat{m}(d,t, Z_i, |N_i|) -  m^c(d,t, Z_i, |N_i|)\Big)|^2 \Big| A, Z\Big]}.  
 	\end{aligned} 
 	$$  
 	 \end{proof}
 
 \vspace{-10mm} 
 
 \subsubsection{Proof of Theorem \ref{prop:minimax}} \label{proof:lower}

The proof constructs an appropriate adjacency matrix, matrix of covariates and distribution of treatments and unobservables to provide the lower bound, taking into account the selection indicators. Recall the definition of $\mathbb{E}_{\mathcal{D}_n(A,Z)}[\cdot]$ in Theorem \ref{prop:minimax}. Let $v = \mathrm{VC}(\Pi)$, and recall, under Assumption \ref{ass:quasi} (i), $R_i \sim_{i.i.d.} \mathrm{Bern}(\alpha), \alpha = n_e/n$. Let $X_i = Z_i$ for expositional convenience not to keep track of both $X_i, Z_i$. Let  $A^* \in \mathcal{A}_n^o$, such that $A_{i,j}^* = 0$ for all $i \neq j$. Let $z_1, \cdots, z_v$ be $v$ points shattered by $\Pi$, which, since $\mathcal{X} = \mathbb{R}^d$ and $\Pi$ has VC dimension $v$ they must exist. Let $Z^*$ such that $\frac{1}{n} \sum_{i=1}^n 1\{Z_i^* = z_j\} = \frac{1}{v}$ for all $j \in \{1, \cdots, v\}$. I write  
 \begin{equation} \label{eqn:lower1} 
 \small 
 \begin{aligned} 
 & \sup_{A \in \mathcal{A}_n^o, Z \in \mathcal{Z}^n}  \sup_{\mathcal{D}_n(A,Z) \in \mathcal{P}_n(A,Z)} \frac{\delta_n}{\mathcal{N}_n^{3/2} \log^{1/2}(\mathcal{N}_n) } \mathbb{E}_{\mathcal{D}_n(A,Z)}\Big[ \Big(\sup_{\pi \in \Pi} W_{A, Z}(\pi) -  W_{A, Z}(\hat{\pi}_n)\Big)\Big| A, Z \Big] \\ &\ge 
  \sup_{\mathcal{D}_n(A^*,Z^*) \in \mathcal{P}_n(A^*, Z^*)} \frac{\delta_n}{\mathcal{N}_n^{3/2} \log^{1/2}(\mathcal{N}_n) } \mathbb{E}_{\mathcal{D}_n(A^*,Z^*)}\Big[ \Big(\sup_{\pi \in \Pi} W_{A^*, Z^*}(\pi) - W_{A^*, Z^*}(\hat{\pi}_n)\Big)\Big| A = A^*, Z = Z^* \Big], 
 \end{aligned} 
 \end{equation} 
 where, recall that $\delta_n, \mathcal{N}_n$ are also a function of $A^*, Z^*$. 
 
I will focus on Equation \eqref{eqn:lower1}. I will indicate for $| A^*, Z^*$ the conditioning set $| A = A^*, Z = Z^*$.  Because I consider a fully disconnected network, we have $\delta_n = 1$ in Assumption \ref{ass:quasi} (since individuals have no neighbors), and $\mathcal{N}_n = 2$ for adjacency matrix $A^*$. I follow the proof of Theorem 14.5 in \cite{devroye2013probabilistic}, and Theorem 2.2 in \cite{KitagawaTetenov_EMCA2018}, while I also condition on $(A^*,Z^*)$, and consider random indicators $R_i$.

\paragraph{Treatment assignments and potential outcomes' distribution} Next, I select the distribution for treatment assignments and potential outcomes. 
Let $D_i$ be a Bernoulli random variable, independent of observables and unobservables with $P(D_i = 1) = 1/2$. Let $\mathbf{b} \in \{0,1\}^v$ be a bit indicator which indexes a distribution $\mathcal{D}_{n, \mathbf{b}}(A^*,Z^*) \in \mathcal{P}_n(A^*, Z^*)$. Namely, I restrict the class of distributions to a finite number of distributions, indexed by $\mathbf{b}$. Denote $Y_i(d) = r(d, 0, Z_i, 0, \varepsilon_i)$, the potential outcome function, where spillovers and number of connections are equal to zero by construction of $A^*$. Let $P(Y_i(1) = 1/2 | Z_i = z_j) = 1/2 + \eta$, $P(Y_i(1) = -1/2 | Z_i = z_j) = 1/2 - \eta$ for $\mathbf{b}_j = 1, j \le v$. If $\mathbf{b}_j = 0$, instead have $P(Y_i(1) = 1/2 | Z_i = z_j) = 1/2 - \eta$, $P(Y_i(1) = -1/2 | Z_i = z_j) = 1/2 + \eta$, where $\eta \in [0,1/2]$ and is selected at the end of the proof. Consider $Y_i(0) = 0$ almost surely. 
 
 \vspace{-2mm} 
\paragraph{Lower bound via Bayes risk} I can therefore write the optimal treatment rule as $\pi_{\mathbf{b}}^*(z_j) = 1\{b_j = 1\}, j \le v$, which satisfies the finite VC dimension. I have $W_{A^*,Z^*}(\pi_{\mathbf{b}}^*) = \frac{\eta}{v} \sum_{j=1}^v \mathbf{b}_j$ under the distribution $\mathcal{D}_{n, \mathbf{b}}$. Consider $\mathbf{b}$ being a random variable with $\mathbf{b}_j \sim_{i.i.d.} \mathrm{Bern}(1/2)$ and independent of observables and unobservables. Denote $\mathbb{E}_{\mathbf{b}}[\cdot]$ the expectation with respect to $\mathbf{b}$ (conditional on $A^*, Z^*$). 
 For any data-dependent $\hat{\pi}_n$,\footnote{See e.g., Appendix A.2 in \cite{KitagawaTetenov_EMCA2018}, Page 8.}
\begin{equation} \label{eqn:lower2} 
 \small 
 \begin{aligned} 
 &  \sup_{\mathcal{D}_n(A^*,Z^*) \in \mathcal{P}_n} \mathbb{E}_{\mathcal{D}_n(A^*, Z^*)}\Big[W_{A^*, Z^*}(\pi_{\mathbf{b}}^*) - W_{A^*,Z^*}(\hat{\pi}_n)\Big| A^*, Z^*\Big] \\
 &\ge \mathbb{E}_{\mathbf{b}} 
 \Big[\mathbb{E}_{\mathcal{D}_{n, \mathbf{b}}(A^*, Z^*)}\Big[W_{A^*, Z^*}(\pi_{\mathbf{b}}^*) - W_{A^*,Z^*}(\hat{\pi}_n)\Big| A^*, Z^*\Big] \Big| A^*, Z^*\Big],  \\
 &\ge \inf_{\hat{\pi}_n} \eta \frac{1}{v} \sum_{j=1}^v \mathbb{E}_{\mathbf{b}}\Big[\mathbb{E}_{\mathcal{D}_{n, \mathbf{b}}(A^*, Z^*)}\Big[1\{b_j \neq \hat{\pi}_n(z_j)\}\Big| A^*, Z^*\Big] \Big| A^*, Z^*\Big]. 
 \end{aligned} 
\end{equation} 
 We can see the minimization in Equation \eqref{eqn:lower2} as a risk-minimization problem with lower bound provided by the Bayes risk. I construct a Bayes classifier of the form 
 $$
 \small 
 \begin{aligned} 
 \hat{\pi}_n(z_j) = 1\left\{P\left(\mathbf{b}_j = 1 | \Big[(Y_i, D_i, D_{k \in N_i})R_i, R_i\Big]_{i=1}^n, A^*, Z^*\right) \ge 1/2\right\}, j \le v.
 \end{aligned} 
 $$ 
I can then follow the same steps of \cite{KitagawaTetenov_EMCA2018}, Equation (A.12), (A.13), with $k_j^+ = \# \Big\{i: Z_i = z_j, R_i Y_i D_i = 1/2\Big\}, k_j^- = \# \Big\{i: Z_i = z_j, R_i Y_i D_i = -1/2\Big\}$ for the case of this paper, and $Y_i D_i R_i$ in lieu of $Y_i D_i$ in the derivation of \cite{KitagawaTetenov_EMCA2018}. Following (A.12), (A.13), and the equation below (A.13) in \cite{KitagawaTetenov_EMCA2018}
  $$
 \small 
 \begin{aligned} 
&  \inf_{\hat{\pi}_n} \eta \frac{1}{v} \sum_{j=1}^v \mathbb{E}_{\mathbf{b}}\Big[\mathbb{E}_{\mathcal{D}_{n, \mathbf{b}}(A^*, Z^*)}\Big[1\{b_j \neq \hat{\pi}_n(z_j)\}\Big| A^*, Z^*\Big] \Big| A^*, Z^*\Big] \\ & \ge \frac{\eta}{2 v} \sum_{j=1}^v a^{-\mathbb{E}_{\mathbf{b}}\left[\mathbb{E}_{\mathcal{D}_{n, \mathbf{b}}(A^*, Z^*)}\Big[|\sum_{i: Z_i^* = z_j} 2 Y_i D_i R_i| \Big| A^*, Z^*\Big] \right]},\quad a = \frac{1 + 2 \eta}{1 - 2\eta}. 
 \end{aligned} 
 $$ 
 \paragraph{Lower bound on the Bayes risk} The marginal distribution of $Y_i(1)$ (once  we integrate over $\mathbf{b}$), is $P(Y_i(1) = 1/2 | Z^*, A^*) = P(Y_i(1) = -1/2| Z^*, A^*) = 1/2$ similarly to \cite{KitagawaTetenov_EMCA2018}. By independence, $P(D_i R_i = 1) = \alpha/2$. We have  
\begin{equation} \label{eqn:bbbbb}
 \small 
 \begin{aligned} 
\mathbb{E}_{\mathbf{b}}\left[\mathbb{E}_{\mathcal{D}_{n, \mathbf{b}}(A^*, Z^*)}\Big[|\sum_{i: Z_i^* = z_j} 2 Y_i D_i R_i| \Big| A^*, Z^*\Big] \right] & = \mathbb{E}_{\mathbf{b}}\left[\mathbb{E}_{\mathcal{D}_{n, \mathbf{b}}(A^*, Z^*)}\Big[|\sum_{i: Z_i^* = z_j, R_i D_i = 1} 2 Y_i| \Big| A^*, Z^*\Big] \right] \\ &= \sum_{k=0}^{n/v} {n/v \choose k} (\frac{\alpha}{2})^k (1 - \frac{\alpha}{2})^{n/v - k}\mathbb{E}\Big| B(k, \frac{1}{2})- k/2\Big|, 
 \end{aligned} 
\end{equation} 
 where $B(k,1/2)$ is a binomial random variable with parameters $(k , 1/2)$.  Equation \eqref{eqn:bbbbb} holds because given $Z = Z^*$, there are $n/v$ many observations with $Z_i^* = z_j, j \le v$ by construction of $Z^*$. 
We can write  
 $
 \mathbb{E}\Big| B(k, \frac{1}{2})- k/2\Big| \le \sqrt{\mathbb{E}\Big(B(k, \frac{1}{2}) - k/2\Big)^2} = \sqrt{\frac{k}{4}}. 
$ 
 It follows 
 $$
 \small
 \begin{aligned} 
 \eqref{eqn:bbbbb} & \le  \sum_{k=0}^{n/v} {n/v \choose k} (\frac{\alpha}{2})^k (1 - \frac{\alpha}{2})^{n/v - k} \sqrt{\frac{k}{4}} = \mathbb{E} \sqrt{\frac{B(n/v, \frac{\alpha}{2})}{4}} \le \sqrt{\frac{\mathbb{E}[B(n/v, \frac{\alpha}{2})]}{4}} = \sqrt{\frac{n \alpha}{v 8}}. 
\end{aligned}  
 $$ 
 Following \cite{KitagawaTetenov_EMCA2018}, equation (A.14) and below, with $\alpha n$ in lieu of $n$ in \cite{KitagawaTetenov_EMCA2018},
 it follows that the Bayes risk is bounded from below by 
 $
 \frac{1}{2} \sqrt{\frac{v}{\alpha n}} \exp(-2 \sqrt{2})
 $
 for $\alpha n \ge 16 v$. Since $n_e = \alpha n, \mathcal{N}_n \le 2$ for $A^*$, the proof completes.  

\vspace{-2mm}

\subsubsection{Proof of Theorem \ref{thm:dr} } \label{app:estimated_1}

For the sake of brevity, I will be using the following notation 
 	$$
 	\small
\begin{aligned}  	
  & \tilde{I}_i(d,t)  = 1\Big\{d = D_i,t = T_i\Big\}, \quad \tilde{e}_i(d,t)  = e\Big(d, t, Z_{k \in N_i}, R_{k \in N_i}, Z_i, |N_i| \Big), \quad  \tilde{m}_i(d,t) & = m\Big(d, t, Z_i, |N_i|\Big). 
 	\end{aligned} 
$$ 
Also, let $\tilde{\varepsilon}_i = Y_i - m(D_i, T_i, Z_i, |N_i|)$. 
 	 With an abuse of notation, I will refer to $\hat{e}_i(d,t), \hat{m}_i(d,t)$ as the estimated counterpart of $\tilde{e}_i(d,t), \tilde{m}_i(d,t)$ from Algorithm \ref{alg:adaptive}, with arguments $(d, t)$. 
 Let	$I_i(\pi), e_i(\pi), m_i(\pi)$ be defined as in Equation \eqref{eqn:I_iii}, and the beginning of Section \ref{sec:prop_score}, and $\hat{e}_i(\pi), \hat{m}_i(\pi)$ be defined as in Algorithm \ref{alg:adaptive} (Equation \eqref{eqn:m_hati}), as a function of the treatment assignment rule $\pi$ (therefore $\hat{e}_i(\pi) := \hat{e}_i(\pi(X_i),T_i(\pi))$ and similarly for $\hat{m}_i(\pi)$). Recall the definitions of $K^*, F_k^j$ in Algorithm \ref{alg:adaptive}: $K^*$ denotes the number of partitions obtained under Algorithm \ref{alg:adaptive}, where we have $k \in \{1,\cdots, K^*\}$ many partitions. Within each partition, we have $j \in \{1, \cdots, J\}$ folds $F_k^j$. For each $k \in \{1, \cdots, K^*\}$, $\cup_{j=1}^J F_k^j$ never contains two units that are either neighbors or share a common neighbor. Let $R = (R_1, \cdots, R_n)$. 
 
The argument I present in the current proof applies to any $K^*$ obtained from Algorithm \ref{alg:adaptive}, and any configurations of folds $(F_k^j)_{j=1}^J, k \in \{1, \cdots, K^*\}$ obtained from Algorithm \ref{alg:adaptive}, including settings with folds $F_k^j$ with one or few units.\footnote{Algorithm \ref{alg:adaptive} estimates $\hat{m}^{(i)}, 1/\hat{e}^{(i)}$ as zero functions for those units $i$, assigned to groups $k \in \{1, \cdots, K^*\}$ with few (a finite) number of units. The estimation error for such units contributes directly to the average error in Equation \eqref{eqn:estimation_error}.  Appendix \ref{sec:lasso} show how to control the estimation error in \eqref{eqn:estimation_error}.}

 \paragraph{Preliminary decomposition}	
 Following the same argument of Corollary \ref{cor:1}, since $\Pi_n \subseteq \Pi$, 
 	$$
 	\small
 	\begin{aligned}
  \mathbb{E}\Big[\sup_{\pi \in \Pi_n} W_{A,Z}(\pi) - W_{A,Z}(\hat{\pi}_{\hat{m}, \hat{e}})\Big| A, Z \Big] &\le 2\underbrace{\mathbb{E}\Big[\sup_{\pi \in \Pi} |W_n(\pi, m, e) - W_{A,Z}(\pi)|\Big| A, Z\Big]}_{(I)}  \\ &\quad + 2 \underbrace{\mathbb{E}\Big[\sup_{\pi \in \Pi} |W_n(\pi,\hat{m}, \hat{e})- W_n(\pi,m, e)|\Big| A, Z\Big]}_{(II)}  .
 	\end{aligned} 
 	$$
 	Term $(I)$ is bounded by Theorem \ref{thm:thm2}. I now study $(II)$. 
$$
\small
\begin{aligned} 
& (II) = \mathbb{E}\Big[\sup_{\pi \in \Pi} \Big|\frac{1}{n_e} \sum_{i=1}^n R_i \Big(\frac{I_i(\pi)}{\hat{e}_i(\pi)}(m_i(\pi) - \hat{m}_i(\pi)) + \tilde{\varepsilon}_i \frac{I_i(\pi)}{\hat{e}_i(\pi)} + \hat{m}_i(\pi) - m_i(\pi)\Big)\Big|| A, Z	\Big] \\ 
&= \mathbb{E}\Big[\sup_{\pi \in \Pi} \Big|\frac{1}{n_e} \sum_{i=1}^n R_i \Big(\frac{I_i(\pi)}{\hat{e}_i(\pi)} - \frac{I_i(\pi)}{e_i(\pi)} \Big)(m_i(\pi) - \hat{m}_i(\pi)) + R_i \tilde{\varepsilon}_i \frac{I_i(\pi)}{\hat{e}_i(\pi)} - R_i \Big( \frac{I_i(\pi)}{e_i(\pi)} - 1\Big)(\hat{m}_i(\pi) - m_i(\pi)) \Big|  | A, Z\Big]. 
\end{aligned}  
$$
The last equality follows after adding and subctracting  $R_i \frac{I_i(\pi)}{e_i(\pi)} (m_i(\pi) - \hat{m}_i(\pi))$. It follows
\begin{equation} 
\small
\begin{aligned}
(II) 
\le &\underbrace{\mathbb{E}\Big[\sup_{\pi \in \Pi} \Big|\frac{1}{n_e} \sum_{i=1}^n R_i \Big(\frac{I_i(\pi)}{\hat{e}_i(\pi)} - \frac{I_i(\pi)}{e_i(\pi)} \Big)(m_i(\pi) - \hat{m}_i(\pi))\Big| | A, Z \Big]}_{(i)} + \underbrace{\mathbb{E}\Big[\sup_{\pi \in \Pi}\Big| \frac{1}{n_e} \sum_{i=1}^n  R_i \tilde{\varepsilon}_i \Big(\frac{I_i(\pi)}{\hat{e}_i(\pi)} -   \frac{I_i(\pi)}{e_i(\pi)}\Big) \Big| | A, Z \Big]}_{(ii)} \\ &+ \underbrace{\mathbb{E}\Big[\sup_{\pi \in \Pi}\Big| \frac{1}{n_e}  \sum_{i=1}^n R_i \tilde{\varepsilon}_i  \frac{I_i(\pi)}{e_i(\pi)} \Big| | A, Z\Big]}_{(iii)} + \underbrace{\mathbb{E}\Big[\sup_{\pi \in \Pi} \Big|\frac{1}{n_e} \sum_{i=1}^n R_i  \Big( \frac{I_i(\pi)}{e_i(\pi)} - 1\Big)(\hat{m}_i(\pi) - m_i(\pi)) \Big| | A, Z\Big]}_{(iv)}. 
\end{aligned} 
\end{equation} 
\vspace{-15mm}
\paragraph{Bounding $(i)$} Consider $(i)$ first. We have 
\begin{equation} \label{eqn:estimation_error}
\small
\begin{aligned} 
(i) & = \mathbb{E}\Big[\sup_{\pi \in \Pi} \Big|\frac{1}{n_e} \sum_{i=1}^n R_i \Big(\frac{I_i(\pi)}{\hat{e}_i(\pi)} - \frac{I_i(\pi)}{e_i(\pi)} \Big) R_i(m_i(\pi) - \hat{m}_i(\pi))\Big| | A, Z \Big]  \quad (\because R_i \in \{0,1\})
\\ 
&\le  \sqrt{ \frac{1}{n_e} \mathbb{E}\Big[\sum_{i=1}^n R_i \sup_{d,t} \Big(\frac{1}{\tilde{e}_i(d,t)} - \frac{1}{\hat{e}_i(d,t)}\Big)^2\Big| A, Z\Big] }\sqrt{ \frac{1}{n_e} \mathbb{E}\Big[ \sum_{i=1}^n R_i \sup_{d,t} \Big(\tilde{m}_i(d,t) - \hat{m}_i(d,t)\Big)^2 \Big| A, Z\Big] }  \\ &= \sqrt{\mathbb{E}[R_i/n_e] \mathbb{E}\Big[\sum_{i=1}^n \sup_{d,t} \Big(\frac{1}{\tilde{e}_i(d,t)} - \frac{1}{\hat{e}_i(d,t)}\Big)^2\Big| R_i = 1, A, Z \Big] } \quad \quad (\because \text{Defn of conditional expectation})  \\ &\quad \times \sqrt{ \mathbb{E}[R_i/n_e] \mathbb{E}\Big[ \sum_{i=1}^n \sup_{d,t} \Big(\tilde{m}_i(d,t) - \hat{m}_i(d,t)\Big)^2\Big| R_i = 1, A, Z\Big] } = \sqrt{\mathcal{R}_n(A,Z) \times \mathcal{B}_n(A,Z)}. 
\end{aligned} 
\end{equation}  

\vspace{-6mm}
\paragraph{Summands in $(ii)$ and $(iii)$, $(iv)$} Next, I show that each summand in $(ii), (iii), (iv)$ has a zero conditional expectation, given $R , A, Z$, for any $\hat{e}^{(i)}, \hat{m}^{(i)}$ in Algorithm \ref{alg:adaptive}.  
\begin{itemize} 
\item[$(ii)$]
I start from summands in $(ii)$. I write the expectation of each summand as
\begin{equation} \label{eqn:help}
\small 
\begin{aligned} 
& \mathbb{E}\Big[R_i \tilde{\varepsilon}_i \Big(\frac{I_i(\pi)}{\hat{e}_i(\pi)} - \frac{I_i(\pi)}{e_i(\pi)}\Big) \Big| R, A, Z\Big] \\ & =  \mathbb{E}\Big[R_i \Big(r(\pi(X_i), T_i(\pi), Z_i, |N_i|, \varepsilon_i) - m_i(\pi)\Big)  \Big(\frac{I_i(\pi)}{\hat{e}_i(\pi)} - \frac{I_i(\pi)}{e_i(\pi)}\Big) \Big| R, Z, A\Big]  \\ 
&= \mathbb{E}\Big[\mathbb{E}\Big[R_i \Big(r(\pi(X_i), T_i(\pi), Z_i, |N_i|, \varepsilon_i) - m_i(\pi)\Big)  \Big(\frac{I_i(\pi)}{\hat{e}_i(\pi)} - \frac{I_i(\pi)}{e_i(\pi)}\Big) \Big| \hat{e}_i(\pi), R, Z, A\Big] \Big|R, Z, A \Big]  \\
&= R_i \underbrace{\mathbb{E}\Big[\Big(r(\pi(X_i), T_i(\pi), Z_i, |N_i|, \varepsilon_i) - m_i(\pi)\Big) | A, Z, R\Big]}_{=0}  \mathbb{E}\Big[\Big(\frac{I_i(\pi)}{\hat{e}_i(\pi)} - \frac{I_i(\pi)}{e_i(\pi)}\Big) \Big|R, Z, A \Big]\\ & (\because \text{ Alg \ref{alg:adaptive} and Assumptions } \ref{ass:quasi} (i, ii))
\quad \quad = 0. 
\end{aligned} 
\end{equation}
The last equality follows from the fact that $T_i(\pi)$ (in Equation \eqref{eqn:welfare}) is a deterministic function of $(A, Z)$, $\varepsilon_i$ is independent of $\hat{e}_i(\pi)$ given $(R, Z, A)$ by Algorithm \ref{alg:adaptive}, and $\varepsilon_i$ is conditionally independent of $(D_i, R_i)_{i=1}^n$ given $A, Z$, by Assumption \ref{ass:quasi} (i, ii). 
\item[$(iii)$] For $(iii)$, $\mathbb{E}[R_i \tilde{\varepsilon}_i I_i(\pi)/e_i(\pi) | R,A,Z] = 0$ directly by Assumptions \ref{ass:quasi} (i, ii). 
\item[$(iv)$] For summands in $(iv)$, we have: 
\begin{equation} \label{eqn:help22}
\small 
\begin{aligned} 
& \mathbb{E}\Big[R_i  \Big( \frac{I_i(\pi)}{e_i(\pi)} - 1\Big)(\hat{m}_i(\pi) - m_i(\pi)) \Big| R, A, Z\Big] \\ &= 
R_i \underbrace{\mathbb{E}\Big[  \Big( \frac{I_i(\pi)}{e_i(\pi)} - 1\Big)\Big| R, A, Z\Big]}_{=0} \mathbb{E}\Big[(\hat{m}_i(\pi) - m_i(\pi)) \Big| R, A, Z\Big] = 0. 
\end{aligned} 
\end{equation}
The first equality follows because $\hat{m}_i(\pi)$ is independent of $(D_i, D_{k \in N_i})$ conditional on $(R, A, Z)$ by Algorithm \ref{alg:adaptive} and Assumption \ref{ass:quasi} (ii). 
 \end{itemize} 

 \vspace{-3mm} 
 \paragraph{Bounds for $(ii)$} Using the triangular inequality and the law of iterated expectations, I write (letting $\hat{e}_i(\cdot)$ be the estimated propensity score function for $i$)
 \begin{equation} \label{eqn:summand_m}
 \small 
 \begin{aligned} 
  (ii) \le & \mathbb{E}\Big[\sum_{k=1}^{K^*} \sum_{j =1}^J \underbrace{\mathbb{E}\Big[\sup_{\pi \in \Pi} \Big| \frac{1}{n_e} \sum_{i \in F_k^j} R_i \tilde{\varepsilon}_i\Big(\frac{I_i(\pi)}{\hat{e}_i(\pi)} - \frac{I_i(\pi)}{e_i(\pi)}\Big) \Big| | \hat{e}_{i \in F_k^j}(\cdot), R, A, Z \Big]}_{:= (M_k^j)} \Big| A, Z\Big],
 \end{aligned} 
 \end{equation}  
 where here we also condition on $R$ and the estimated functions $\hat{e}_i$ for units in the fold $i \in F_k^j$. 
 Next, we bound each component $(M_k^j)$ in \eqref{eqn:summand_m}. 
We make the following observations. 
 \begin{itemize} 
  \item[(1)] $(F_k^j)_{j=1}^J, K^*$ are deterministic functions of $(R, A)$ by construction of Algorithm \ref{alg:adaptive}.
 \item[(2)] For each $i \in F_{k}^j$, $\mathbb{E}\Big[R_i \tilde{\varepsilon}_i\Big(\frac{I_i(\pi)}{\hat{e}_i(\pi)} - \frac{I_i(\pi)}{e_i(\pi)}\Big) \Big| A, R, Z, \hat{e}_{i \in F_{k}^j}(\cdot)\Big] = 0$ by \eqref{eqn:help} and independence of $\hat{e}_{i \in F_k^j}(\cdot)$ with $\tilde{\varepsilon}_i$ (independence follows from Alg \ref{alg:adaptive} and Assumptions \ref{ass:quasi} (i,ii)).\footnote{Independence follows from the fact that  $\cup_{j=1}^J F_k^j$ does not contain two sampled individuals that are either neighbors or share a common neighbor. Therefore, we never use information from $(D_i, D_{k \in N_i})$ to estimate $\hat{e}_i(\cdot)$ for all $i: R_i = 1$. Also, note that the argument holds if, for estimating the propensity score for $i$, we also use information from the neighbors of the units in $\cup_{j=1}^J F_k^j \setminus F_k^{j(i)}$ which have \textit{not} been sampled, where $F_k^{j(i)}$ denotes the fold containing $i$. These units  (i.e., non-sampled neighbors of elements in $\cup_{j=1}^J F_k^j \setminus F_k^{j(i)}$) cannot be neighbors of $i$ (with $R_i = 1$) since  $\cup_{j=1}^J F_k^j$ does not contain sampled units with a common neighbor. } 
 \item[(3)] Conditional on $(\hat{e}_{i \in F_k^j}(\cdot), R, A, Z)$, we have that $\left\{R_i \tilde{\varepsilon}_i\Big(\frac{I_i(\pi)}{\hat{e}_i(\pi)}(\cdot) - \frac{I_i(\pi)}{e_i(\pi)}\Big)\right\}$  are mutually independent among units in the same fold ($i \in F_k^j$), by \ref{ass:quasi} (i,ii), and Alg \ref{alg:adaptive}. 
 \end{itemize} 
Therefore, by (2), and (3) above I can invoke standard symmetrization arguments for centered independent random variables \citep[see Lemma 6.4.2 in][]{vershynin2018high} to bound
 \begin{equation} \label{eqn:aiii}
 \small 
 \begin{aligned} 
 (M_k^j) \le 2 \mathbb{E}\Big[\mathbb{E}_{\tilde{\varepsilon}, \sigma}\Big[\sup_{\pi \in \Pi} \Big| \frac{1}{n_e} \sum_{i \in F_k^j} \sigma_i R_i \tilde{\varepsilon}_i\Big(\frac{I_i(\pi)}{\hat{e}_i(\pi)} - \frac{I_i(\pi)}{e_i(\pi)}\Big) \Big|\Big] | \hat{e}_{i \in F_k^j}(\cdot), R, A, Z \Big]
 \end{aligned} 
 \end{equation}  
 for $(\sigma_1, \cdots, \sigma_n)$ be $i.i.d.$ exogenous Radamacher random variables (recall that $\mathbb{E}_{\tilde{\varepsilon}, \sigma}[\cdot]$ indicates that the inner expectation is conditional on everything else except $\sigma, \tilde{\varepsilon}$). 
 
I can now directly use Lemma \ref{lem:final_lemmab} to bound the right-hand-side of \eqref{eqn:aiii}. Namely, I invoke Lemma \ref{lem:final_lemmab} where $\Omega_i$ in the statement of Lemma \ref{lem:final_lemmab} is $\tilde{\varepsilon}_i$ in Equation \eqref{eqn:aiii}, $g_i(\cdot)$ in Lemma \ref{lem:final_lemmab} is $\Big(\frac{I_i(\pi)}{\hat{e}_i(\pi)} - \frac{I_i(\pi)}{e_i(\pi)}\Big)$ in Equation \eqref{eqn:aiii}; $U_n$ in the statement of Lemma \ref{lem:final_lemmab} is $\frac{2}{\gamma \delta_n}$ in \eqref{eqn:aiii}. Therefore, by Lemma \ref{lem:final_lemmab}, for a universal constant $\bar{C} < \infty$
$$
\small 
\begin{aligned} 
& \mathbb{E}_{\tilde{\varepsilon}, \sigma}\Big[\sup_{\pi \in \Pi} \Big| \frac{1}{n_e} \sum_{i \in F_k^j} \sigma_i R_i \tilde{\varepsilon}_i\Big(\frac{I_i(\pi)}{\hat{e}_i(\pi)} - \frac{I_i(\pi)}{e_i(\pi)}\Big)\Big|\Big] \le \frac{\bar{C} \Gamma}{n_e} \sqrt{\mathcal{N}_n \log(\mathcal{N}_n) \sum_{i=1}^n R_i 1\{i \in F_k^j\} \mathrm{VC}(\Pi)}.  
\end{aligned} 
 $$  
 It follows 
\begin{equation} \label{eqn:hgfr}
 \small 
 \begin{aligned} 
 \sum_{j=1}^J \mathbb{E}\Big[\sum_k^{K^*} (M_k^j)\Big| A, Z\Big] & \le  J \mathbb{E}\Big[K^* \frac{\bar{C} \Gamma}{n_e} \sqrt{\frac{\sum_{j=1}^J \sum_{k = 1}^{K^*} \mathcal{N}_n \log(\mathcal{N}_n) \sum_{i=1}^n R_i 1\{i \in F_k^j\} \mathrm{VC}(\Pi)}{JK^*}}\Big| A, Z\Big] \\  & \quad \quad (\because \text{concavity of } \sqrt{x})\\
 &\le  \mathbb{E}\Big[\sqrt{J K^*} \frac{\bar{C} \Gamma}{n_e} \sqrt{\mathcal{N}_n \log(\mathcal{N}_n) \sum_{i=1}^n R_i \mathrm{VC}(\Pi)}\Big| A, Z\Big] \quad (\because \cup_{k=1, j=1}^{K^*, J} F_k^j \subseteq \{1, \cdots, n\}) \\
  & \le \mathbb{E}\Big[\sqrt{J \chi(A^2)} \frac{\bar{C} \Gamma}{n_e} \sqrt{\mathcal{N}_n \log(\mathcal{N}_n) \sum_{i=1}^n R_i \mathrm{VC}(\Pi)}\Big| A, Z\Big] \quad (\because K^* \le \chi(A^2)  \text{ by Lem } \ref{lem:K})\\
 &\le \sqrt{J \chi(A^2)} \frac{\bar{C} \Gamma}{n_e} \sqrt{\mathcal{N}_n \log(\mathcal{N}_n) \sum_{i=1}^n \mathbb{E}[R_i] \mathrm{VC}(\Pi)} \quad (\because \text{Jensen's inequality}). 
 \end{aligned} 
 \end{equation} 
 By Assumption \ref{ass:quasi} (i) $\eqref{eqn:hgfr} \le    \sqrt{J \chi(A^2)}\bar{C}  \Gamma  \sqrt{\frac{\mathcal{N}_n \log(\mathcal{N}_n)  \mathrm{VC}(\Pi)}{n_e}}.$  
By construction of Algorithm \ref{alg:adaptive}, $J = \mathcal{O}(1)$. By Lemma \ref{lem:boundnumber}, $\chi(A^2) \le 2\mathcal{N}_n^2$.

 \vspace{-3mm} 
 
\paragraph{Rademacher complexity bounds for $(iii)$}   Since $(iii)$ does not depend on estimators, the bound for $(iii)$ follows from the same argument in Theorem \ref{thm:thm2}. Recall the definitions of $\chi(A^2), \mathcal{C}_n^2(g)$ I used in Theorem \ref{thm:thm2}. Following the proof of Theorem \ref{thm:thm2} (Paragraph ``Symmetrization and proper cover"), I can write 
$$
\small 
\begin{aligned} 
(iii) \le \sum_{g \in \{1, \cdots, \chi(A^2)\}} \mathbb{E}\Big[\mathbb{E}_{\sigma, \tilde{\varepsilon}}\Big[ \sup_{\pi \in \Pi} \Big|\frac{1}{n_e} \sum_{i \in \mathcal{C}_n^2(g)} R_i \tilde{\varepsilon}_i \frac{I_i(\pi)}{e_i(\pi)}\Big| \Big] | A, Z\Big]. 
\end{aligned} 
$$ 
I can now bound $\mathbb{E}_{\sigma, \tilde{\varepsilon}}\Big[ \sup_{\pi \in \Pi} \Big|\frac{1}{n_e} \sum_{i \in \mathcal{C}_n^2(g)} R_i \tilde{\varepsilon}_i \frac{I_i(\pi)}{e_i(\pi)}\Big| \Big]$ directly with Lemma \ref{lem:final_lemmab}, with $\tilde{\varepsilon}_i$ in lieu of $\Omega_i$ in Lemma \ref{lem:final_lemmab} and $I_i(\pi)/e_i(\pi)$ in lieu of $g_i(\cdot)$ in Lemma \ref{lem:final_lemmab}, with upper bound $U_n = 2/(\gamma \delta_n)$. Following the same argument as in Equation \eqref{eqn:aii2}
$$
\small 
\begin{aligned} 
\sum_{g \in \{1, \cdots, \chi(A^2)\}} \mathbb{E}\Big[\mathbb{E}_{\sigma, \tilde{\varepsilon}}\Big[ \sup_{\pi \in \Pi} \Big|\frac{1}{n_e} \sum_{i \in \mathcal{C}_n^2(g)} R_i \tilde{\varepsilon}_i \frac{I_i(\pi)}{e_i(\pi)}\Big| \Big] | A, Z\Big] \le  c' \frac{\Gamma \sqrt{\chi(A_n^2)}}{\gamma \delta_n} \sqrt{\frac{ \mathcal{N}_n \log(\mathcal{N}_n) \mathrm{VC}(\Pi)}{n_e}}. 
\end{aligned} 
$$ 
By Lemma \ref{lem:boundnumber}, $\chi(A^2) \le 2 \mathcal{N}_n^2$, for a universal constant $c' < \infty$.  
\vspace{-4mm} 

\paragraph{Rademacher complexity bounds for $(iv)$} The bound for $(iv)$ follows verbatim as the bound for $(ii)$, where, here, instead of conditioning on $\hat{e}_{i \in F_k^j}$ as in Equation \eqref{eqn:summand_m}, I condition on $\hat{m}_{i \in F_k^j}$. This is omitted for space constraints. The proof completes. 

\vspace{-2mm} 

\subsubsection{Proof of Theorem \ref{thm:trimming}} 
Define 
$
W_{A, Z}^{tr}(\pi) = \frac{1}{n} \sum_{i=1}^n m\Big(\pi(X_i), T_i(\pi), Z_i, |N_i|\Big)1\Big\{|N_i| \le \log_\gamma(\kappa_n)\Big\}
$
the trimmed version of welfare. Following Corollary \ref{cor:1}, 
\begin{equation} \label{eqn:trimming_helper1} 
\small 
\begin{aligned} 
 & \mathbb{E}\Big[\sup_{\pi \in \Pi_n} W_{A,Z}(\pi) - W_{A,Z}(\hat{\pi}_{\kappa_n}^{tr}) \Big| A, Z\Big]  \le 2 \mathbb{E}\Big[\sup_{\pi \in \Pi}\Big|W_{A,Z}(\pi) - W_n^{tr}(\pi)\Big| | A, Z\Big]  \\ 
 &\le  2 \mathbb{E}\Big[\sup_{\pi \in \Pi}\Big|W_{A,Z}^{tr}(\pi) - W_n^{tr}(\pi)\Big| | A, Z\Big]  + 2 \sup_{\pi \in \Pi}\Big|W_{A,Z}^{tr}(\pi) - W_{A,Z}(\pi)\Big|.
\end{aligned} 
\end{equation} 
The bounds for the first component in the right-hand side of Equation \eqref{eqn:trimming_helper1} follows verbatim the proof of Theorem \ref{thm:thm2}, since $\mathbb{E}[W_n^{tr}(\pi)|A, Z] = W_{A,Z}^{tr}(\pi)$, with the difference that the overlap constant is  $\gamma^{\log_\gamma(\kappa_n) + 1}$ under Assumption \ref{ass:quasi} (iii). For the second component,
\begin{equation} \label{eqn:utb}
\small 
\begin{aligned} 
\Big|W_{A,Z}^{tr}(\pi) - W_{A,Z}(\pi)\Big| & \le \frac{1}{n} \sum_{i=1}^n m\Big(\pi(X_i), T_i(\pi), Z_i, |N_i|\Big)\Big(1 - 1\Big\{|N_i| \le \log_\gamma(\kappa_n)\Big\}\Big). 
\end{aligned} 
\end{equation}
Here, $\eqref{eqn:utb} =  \mathcal{O}\Big(\frac{1}{n} \sum_{i=1}^n 1\Big\{|N_i| > \log_\gamma(\kappa_n)\Big\}\Big)$, by \ref{ass:ignorability} (C) and Holder's inequality.  

\vspace{-2mm} 

\subsubsection{Proof of Theorem \ref{thm:expected}} \label{proof:expected} 

Define $W(\pi) = \mathbb{E}_{A', Z'}[W_{A', Z'}(\pi)]$ and $W(\hat{\pi}_{m^c, e}) = \mathbb{E}_{A', Z'}[W_{A', Z'}(\hat{\pi}_{m^c, e}) | \hat{\pi}_{m^c, e}]$, where $\hat{\pi}_{m^c, e} \perp (A', Z')$ by assumption. We can write, following similar steps as in Equation \eqref{eqn:kk1} with $W(\pi)$ in lieu of $W_{A, Z}(\pi)$, 
$\sup_{\pi \in \Pi} W(\pi) - W(\hat{\pi}_{m^c, e}) \le 2 \sup_{\pi \in \Pi} |W(\pi) - W_n(\pi, m^c, e)|$. Therefore, by taking expectations, 
\begin{equation} \label{eqn:hhhhh}
\small 
\begin{aligned} 
& \sup_{\pi \in \Pi} W(\pi) - \mathbb{E}[W(\hat{\pi}_{m^c, e})] = \mathbb{E}\Big[\sup_{\pi \in \Pi} W(\pi) - W(\hat{\pi}_{m^c, e})\Big] \le 2 \mathbb{E}\Big[\sup_{\pi \in \Pi} |W(\pi) - W_n(\pi, m^c, e)|\Big] \\ &= 2 \mathbb{E}\left[\sup_{\pi \in \Pi} \Big|W_n(\pi, m^c, e) - W_{A, Z}(\pi) + W_{A, Z}(\pi) - \mathbb{E}[W_{A', Z'}(\pi)]\Big| \right] \\ 
&= 2 \underbrace{\mathbb{E}\left[\sup_{\pi \in \Pi} \Big|W_n(\pi, m^c, e) - W_{A, Z}(\pi)\Big|\right]}_{(A)} + 2 \underbrace{\mathbb{E}\left[\sup_{\pi \in \Pi} \Big|W_{A, Z}(\pi) - \mathbb{E}[W_{A', Z'}(\pi)]\Big| \right]}_{(B)}.  
\end{aligned} 
\end{equation} 
 $(A)$ can be bounded using directly Theorem \ref{thm:thm2} and the law of iterated expectations.

\subsubsection{Proof of Proposition \ref{prop:3}} \label{app:more}

  To show that Proposition \ref{prop:3} I need to show that (i) the VC dimension of $\tilde{\Pi}_n$ is at most $\mathrm{VC}(\Pi)$ up-to a constant factor; (ii) overlap holds for any class of policy $\pi \in \tilde{\Pi}_n$, namely $e_i(\pi) \in (\gamma \delta_n, 1 - \gamma \delta_n)$. The rest of the proof then follows verbatim from Theorem \ref{thm:thmmain}. 
  
  First, for (i), note that by Theorem 13.1 in \cite{devroye2013probabilistic}, the VC dimension of the classifier $\tilde{\pi}(x,d) = \pi(x)(1 - d)$ equals the VC dimension of $\pi(x)$, namely $\mathrm{VC}(\Pi)$. By Lemma 29.4 in \cite{devroye2013probabilistic} it follows that the VC dimension of $\tilde{\Pi}_n$ equals VC$(\Pi)$. 
  
Second, for (ii),  for $\tilde{\pi}(x,d) = \pi(x)(1 - d) + d$
  $$
  \small 
  \begin{aligned} 
  P\Big(D_i = \tilde{\pi}(X_i, D_i) | Z_i, R_i = 1\Big) = \begin{cases} 
   P(D_i = 1 | Z_i, R_i = 1) & \text{ if } \pi(X_i) = 1\\ 
   1 & \text{ otherwise}. 
  \end{cases} 
    \end{aligned} 
  $$
It follows that $P\Big(D_i = \tilde{\pi}(X_i, D_i) | Z_i, R_i = 1\Big) \ge \min\{P(D_i = 1| Z_i, R_i = 1), P(D_i = 0 | Z_i, R_i = 1)\} \in (\gamma, 1- \gamma)$. Similarly, I can show that $P\Big(D_i = \tilde{\pi}(X_i, D_i) | Z_i, R_i = 0, R_i^f = 1\Big) \in (\gamma, 1-\gamma)$ and $P(T_i = t|Z_i, R_i = 1, R_{k \in N_i}, Z_{k \in N_i}, |N_i|) \ge \delta_n$ almost surely for any $t \in \mathcal{T}_n$, under Assumption \ref{ass:quasi} (ii). Intuitively, because I \textit{always} treat those units \textit{also treated} in the experiment, overlap for $\tilde{\pi} \in \tilde{\Pi}_n$ is guaranteed, under overlap in the experiment. It follows that the propensity score $e_i(\tilde{\pi}) = e(\tilde{\pi}(X_i, D_i), T_i(\tilde{\pi}), Z_i, Z_{k \in N_i}, R_{k \in N_i}, |N_i|)$, $\tilde{\pi} \in \tilde{\Pi}_n$ satisfies the overlap conditions imposed in Assumption \ref{ass:quasi}. Finally, it is easy to show that Lemma \ref{prop:welfare} directly holds also for any $\tilde{\pi} \in \tilde{\Pi}_n$, following verbatim the proof of Lemma \ref{prop:welfare}, reweighting for $e_i(\tilde{\pi})$. The rest of the proof follows verbatim the one of Theorem \ref{thm:thmmain} once we define the policy as $D_i + (1 -D_i) \pi(X_i)$, and the outcomes evaluated at the new policy are $r\Big(D_i + (1 -D_i) \pi(X_i), T_i(\pi), Z_i, |N_i|, \varepsilon_i\Big)$ with $T_i(\pi) = g_n\Big(\sum_{k \in N_i} D_k + (1 - D_k) \pi(X_k), Z_i, |N_i|\Big)$.

\subsection{Lemmas} \label{sec:lem}

\begin{lem} \label{lem:degreebound}

The following holds: $\chi(A_n) \le \chi(A_n^M) \le M \mathcal{N}_{n}^{M}$ for all $n \ge 1$. 
\end{lem} 
\begin{proof}[Proof of Lemma \ref{lem:degreebound}]
The first inequality follows by Definition \ref{defn:matrix}. The second inequality follows by Brook's Theorem \citep{brooks1941colouring}, since the maximum degree under $A_n^M$ is bounded by $\mathcal{N}_{n} + \mathcal{N}_{n} \times \mathcal{N}_{n} + \cdots + \prod_{s=1}^{M} \mathcal{N}_{n} \le M \mathcal{N}_{n}^{M}$.
\end{proof}

\begin{lem} \label{lem:lipweight} For $i \in \{1, \cdots, n\}$ consider functions $f_i: \mathcal{T}_n \mapsto [-U_n, U_n]$ for some $U_n >0$, and $\mathcal{T}_n \subseteq \mathbb{Z}$. Then for any $i \in\{1, \cdots, n\}, n \ge 1$, $f_i(t)$ is $2U_n$-Lipschitz in $t$. 
\end{lem} 
\begin{proof}[Proof of Lemma \ref{lem:lipweight}] For any $t, t' \in \mathbb{Z}$, 
$
\small
\Big|f_i(t) - f_i(t')\Big| \le 2 U_n 
$
for $t \neq t'$, by the triangular inequality. Since $\mathcal{T}_n \subseteq \mathbb{Z}$ is discrete, $\Big|f_i(t) - f_i(t')\Big| \le 2 U_n|t - t'|$. 
\end{proof}


\begin{lem} \label{lem:kita} For any $i \in \{1, \cdots, n\}$, let $X_i \in \mathcal{X}$ be an arbitrary random variable and $\mathcal{F}$ a class of uniformly bounded functions with envelope $\bar{F}$. 
Let $\Omega_i | X_1, \cdots, X_n$ be random variables independently but not necessarily identically distributed, where $\Omega_i \ge 0$ is a scalar. Assume that for some $u > 0$, 
$
\mathbb{E}[\Omega_i^{2 + u} | Z] < B, \quad \forall i \in \{1, \cdots, n\}.
$ In addition, assume that for any fixed points $x_1^n \in \mathcal{X}^n$, for some $V_n \ge 0$, for all $n \ge 1$,  
$\int_0^{2\bar{F}} \sqrt{\log\Big(\mathcal{M}_1\Big(\eta, \mathcal{F}(x_1^n)\Big)\Big)} d\eta < \sqrt{V_n}. 
$ Let $\sigma_i$ be $i.i.d$ Rademacher random variables independent of $(\Omega_i)_{i=1}^n,(X_i)_{i=1}^n$. 
Then for a constant $0 < C_{\bar{F}}< \infty$ that only depend on $\bar{F}$ and $u$, for all $n \ge 1$
$$
\small
\int_0^{\infty} \mathbb{E}\Big[\sup_{f \in \mathcal{F}}\Big|\frac{1}{n} \sum_{i=1}^n  \sigma_i f(X_i)1\{\Omega_i > \omega\}\Big| | X_1, \cdots, X_n \Big]d\omega \le  C_{\bar{F}}  \sqrt{\frac{B V_n}{n}}. 
$$
\end{lem}

\begin{proof}[Proof of Lemma \ref{lem:kita}] The proof follows verbatim the proof of Lemma A.5 in \cite{kitagawa2017equality}, with two small differences that do not affect the argument of the proof: I must control the Rademacher complexity using the Dudley's entropy integral bound (instead of the VC dimension), and $\Omega_i$ are independent but not necessarily identically distributed random variables. Given that the argument follows verbatim the one of Lemma A.5 of \cite{kitagawa2017equality}, the proof is omitted for space constraints.\footnote{The reader may refer to a technical note that collects lemmas from past literature available at \url{dviviano.github.io/projects/note_preliminary_lemmas.pdf} for details or Appendix E below.} 
\end{proof} 

\begin{lem} \label{lem:boundnumber} Take any $k \ge 2$. Let $\mathcal{F}_1, \cdots, \mathcal{F}_{k}
$ be classes of bounded functions with VC dimension $v$ and envelope $\bar{F} < \infty$. Let 
$$
\small
\begin{aligned} 
\mathcal{J}_n = \Big\{f_1(f_2 + ... + f_{k}), \quad f_j \in \mathcal{F}_j, \quad j = 1, \cdots, k\Big\}, \quad \mathcal{J}_n(x_1^n) = \Big\{h(x_1), \cdots, h(x_n); h \in \mathcal{J}_n\Big\}.
\end{aligned} 
$$
 For arbitrary fixed points $x_1^n \in \mathcal{X}^n$, for \textit{any} $n \ge 1, k \ge 2,  v\ge 1$, 
$
\int_0^{2\bar{F}} \sqrt{\log\Big(\mathcal{M}_1\Big(\eta, \mathcal{J}(x_1^n)\Big)\Big)} d\eta < c_{\bar{F}} \sqrt{k \log(k) v}
$
for a constant $c_{\bar{F}} < \infty$ that only depends on $\bar{F}$. 
\end{lem} 
\begin{proof}[Proof of Lemma \ref{lem:boundnumber}] Without loss of generality let $\bar{F} \ge 1$ (since if less than one the envelope is also uniformly bounded by one). Let 
$
\mathcal{F}_{-1,n}(x_1^n) = \{f_2(x_1^n) + ... + f_{k}(x_1^n), f_j \in \mathcal{F}_j, j = 2, ...,k_n \}.
$
 By \cite{devroye2013probabilistic}, Theorem 29.6, 
$
\small 
\mathcal{M}_1\Big(\eta, \mathcal{F}_{-1,n}(x_1^n)\Big) \le \prod_{j=2}^{k} \mathcal{M}_1\Big(\eta/(k-1), \mathcal{F}_j(x_1^n)\Big).
$
By Theorem 29.7 in \cite{devroye2013probabilistic}, 
\begin{equation} \label{eqn:helper3a} 
\small
\begin{aligned} 
\mathcal{M}_1\Big(\eta, \mathcal{J}_n(x_1^n)\Big) \le \prod_{j=2}^{k} \mathcal{M}_1\Big(\frac{\eta}{2 (k-1) \bar{F}}, \mathcal{F}_j(x_1^n)\Big) \mathcal{M}_1\Big(\frac{\eta}{2 \bar{F}}, \mathcal{F}_1(x_1^n)\Big) .
\end{aligned} 
\end{equation} 
By standard properties of covering numbers, for a generic set $\mathcal{H}$, $\mathcal{N}_1(\eta, \mathcal{H}) \le \mathcal{N}_2(\eta, \mathcal{H})$. It follows
$
 \eqref{eqn:helper3a} \le  \prod_{j=2}^{k} \mathcal{M}_2\Big(\frac{\eta}{2(k-1) \bar{F}}, \mathcal{F}_j(x_1^n)\Big) \mathcal{M}_2\Big(\frac{\eta}{2 \bar{F}}, \mathcal{F}_1(x_1^n)\Big). 
$
I now apply a uniform entropy bound for the covering number. By Theorem 2.6.7 of \cite{van1996weak}, we have that for a universal constant $C < \infty$ (that without loss of generality we can assume $C \ge 1$), 
$
\mathcal{M}_2\Big(\frac{\eta}{2(k-1) \bar{F}}, \mathcal{F}_j(x_1^n)\Big) \le C (v+1) (16 e)^{(v+1)} \Big(\frac{2\bar{F}^2(k-1)}{\eta}\Big)^{2 v}
$
which implies that 
$$
\small
\begin{aligned} 
\log\Big(\mathcal{M}_1\Big(\eta, \mathcal{J}_n(x_1^n)\Big)\Big) &\le\sum_{j = 1}^{k_n-1} \log\Big(\mathcal{M}_2\Big(\frac{\eta}{2\bar{F}(k-1)}, \mathcal{F}_j(x_1^n)\Big)\Big) + \log\Big(\mathcal{M}_2\Big(\frac{\eta}{2\bar{F}}, \mathcal{F}_1(x_1^n)\Big)\Big) \\ &\le k\log\Big(C (v +1)(16 e)^{v+1}\Big) + k2v \log(2C \bar{F}^2(k-1)/\eta).
\end{aligned}  
$$
Since $\int_0^{2 \bar{F}} \sqrt{k\log\Big(C (v +1)(16 e)^{v+1}\Big) + k_n2v \log(2C \bar{F}^2(k-1)/\eta)} d \eta \le c_{\bar{F}} \sqrt{k \log(k) v}$ for a constant $c_{\bar{F}} < \infty$, the proof completes.
\end{proof}

We discuss the \cite{ledoux2011probability}'s inequality for the case of interest here.  

\begin{lem} \label{lem:ledoux} For all $i \in \{1, \cdots, n\}$, let $\phi_i: \mathbb{R} \mapsto \mathbb{R}$ be such that $|\phi_i(a) - \phi_i(b)| \le L|a - b|$ for all $a, b \in \mathbb{R}$, with $\phi_i(0) = 0$, and arbitrary $L  > 0$. Then, for any $n \ge 1, L > 0$, any $\mathcal{U}_n \subseteq \mathbb{R}^n, \mathcal{K}_n \subseteq \{0,1\}^n$, with $u = (u_1, \cdots, u_n) \in \mathcal{U}_n$, $\alpha = (\alpha_1, \cdots, \alpha_n) \in \mathcal{K}_n$, 
$$
\small
\begin{aligned} 
\frac{1}{2} \mathbb{E}_{\sigma} \Big[ \sup_{u \in \mathcal{U}_n, \alpha \in \mathcal{K}_n} \Big| \frac{1}{n} \sum_{i=1}^n \sigma_i \phi_i(u_i) \alpha_i \Big|\Big] \le L \mathbb{E}_{\sigma}\Big[ \sup_{u \in \mathcal{U}_n, \alpha \in \mathcal{K}_n} \Big| \frac{1}{n} \sum_{i=1}^n \alpha_i \sigma_i u_i\Big|\Big].
\end{aligned} 
$$
\end{lem}  

\begin{proof}[Proof of Lemma \ref{lem:ledoux}]
	The proof follows closely the one of Theorem 4.12 in \cite{ledoux2011probability} while dealing with the additional $\alpha$ vector. We provide here the main argument and refer to \cite{ledoux2011probability} for additional details. First, note that if $\mathcal{U}_n$ is unbounded, there will be settings such that the right hand side is infinity and the result trivially holds. Therefore, let $\mathcal{U}_n$ be bounded. We aim to show that 
\begin{equation} \label{eqn:toprovelem}
\small
\begin{aligned} 
	\mathbb{E}\Big[\sup_{u \in \mathcal{U}_2, \alpha \in \mathcal{K}_2} \alpha_1 u_1 + \sigma_2 \phi(u_2) \alpha_2\Big] \le \mathbb{E}\Big[\sup_{u \in \mathcal{U}_2, \alpha \in \mathcal{K}_2} \alpha_1 u_1 + L\sigma_2 u_2 \alpha_2\Big].
	\end{aligned} 
	\end{equation}  
	If Equation \eqref{eqn:toprovelem}, it follows that 
	$$
\small
\begin{aligned} 
	\mathbb{E}\Big[\sup_{u \in \mathcal{U}_2, \alpha \in \mathcal{K}_2} \alpha_1 \phi_1(u_1) \sigma_1 + \sigma_2 \phi(u_2) \alpha_2 | \sigma_1\Big] \le \mathbb{E}\Big[\sup_{u \in \mathcal{U}_2, \alpha \in \mathcal{K}_2} \alpha_1 \phi_1(u_1) \sigma_1 + L\sigma_2 u_2 \alpha_2 \Big| \sigma_1\Big].
	\end{aligned} 
	$$
	Because $\sigma_1 \phi(u_1)$ simply transforms $\mathcal{U}_2$,
	and we can iteretively apply this result. 
	
I first prove Equation \eqref{eqn:toprovelem}. Define for $a,b \in \{0,1\}^2$,  
	$
	I(u,s, a, b) : = \frac{1}{2}\Big ( u_1 a_1 + a_2 \phi(u_2)\Big) +  \frac{1}{2} \Big( s_1 b_1 - b_2 \phi(s_2)\Big). 
	$
	I want to show that the right hand side in Equation \eqref{eqn:toprovelem} is larger than $I(u,s,a,b)$ for all $u,s \in \mathcal{U}_2$ and $a,b \in \{0,1\}^2$. Since I am taking the supremum of $I(u,s,a,b)$ over $u,s,a,b$, I can assume without loss of generality \citep[as in][]{ledoux2011probability} 
	\begin{equation} \label{eqn:helper} 
	\small 
	\begin{aligned} 
	u_1a_1 + a_2 \phi(u_2) \ge s_1 b_1 + b_2 \phi(s_2), \quad s_1 b_1 - b_2 \phi(s_2) \ge u_1 a_1 - a_2 \phi(u_2).  
	\end{aligned} 
	\end{equation} 
	I can now define four quantities of interest
	$$
	\small
	\begin{aligned} 
	q_1 = b_1 s_1 - b_2 \phi(s_2), \quad q_2 = b_1 s_1 - L s_2 b_2, \quad q_1' = a_1 u_1 + L a_2 u_2, \quad q_2' = a_1 u_1 + a_2 \phi(u_2).
	\end{aligned} 
	$$
	 I consider four different cases, similarly to \cite{ledoux2011probability} and argue that for any value of $(a_1, a_2, b_1, b_2) \in \{0,1\}^4$, $2 I(u,s,a,b)  = q_1 + q_2' \le q_1' + q_2$. 
	\\ \textit{Case 1}
	Start from the case $a_2u_2, s_2b_2 \ge 0$. We know that $\phi(0) = 0$, so that $|b_2 \phi(s_2)| \le L b_2 s_2$. 
	Now assume that $a_2 u_2 \ge b_2 s_2$. In this case 
$
	q_1 - q_2 =  L b_2 s_2  - b_2 \phi(s_2) \le L a_2 u_2 - a_2 \phi(u_2) = q_1' - q_2'   
	$
	since $|a_2\phi(u_2) - b_2 \phi(s_2)| \le L|a_2 u_2 - b_2 s_2| = L(a_2u_2 - b_2 s_2)$. To see why this last claim holds, note that for $a_2, b_2 = 1$, then the results hold by the condition $a_2 u_2 \ge b_2 s_2$ and Lipschitz continuity.  If instead $a_2 = 1, b_2 = 0$, the claim trivially holds. While the case $a_2 = 0, b_2 = 1$, then it must be that $s_2 = 0$ since we assumed that $a_2 u_2 \ge 0, b_2 s_2 \ge 0$ and $a_2 u_2 \ge b_2 s_2$. Thus $q_1 - q_2 \le q_1' - q_2'$.  If instead $b_2s_2 \ge a_2 u_2$, then use $-\phi$ instead of $\phi$ and switch the roles of $s,u$ giving a similar proof. \\ \textit{Case 2} Let $a_2 u_2 \le 0, b_2 s_2 \le 0$. The proof is as Case 1, switching the signs where necessary. \\ \textit{Case 3} Let $a_2 u_2 \ge 0, b_2 s_2 \le 0$. Then $a_2 \phi(u_2) \le L a_2 u_2$, since $a_2 \in \{0,1\}$ and by Lipschitz properties of $\phi$, $- b_2 \phi(s_2) \le - b_2 L s_2$ so that 
	$
	\small
	a_2 \phi(u_2) - b_2 \phi(s_2) \le a_2 L u_2 - b_2 L s_2. 
	$
	\\ \textit{Case 4} Let $a_2 u_2 \le 0, b_2 s_2 \ge 0$. Then the claim follows symmetrically to Case 3. \\ The conclusion of the proof follows verbatim the one in \cite{ledoux2011probability}. 
\end{proof}

\begin{lem} \label{lem:finallemma} 
Let $\Pi$, $\Pi'$ be two function classes, each with VC dimension $v$, and $\pi: \mathcal{X} \mapsto \{0,1\}$ for any $\pi \in \Pi, \Pi'$.
For $i \in \{1, \cdots, n\}$, take arbitrary $(X_{k \in N_i}, X_i), X_i \in \mathcal{X}, \Omega_i \in \mathbb{R}, R_i \in \{0,1\}$, adjacency matrix $A$, and functions $f_i: \mathbb{Z} \mapsto [-U_n, U_n]$, for a positive constant $U_n > 0$. Assume that $\mathbb{E}[|\Omega_i|^{3} | (R_i)_{i=1}^n, (X_i)_{i=1}^n, A] < B$, for some $B < \infty$, and $(\Omega_i)_{i=1}^n | (R_i)_{i=1}^n, (X_i)_{i=1}^n, A$ are independent but not necessarily identically distributed. Let $\sigma_1, \cdots, \sigma_n$ be $i.i.d.$ Rademacher random variables, independent of $\Big[\Big(X_i, R_i, \Omega_i\Big)_{i=1}^n, A\Big]$. Then for a universal constant $c_0 < \infty$, for any $n \ge 1$, $v = \mathrm{VC}(\Pi) = \mathrm{VC}(\Pi')$
\begin{equation} \label{eqn:rhs} 
\small 
\begin{aligned} 
\mathbb{E}_{\Omega, \sigma}\Big[\sup_{\pi_1 \in \Pi, \pi_2 \in \Pi'} \Big | \sum_{i=1}^n R_i f_i \Big(\sum_{k \in N_i} \pi_2(X_k) \Big) \pi_1(X_i)\sigma_i \Omega_i  \Big | \Big] \le  c_0 U_n \sqrt{v B \mathcal{N}_n \log(\mathcal{N}_n)  \sum_{i=1}^n R_i   }.
\end{aligned} 
\end{equation} 
\end{lem}  

\begin{proof}[Proof of Lemma \ref{lem:finallemma}]
First, note that since $R_i \in \{0,1\}$, and we take the expectation conditional on $(R_i)_{i=1}^n$, we can interpret the sum in Equation \eqref{eqn:rhs} as a sum over elements $\sum_{i=1}^n R_i$ many elements.   Also, note that from Lemma \ref{lem:lipweight}, we have that $f_i(t)$ is $2 U_n$-Lipschitz in $t$.
\vspace{-3mm} 
\paragraph{First decomposition} 
First, we add and subtract the value of the function $f_i(0)$ at zero. The left hand side in Equation \eqref{eqn:rhs} equals
\begin{equation} \label{eqn:lasterms}
\small
\begin{aligned} 
&\mathbb{E}_{\Omega,\sigma}\Big [ \sup_{\pi_1 \in \Pi, \pi_2 \in \Pi'}   \Big | \sum_{i = 1}^n R_i \sigma_i  \Big(f_i \Big(\sum_{k \in N_i} \pi_2(X_k)\Big) - f_i (0)  + f_i(0) \Big)   \Omega_i   \pi_1(X_i) \Big |\Big]
\\ &\le \underbrace{ \mathbb{E}_{\Omega,\sigma}\Big [ \sup_{\pi_1 \in \Pi, \pi_2 \in \Pi'}   \Big | \sum_{i = 1}^n R_i \sigma_i  \Big(f_i \Big(\sum_{k \in N_i} \pi_2(X_k) \Big) - f(0) \Big)   \Omega_i   \pi_1(X_i) \Big |\Big]}_{(1)} +\underbrace{ \mathbb{E}_{\Omega,\sigma}\Big [ \sup_{\pi_1 \in \Pi}  \Big | \sum_{i = 1}^n R_i \sigma_i  f_i(0)   \Omega_i   \pi_1(X_i) \Big |\Big]}_{(2)}. 
\end{aligned} 
\end{equation} 
First, I bound $(1)$. I write 
\begin{equation} \label{eqn:hh}
\small
\begin{aligned}
(1) & = \mathbb{E}_{\Omega,\sigma}\Big [ \sup_{\pi_1 \in \Pi, \pi_2 \in \Pi'}\Big |  \sum_{i = 1}^n R_i \sigma_i  \Big(f_i \Big(\sum_{k N_i} \pi_2(X_k) \Big) - f_i(0) \Big) |\Omega_i| \text{sign}(\Omega_i) \pi_1(X_i) \Big |\Big ] \\
&=  \mathbb{E}_{\Omega,\tilde{\sigma}}\Big [ \sup_{\pi_1 \in \Pi, \pi_2 \in \Pi'} \Big |  \sum_{i = 1}^n R_i \tilde{\sigma}_i  \Big(f_i \Big(\sum_{k \in N_i} \pi_2(X_k)\Big) - f_i(0) \Big) |\Omega_i|  \pi_1(X_i) \Big |\Big] 
\end{aligned} 
\end{equation} 
where $\tilde{\sigma}_i = \text{sign}(\Omega_i) \sigma_i$ which are $i.i.d.$ Rademacher random variables independent of $(\Omega_i, X_i, R_i)_{i=1}^n, A$, since $P(\tilde{\sigma}_i = 1 | \Omega) = P(\sigma_i \text{sign}(\Omega_i) = 1 |\Omega) = 1/2$.  Using the fact that $|\Omega_i| \ge 0$, I have
\begin{equation} \label{eqn:jjhgg1}
\small
\begin{aligned}
\eqref{eqn:hh}  &=\mathbb{E}_{\Omega,\tilde{\sigma}}\Big [ \sup_{\pi_1 \in \Pi, \pi_2 \in \Pi'}\Big | \sum_{i = 1}^n R_i \tilde{\sigma}_i  \Big(f_i \Big(\sum_{k \in N_i} \pi_2(X_k)\Big) - f(0) \Big) \int_0^{\infty} 1\{|\Omega_i| > \omega\} d\omega  \pi_1(X_i) \Big |\Big]  \\ &\le \mathbb{E}_{\Omega,\tilde{\sigma}}\Big [ \sup_{\pi_1 \in \Pi, \pi_2 \in \Pi'} \int_0^{\infty}  \Big |\sum_{i = 1}^n R_i \tilde{\sigma}_i  \Big(f_i \Big(\sum_{k \in N_i} \pi_2(X_k)\Big) - f_i(0) \Big)  1\{|\Omega_i| > \omega\}   \pi_1(X_i) \Big |d\omega\Big] \\
&\le  \int_0^{\infty} \mathbb{E}_{\Omega,\tilde{\sigma}}\Big [ \sup_{\pi_1 \in \Pi, \pi_2 \in \Pi'}  \Big |  \sum_{i = 1}^n R_i \tilde{\sigma}_i   \Big(f_i \Big(\sum_{k \in N_i} \pi_2(X_k)\Big) - f_i(0) \Big)  1\{|\Omega_i| > \omega\}   \pi_1(X_i) \Big |\Big] d\omega . 
\end{aligned}
\end{equation} 
Next, I use the law of iterated expectation to first take the expectation over $\tilde{\sigma}$ (conditional on $\Omega$) and then take the expectation over $\Omega$. I also divide and multiplied by $U_n$. I obtain 
\begin{equation} \label{eqn:jjhgg}
\small 
\begin{aligned} 
\eqref{eqn:jjhgg1} \le U_n \int _0^{\infty} \mathbb{E}_{\Omega}\Big[\mathbb{E}_{\tilde{\sigma}}\Big [ \sup_{\pi_1 \in \Pi, \pi_2 \in \Pi'}   \Big |  \sum_{i = 1}^n R_i \tilde{\sigma}_i \frac{1}{U_n} \Big(f_i \Big(\sum_{k \in N_i} \pi_2(X_k)\Big) - f_i(0) \Big)   1\{|\Omega_i| > \omega\}   \pi_1(X_i) \Big |\Big]\Big]d\omega.
\end{aligned}
\end{equation} 
\paragraph{Lipschitz property} Let $\phi_i(t) = \frac{1}{U_n}(f_i(t) - f_i(0))$. Here, $\phi_i$ is Lipschitz in $t$, with Lipschitz constant equal to $1$. In addition, $\phi_i(0) = 0$. By Lemma \ref{lem:ledoux}\footnote{ 
Conditional on $X, A, \Omega$, I invoke Lemma \ref{lem:ledoux} with $(\pi_1(X_i)1\{|\Omega_i| > \omega\})_{i=1}^n$ in lieu of $(\alpha_1, \cdots, \alpha_n) \in \mathcal{K}_n \subseteq \{0,1\}^n$ in the statement of Lemma \ref{lem:ledoux}, since $\pi_1(X_i) 1\{|\Omega_i| > \omega\}$ is binary. Here $(\sum_{k \in N_i} \pi_2(X_k))_{i=1}^n$ is in lieu of $(u_1, \cdots, u_n)\in \mathcal{U}_n$ in Lemma \ref{lem:ledoux}.  The spaces $\mathcal{K}_n, \mathcal{U}_n$ in Lemma \ref{lem:ledoux}, here are those defined (given $\Omega, X, A$), by $\pi_1(X_i) 1\{|\Omega_i| > \omega\}, \pi_1 \in \Pi$ and $(\sum_{k \in N_i} \pi_2(X_k))_{i=1}^n, \pi_2 \in \Pi'$, respectively.}, 
\begin{equation} \label{eqn:jjj}
\small
\begin{aligned} 
&\mathbb{E}_{\tilde{\sigma}}\Big [ \sup_{\pi_1 \in \Pi, \pi_2 \in \Pi'}   \Big |  \sum_{i = 1}^n R_i \tilde{\sigma}_i \frac{1}{U_n} \Big(f_i \Big(\sum_{k \in N_i} \pi_2(X_k)\Big) - f_i(0) \Big)   1\{|\Omega_i| > \omega\}   \pi_1(X_i) \Big |\Big]
\\ &\le 
2 \mathbb{E}_{\tilde{\sigma}}\Big [ \sup_{\pi_1 \in \Pi, \pi_2 \in \Pi'}   \Big | \sum_{i = 1}^n R_i \tilde{\sigma}_i  \Big(\sum_{k \in N_i} \pi_2(X_k)\Big)   1\{|\Omega_i| > \omega\}   \pi_1(X_i) \Big |\Big].
\end{aligned} 
\end{equation}
I can therefore write 
$$
\small 
\begin{aligned} 
\eqref{eqn:jjhgg} \le 2 U_n \int_0^{\infty} \mathbb{E}_{\Omega,\tilde{\sigma}}\Big [ \sup_{\pi_1 \in \Pi, \pi_2 \in \Pi'}   \Big | \sum_{i = 1}^n R_i \tilde{\sigma}_i  \Big(\sum_{k \in N_i} \pi_2(X_k)\Big)   1\{|\Omega_i| > \omega\}   \pi_1(X_i) \Big |\Big] d\omega.
\end{aligned} 
$$ 
\paragraph{Function reparametrization} I now consider a reparametrization of the function class. Define $\tilde{X}_{i} \in \mathcal{X}^{\mathcal{N}_n} = (X_i, X_{k \in N_i}, \emptyset, \cdots, \emptyset)$, where for the entries $h > |N_i| + 1$, $\tilde{X}_i^{(h)} = \emptyset$, denoting the $h^{th}$ entry of $\tilde{X}_i$. Without loss of generality, let $\pi(\emptyset) = 0$. Define $\pi_j \in \Pi_j$ a function class of the form $\pi_j(\tilde{X}_i) = \pi(\tilde{X}_i^{(j)}), \pi \in \Pi'$ for $j > 1$ and $\pi_1(\tilde{X}_i) = \pi(\tilde{X}_i^{(1)}), \pi \in \Pi$, i.e., equal to $\pi$ applied to the $j^{th}$ entry of the vector $\tilde{X}_i$. Since this is a trivial reparametrization, $\mathrm{VC}(\Pi_j) = \mathrm{VC}(\Pi)$ ($= \mathrm{VC}(\Pi')$ by assumption) for all $j \in \{1, \cdots, \mathcal{N}_n\}$.\footnote{See e.g., Theorem 29.4 in \cite{devroye2013probabilistic}.}  I can write 
$$
\small 
\begin{aligned} 
& U_n \int _0^{\infty} \mathbb{E}_{\Omega,\tilde{\sigma}}\Big [ \sup_{\pi_1 \in \Pi, \pi_2 \in \Pi'}   \Big |  \sum_{i = 1}^n R_i \tilde{\sigma}_i  \Big(\sum_{k \in N_i} \pi_2(X_k)\Big)   1\{|\Omega_i| > \omega\}   \pi_1(X_i) \Big |\Big]d\omega \\ &\le 
U_n \int _0^{\infty} \mathbb{E}_{\Omega,\tilde{\sigma}}\Big [ \sup_{\tilde{\pi}_{1} \in \Pi_{1}, \cdots, \tilde{\pi}_{\mathcal{N}_n} \in \Pi_{\mathcal{N}_n}}   \Big | \sum_{i = 1}^n R_i \tilde{\sigma}_i  \Big(\sum_{k = 1}^{\mathcal{N}_n - 1} \tilde{\pi}_{k+1}(\tilde{X}_i)\Big)   1\{|\Omega_i| > \omega\}   \tilde{\pi}_1(\tilde{X}_i) \Big |\Big]d\omega \\
&= U_n \int_0^{\infty} \mathbb{E}_{\Omega, \tilde{\sigma}} \Big[\sup_{\tilde{\pi} \in \tilde{\Pi}_n} \Big|  \sum_{i=1}^n R_i \tilde{\sigma}_i \tilde{\pi}(\tilde{X}_{i})1\{|\Omega_i| > \omega\}\Big|\Big] d\omega
\end{aligned} 
$$ 
where $\tilde{\Pi}_n = \Big\{\pi_1 \Big(\sum_{j = 2}^{\mathcal{N}_n - 1} \pi_{j+1} \Big), \pi_j \in \Pi_j, j = 1, \cdots, \mathcal{N}_n\Big\}$. I now apply Lemma \ref{lem:boundnumber}, using the fact that $\mathrm{VC}(\Pi_j) = \mathrm{VC}(\Pi) = \mathrm{VC}(\Pi')$, for any $j \in \{1, \cdots, \mathcal{N}_n\}$. By Lemma \ref{lem:boundnumber}, for any $n \ge 1$, the Dudley's integral of the function class $\tilde{\Pi}_n$  
is uniformly bounded by $C \sqrt{\mathcal{N}_n \log(\mathcal{N}_n)\mbox{VC}(\Pi)}$, for a finite universal constant $C$. By Lemma \ref{lem:kita}, since I am summing over $\sum_{i=1}^n R_i$ elements (conditional on $(R_1, \cdots, R_n)$), for a universal constant $\bar{C}' < \infty$
$$
\small
\begin{aligned} 
 &U_n \int_0^{\infty} \mathbb{E}_{\Omega, \tilde{\sigma}} \Big[\sup_{\tilde{\pi} \in \tilde{\Pi}_n} \Big| \sum_{i=1}^n R_i \tilde{\sigma}_i \tilde{\pi}(\tilde{X}_{i})1\{|\Omega_i| > \omega\}\Big|\Big] d\omega &\le \bar{C}' U_n \sqrt{B \mathcal{N}_n \mbox{VC}(\Pi)\log(\mathcal{N}_n) \sum_{i=1}^n R_i} . 
\end{aligned} 
$$

\paragraph{Term (2)} Next, I bound the term $(2)$ in Equation \eqref{eqn:lasterms}. Similar to $(1)$, 
$$
\small
\begin{aligned} 
\mathbb{E}_{\Omega,\sigma}\Big [ \sup_{\pi \in \Pi}   \Big | \sum_{i = 1}^n R_i \sigma_i  f_i(0)   \Omega_i   \pi(X_i) \Big |\Big] &\le 
 U_n \mathbb{E}_{\Omega,\tilde{\sigma}}\Big [ \sup_{\pi \in \Pi}   \Big | \sum_{i = 1}^n R_i \tilde{\sigma}_i  |\frac{f_i(0)}{U_n}   \Omega_i|   \pi(X_i) \Big |\Big] 
\\ &\le U_n \int_0^{\infty} \mathbb{E}_{\Omega, \tilde{\sigma}}\Big [ \sup_{\pi \in \Pi}   \Big |  \sum_{i = 1}^n R_i \tilde{\sigma}_i  1\{|f_i(0)   \Omega_i|/U_n > \omega\}   \pi(X_i) \Big |\Big] d\omega. 
\end{aligned} 
$$
Since $\Pi$ has finite VC dimension, by Theorem 2.6.7 of \cite{van1996weak} (the argument is the same as in Lemma \ref{lem:boundnumber}), 
$\int_0^{2} \sqrt{\mathcal{M}_1(\eta, \Pi(x_1^n))} d\eta < C \sqrt{\mbox{VC}(\Pi)}$ for a universal constant $C$, and for any $x_1^n \in \mathcal{X}^n$. 
Since $\mathbb{E}_{\Omega}[|f_i(0) \Omega_i/U_n|^3] \le B$ ($f_i(0)/U_n \in [-1, 1]$)
we can apply Lemma \ref{lem:kita}, with $|f_i(0)   \Omega_i|/U_n$ in lieu of $|\Omega_i|$ in Lemma \ref{lem:kita}, and obtain  
$$
\small
\begin{aligned} 
U_n \int_0^{\infty} \mathbb{E}_{\Omega, \sigma}\Big [ \sup_{\pi \in \Pi}   \Big | \sum_{i = 1}^n R_i \sigma_i  1\{|f_i(0)   \Omega_i|/U_n > \omega\}   \pi(X_i) \Big |\Big] d\omega \le C' U_n \sqrt{B \mathrm{VC}(\Pi) \sum_{i=1}^n R_i}
\end{aligned}
$$ 
for a universal constant $C' < \infty$. The proof completes. 
\end{proof} 

 The following lemma is a direct corollary of Lemma \ref{lem:finallemma}. 
 
 \begin{lem} \label{lem:final_lemmab} Let $\pi \in \Pi$, be a function class, with $\pi: \mathcal{X} \mapsto \{0,1\}$. 
For $i \in \{1, \cdots, n\}$, take arbitrary $(X_{k \in N_i}, X_i), X_i \in \mathcal{X}, \Omega_i \in \mathbb{R}, R_i \in \{0,1\}$, adjacency matrix $A$, and functions $g_i: \mathbb{Z} \times \{0,1\} \mapsto [-U_n, U_n]$, for a positive constant $U_n > 0$. Assume that $\mathbb{E}[|\Omega_i|^{3} | (R_i)_{i=1}^n, (X_i)_{i=1}^n, A] < B$, for some $B < \infty$, and $(\Omega_i)_{i=1}^n | (R_i)_{i=1}^n, (X_i)_{i=1}^n, A$ are independent but not necessarily identically distributed. Let $\sigma_1, \cdots, \sigma_n$ be $i.i.d.$ Rademacher random variables, independent of $\Big[\Big(X_i, R_i, \Omega_i\Big)_{i=1}^n, A\Big]$. Then for a universal constant $c_0 < \infty$, for any $n \ge 1$
\begin{equation} \label{eqn:rhs2} 
\small 
\mathbb{E}_{\Omega, \sigma}\Big[\sup_{\pi \in \Pi} \Big | \sum_{i=1}^n R_i g_i \Big(\sum_{k \in N_i} \pi(X_k), \pi(X_i) \Big)\sigma_i \Omega_i  \Big | \Big] \le  c_0 U_n \sqrt{ \mathrm{VC}(\Pi) B \mathcal{N}_n \log(\mathcal{N}_n)  \sum_{i=1}^n R_i   }.
\end{equation} 
 \end{lem} 
 
 \begin{proof}[Proof of Lemma \ref{lem:final_lemmab}] By Lemma \ref{lem:lipweight}, $g_i(t, 1), g_i(t, 0)$ are $2 U_n$-Lipschitz in $t$. It follows 
\begin{equation} \label{eqn:final_lemma} 
 \small 
 \begin{aligned} 
 & \mathbb{E}_{\Omega, \sigma}\Big[\sup_{\pi \in \Pi} \Big | \sum_{i=1}^n R_i g_i \Big(\sum_{k \in N_i} \pi(X_k), \pi(X_i) \Big)\sigma_i \Omega_i  \Big | \Big] \\ & \le \mathbb{E}_{\Omega, \sigma}\Big[\sup_{\pi \in \Pi} \Big | \sum_{i=1}^n R_i g_i \Big(\sum_{k \in N_i} \pi(X_k), 1 \Big) \pi(X_i) \sigma_i \Omega_i  \Big | \Big] +  \mathbb{E}_{\Omega, \sigma}\Big[\sup_{\pi \in \Pi} \Big | \sum_{i=1}^n R_i g_i \Big(\sum_{k \in N_i} \pi(X_k), 0 \Big) (1 - \pi(X_i)) \sigma_i \Omega_i  \Big | \Big] 
  \end{aligned} 
 \end{equation}  
 It follows 
 $$
 \small 
 \begin{aligned} 
 \eqref{eqn:final_lemma} & \le \mathbb{E}_{\Omega, \sigma}\Big[\sup_{\pi_1 \in \Pi, \pi_2 \in \Pi} \Big | \sum_{i=1}^n R_i g_i \Big(\sum_{k \in N_i} \pi_2(X_k), 1 \Big) \pi_1(X_i) \sigma_i \Omega_i  \Big | \Big]  \\ &\quad +  \mathbb{E}_{\Omega, \sigma}\Big[\sup_{\pi_1' \in \Pi, \pi_2' \in \Pi} \Big | \sum_{i=1}^n R_i g_i \Big(\sum_{k \in N_i} \pi_2'(X_k), 0 \Big) (1 - \pi_1'(X_i)) \sigma_i \Omega_i  \Big | \Big]. 
 \end{aligned}
 $$ 
 By Lemma 29.4 in \cite{devroye2013probabilistic}, the VC dimension of the function class $1 - \pi, \pi \in \Pi$ equals the VC$(\Pi)$. By Lemma \ref{lem:finallemma} each term in Equation \eqref{eqn:final_lemma} is bounded by $C U_n \sqrt{ \mathrm{VC}(\Pi) B \mathcal{N}_n \log(\mathcal{N}_n)  \sum_{i=1}^n R_i   }$, for a universal constant $C < \infty$. 
 \end{proof}

\begin{lem} \label{lem:K} Let $K^*$ be as in Algorithm \ref{alg:adaptive} (Equation \ref{eqn:coloring}). Then $K^* \le \chi(A^2)$ almost surely. 
\end{lem} 
\begin{proof}[Proof of Lemma \ref{lem:K}] To prove the claim it suffices to show that a partition such that the constraints in Equation \eqref{eqn:coloring} holds exists, and such a partition has size at most $\chi(A^2)$, for all possible realizations of $R = (R_1, \cdots, R_n)$. As a first step, observe that for fixed $K$, binary variables $G_{j,k} \in \{0,1\}, j \in \{1, \cdots, n\}, k \in \{1, \cdots, K\}$, with $\sum_{k=1}^K G_{j, k} = 1 \forall j \in \{1, \cdots, n\}$, 
$$
\small 
\begin{aligned} 
\sum_{k=1}^K \sum_{j = 1}^n 1\{j \in N_i \text{ or } N_i \cap N_j \neq \emptyset \} G_{j,k} G_{i,k} = 0 \text{ implies } \sum_{k=1}^K \sum_{j = 1}^n R_i R_j 1\{j \not \in \mathcal{I}_i\} G_{j,k} G_{i,k} = 0. 
\end{aligned} 
$$ 
Namely, $\sum_{k=1}^K \sum_{j = 1}^n 1\{j \in N_i \text{ or } N_i \cap N_j \neq \emptyset \} G_{j,k} G_{i,k} = 0$ is a stricter constraint than $\sum_{k=1}^K \sum_{j = 1}^n R_i R_j 1\{j \not \in \mathcal{I}_i\} G_{j,k} G_{i,k} = 0,$ in Equation \eqref{eqn:coloring}, for all $R_1, \cdots, R_n, R_i \in \{0,1\}$ (because $R_i$ is binary). I can therefore bound the solution to the optimization problem in Equation \eqref{eqn:coloring} as follows  
\begin{equation} \label{eqn:K_star}
\small 
\begin{aligned} 
K^* \le \mathrm{arg} \min_{K \in \mathbb{Z}} & \min_{G \in \{0,1\}^{n \times K}} K \\ &\text{ such that }  \sum_{k=1}^K \sum_{j = 1}^n 1\{j \in N_i \text{ or } N_i \cap N_j \neq \emptyset \} G_{j,k} G_{i,k} = 0 , \text{ and } \sum_{k =1}^K G_{i,k} = 1 \forall i.
\end{aligned} 
\end{equation} 
The right-hand side in Equation \eqref{eqn:K_star} equals $\chi(A^2)$ by definition of smallest proper cover. 
\end{proof} 

\vspace{-4mm} 

\subsubsection{Identification} \label{app:ident}

\vspace{-1mm}

\begin{proof}[Proof of Lemma \ref{prop:welfare}] 
Let $e \Big(\pi(X_i), T_i(\pi), Z_{k \in N_i},R_{k \in N_i},  Z_i, |N_i|\Big) = e_i(\pi), I_i(\pi) = 1 \{T_i(\pi) = T_i, \pi(X_i) = D_i \}$. 
Under Assumption \ref{ass:sutnva}, I can write 
\begin{equation} \label{eqn:ipweq}
\small
\begin{aligned} 
\mathbb{E}\Big[ R_i \frac{I_i(\pi)}{e_i(\pi)}Y_i \Big|  A, Z\Big]  =  \mathbb{E}\Big[ R_i \frac{I_i(\pi)}{e_i(\pi)}r\Big(\pi(X_i), T_i(\pi), Z_i, |N_i|, \varepsilon_i\Big) \Big|  A, Z \Big].
\end{aligned} 
\end{equation}
Under Assumption \ref{ass:quasi} (i,ii),
$$
\small
\begin{aligned} 
\eqref{eqn:ipweq} &=  \mathbb{E}\Big[\frac{R_i I_i(\pi)}{e_i(\pi)}| A, Z\Big] \times \mathbb{E}\Big[r\Big(\pi(X_i), T_i(\pi), Z_i, |N_i|,  \varepsilon_i\Big)\Big | A, Z\Big]. 
\end{aligned} 
$$
By Assumption \ref{ass:quasi} (i), 
$
\small 
\begin{aligned} 
\mathbb{E}\Big[\frac{R_i I_i(\pi)}{e_i(\pi)}| A, Z\Big] = \mathbb{E}\Big[R_i \mathbb{E}\Big[\frac{ I_i(\pi)}{e_i(\pi)}| A, Z, (R_i)_{j\neq i}, R_i = 1\Big]\Big]= \frac{n_e}{n}.
\end{aligned}
$
\end{proof} 

\begin{lem} \label{lem:doublerobust} Let Assumptions \ref{ass:sutnva}, \ref{ass:quasi} hold. Then 
$$ 
\small
\begin{aligned} 
 & \frac{1}{n_e} \sum_{i=1}^n \mathbb{E}\Big[R_i \frac{1 \{T_i(\pi) = T_i, d = D_i \}}{e^c \Big(\pi(X_i), T_i(\pi), Z_{k \in N_i}, R_{k \in N_i}, Z_i,  |N_i|\Big)}\Big(Y_i -m^c\Big(\pi(X_i), T_i(\pi), Z_i, |N_i|\Big)\Big) \Big| A, Z\Big] \\ & + \frac{1}{n_e} \sum_{i=1}^n \mathbb{E}\Big[R_i m^c\Big(\pi(X_i), T_i(\pi), Z_i, |N_i|\Big) \Big| A, Z\Big] = \frac{1}{n} \sum_{i=1}^n \mathbb{E}\Big[r\Big((\pi(X_i), T_i(\pi), Z_i, |N_i|, \varepsilon_i\Big) \Big| A, Z\Big]
 \end{aligned} 
$$
if either $e^c = e$ or (and) Assumption \ref{ass:ignorability} (A) holds with $m^c = m$.
\end{lem} 

\begin{proof}[Proof of Lemma \ref{lem:doublerobust}]
Define $e_i^c(\pi) = e^c \Big(\pi(X_i), T_i(\pi), Z_{k \in N_i}, R_{k \in N_i}, Z_i, |N_i|\Big), I_i(\pi) = 1 \{T_i(\pi) = T_i, \pi(X_i) = D_i \}, m_i^c = m^c(\pi(X_i), T_i(\pi), Z_i, |N_i|)$. 
Whenever $e^c = e$, the result directly follows from Lemma \ref{prop:welfare}. 
Let now $m^c = m$ and Assumption \ref{ass:ignorability} (A) hold. Then (since the indicators $R$ are independent of $\varepsilon$ by Assumption \ref{ass:ignorability})
$$ 
\small
\begin{aligned} 
 & \mathbb{E}\Big[\frac{R_i I_i(\pi)}{e_i^c(\pi)}\Big(Y_i - m_i^c(\pi)\Big)\Big | A, Z\Big]  = 
 \mathbb{E}\Big[R_i \frac{I_i(\pi)}{e_i^c(\pi)}\Big(r\Big(\pi(X_i), T_i(\pi), Z_i, |N_i|, \varepsilon_i\Big) - m_i(\pi)\Big)\Big | A, Z\Big] \\ &=  \mathbb{E}\Big[R_i \frac{I_i(\pi)}{e_i^c(\pi)} \Big | A, Z\Big] \times \mathbb{E}\Big[\Big(r\Big(\pi(X_i), T_i(\pi), Z_i, |N_i|, \varepsilon_i\Big) - m_i(\pi)\Big)\Big |A, Z \Big] = 0.
 \end{aligned} 
$$
By Assumption \ref{ass:quasi} (i), 
$
\frac{1}{n_e} \sum_{i=1}^n \mathbb{E}\Big[R_i m_i(\pi) \Big| A, Z\Big] = \frac{1}{n} \sum_{i=1}^n m(\pi(X_i), T_i(\pi), Z_i, |N_i|). 
$
\end{proof}

\vspace{-3mm}

\subsection{Proofs for ``Additional extensions"} \label{proof:2}

\vspace{-1mm}

\subsubsection{Proof of Proposition \ref{prop:rate_estimator}} \label{app:lasso}

Define $k(i)$ the partition $k \in \{1, \cdots, K^*\}$ associated with unit $i$ under Algorithm \ref{alg:adaptive} and $j(i)$ the fold $j$ within partition $k(i)$ associated with $i$ under Algorithm \ref{alg:adaptive}. 
Recall the definition of $\phi_{s}^m(i) = 1\{k(s) = k(i), j(s) \neq j(i)\}$ is Section \ref{sec:lasso}. 
Note that $\phi_s^m(i)$ are random variables since they depend on sampled indicators $R_1, \cdots, R_n$. By Lemma \ref{lem:K}, $K^* \le \chi(A^2)$. 

For each partition $k$, Algorithm \ref{alg:adaptive} creates $J$ folds with the same number of units. I can write 
$
 \sum_{s=1}^n R_s \phi_{s}^m(i) \ge \Big\lfloor \frac{J - 1}{J} \sum_{s=1}^n R_s 1\{k(s) = k(i)\} \Big\rfloor 
$
where I take the floor function for cases where $J$ is not a multiple of the number of sampled units in the partition $k(i)$. We have 
\begin{equation}  \label{eqn:helperfinalprop}
\small 
\begin{aligned} 
& \frac{1}{n} \sum_{i=1}^n \mathbb{E}\Big[\Big(1 + \sum_{s=1}^n R_s \phi_s^m(i)\Big)^{-2\zeta_m} | R_i = 1, A, Z\Big] \\ 
&\le \frac{1}{n} \sum_{i=1}^n \mathbb{E}\Big[\max\Big\{1, \Big(\frac{J - 1}{J} \sum_{s=1}^n R_s 1\{k(s) = k(i)\} \Big)^{-2\zeta_m} \Big\} | R_i = 1, A, Z\Big].
\end{aligned} 
\end{equation}  

\paragraph{Worst-case partition} Next, I replace the (random) partitions $k \in \{1, \cdots, K^*\}$ with \textit{worst-case} non-random partitions.
Denote $k^w(i) \in \{1, \cdots, \chi(A^2)\}$ the worst-case partition
\begin{equation} \label{eqn:find_feasible}
\small 
\begin{aligned} 
k^w(\cdot) \in \mathrm{arg} & \max_{\underline{k}(i) \in \{1, \cdots, \chi(A^2)\}, i \in \{1, \cdots, n\}} \frac{1}{n} \sum_{i=1}^n \mathbb{E}\Big[\max\Big\{1, \Big(\frac{J - 1}{J} \sum_{s=1}^n R_s 1\{\underline{k}(s) = \underline{k}(i)\}\Big)^{-2\zeta_m} \Big\} | R_i = 1, A, Z\Big]  \\
& \text{ such that } \underline{k}(i) \neq \underline{k}(j), \forall j \in N_i \text{ or } N_i \cap N_j \neq \emptyset, \quad \sum_{k = 1}^{\chi(A^2)}  1\{\underline{k}(i) = k\} = 1, \quad \forall i \in \{1, \cdots, n\}.
\end{aligned} 
\end{equation}
Here, $k^w(\cdot)$ always exists by definition of $\chi(A^2)$.\footnote{Existence is satisfied if a feasible solution to Equation \eqref{eqn:find_feasible} exists. One example is the smallest proper cover $\mathcal{C}_n(A^2)$ as in Definition \ref{defn:cover} for the adjacency matrix $A^2$. This satisfies the constraints in Equation \eqref{eqn:find_feasible} by definition. A proper cover always exists (e.g., if the network is fully connected, $\chi(A^2) = n$).} In addition, $k^w$ does not depend on the \textit{realized} $R$ by construction. I claim that 
\begin{equation} \label{eqn:to_bound1}
\small 
\begin{aligned} 
& \eqref{eqn:helperfinalprop} \le  
\frac{1}{n} \sum_{i=1}^n \underbrace{\mathbb{E}\Big[\max\Big\{1, \Big(\frac{J - 1}{J} \sum_{s=1}^n R_s 1\{k^w(s) = k^w(i)\}\Big)^{-2\zeta_m} \Big\} | R_i = 1, A, Z\Big]}_{(I)}
\end{aligned} 
\end{equation} 
Equation \eqref{eqn:to_bound1} holds for two reasons: (i)  $K^* \le \chi(A^2)$ by Lemma \ref{lem:K}; (ii) I can show that the constraint in Equation \eqref{eqn:find_feasible} is a stricter constraint than the constraint in Equation \eqref{eqn:coloring} for \textit{any} realization of $(R_1, \cdots, R_n)$ (see the proof of Lemma \ref{lem:K} for details). 

\paragraph{Upper bound on $(I)$} Take any $i \in \{1, \cdots, n\}$ such that $1\{k^w(s) = k^w(i)\} = 1$ for some $s \neq i$. It follows from \cite{cribari2000note} (equation at the bottom of Page 274)
\begin{equation} \label{eqn:helper11111}
\small 
\begin{aligned} 
& (I) \le (\frac{J - 1}{J})^{-2\zeta_m} \mathbb{E}\Big[\Big(1 + \sum_{s \neq i} R_s 1\{k^w(s) = k^w(i)\}\Big)^{-2\zeta_m} | R_i = 1, A, Z\Big] \quad (\because R_i = 1) \\ & \le \frac{(\frac{J - 1}{J})^{-2\zeta_m}}{\Big(\frac{n_e}{n} \sum_{s \neq i} 1\{k^w(s) = k^w(i)\}\Big)^{2\zeta_m}} + \mathcal{O}\Big(\frac{1}{(\sum_{s \neq i} 1\{k^w(s) = k^w(i)\})^{2\zeta_m + 1}}\Big) \\ &\quad  (\because n_e/n = \alpha \in (0,1), J = \mathcal{O}(1)) .  
\end{aligned} 
\end{equation} 
In the right-hand-side (first equation) we added one since $k^w(i) = k^w(s)$ for $s = i$. If instead there is no $s \neq i$, such that $1\{k^w(s) = k^w(i)\} = 1$, then trivially $(I) = \mathcal{O}(1)$. 

\paragraph{Sum over all partitions} Summing over all $\chi(A^2)$ partitions, we obtain 
$$
\small 
\begin{aligned} 
& \eqref{eqn:to_bound1} \le \sum_{k=1}^{\chi(A^2)} \frac{\sum_{i=1}^n 1\{k^w(i) = k\}}{n} \mathbb{E}\Big[\max\Big\{1, \Big(\frac{J - 1}{J} \sum_{s=1}^n R_s 1\{k^w(s) = k\}\Big)^{-2\zeta_m}\Big\}\Big] \le  \underbrace{\mathcal{O}(\chi(A^2)/n)}_{(A)} \\ & + \underbrace{\mathcal{O}\Big(\sum_{k=1}^{\chi(A^2)} \Big(\frac{\sum_{i=1}^n 1\{k^w(i) = k\}}{n}\Big)^{1 - 2\zeta_m} \Big(\frac{J}{(J - 1) n_e}\Big)^{2 \zeta_m}\Big) + \mathcal{O}\Big(\frac{1}{n} \sum_{k=1}^{\chi(A^2)} (1 + \sum_{s \neq i} 1\{k^w(i) = k\})^{- 2 \zeta_m} \Big)}_{(B)} 
\end{aligned} 
$$ 
where $(B)$ correspond to cases where partitions $k^w(i)$ contain at least two elements (and bounded as in Equation \eqref{eqn:helper11111})\footnote{For the first component in $(A)$ we sum over all $i \in \{1, \cdots, n\}$ instead of $n-1$ elements since the last term is absorbed in $\mathcal{O}(1)$.}, and $(A)$ corresponds to partitions with only one element, whose overall number is at most $\chi(A^2)$ (since there are at most $\chi(A^2)$ many partitions, and for such partitions $\frac{\sum_{i=1}^n 1\{k^w(i) = k\}}{n} = 1/n$). For $(B)$ we write 
$$ 
\small 
\begin{aligned} 
(B) &\le \mathcal{O}\Big(\chi(A^2) \Big(\frac{1}{\chi(A^2)} \sum_{k=1}^{\chi(A^2)} \frac{\sum_{i=1}^n 1\{k^w(i) = k\}}{n}\Big)^{1 - 2\zeta_m} \Big(\frac{J}{(J - 1) n_e}\Big)^{2 \zeta_m}\Big) \\ &+ \mathcal{O}\Big(\chi(A^2) \frac{1}{n}  (\frac{1}{\chi(A^2)} \sum_{k=1}^{\chi(A^2)} \sum_{i=1}^n 1\{k(i) = k\})^{1 - 2 \zeta_m} \Big) \quad (\because x^{-2\zeta_m} \le x^{1 - 2\zeta_m} \text{ for } x \ge 1, \text{ concave } x^{1 - 2 \zeta_m}). 
\end{aligned} 
$$ 
It follows that $(B) \le \chi(A^2) \Big(\frac{J}{(J - 1) n_e}\Big)^{2 \zeta_m} + \mathcal{O}(\chi(A^2) n^{-2\zeta_m}) \quad (\because  \sum_{k=1}^{\chi(A^2)} \sum_{i=1}^n 1\{k(i) = k\} = n)$. 
 From \ref{lem:degreebound}, $\chi(A^2) \le 2 \mathcal{N}_n^2$, which completes the proof for the conditional mean after simple rearrangement (since the bound for $(A)$ follows directly from Lemma \ref{lem:degreebound}). The argument follows verbatim for $\mathcal{B}_n(A,Z)$, taking into account $1/\delta_n^2$, and omitted for brevity.

\vspace{-2mm} 

\subsubsection{Proof of Proposition \ref{thm:selection} } \label{sec:ee}

Denote $\mathbb{E}_\pi[\cdot]$ the expectation conditional on $\Big\{D_i = \pi(X_i) \Big\}_{i=1}^n$, let $R = (R_i)_{i=1}^n$. We have
\begin{equation} \label{eqn:S_pi}
\small 
\begin{aligned} 
\mathbb{E}_\pi\Big[r\Big(S_i, \sum_{k \in N_i} S_k, Z_i, |N_i|, \varepsilon_i\Big) \Big| A, Z\Big] = \mathbb{E}\Big[r\Big(S_i(\pi), \sum_{k \in N_i} S_k(\pi), Z_i, |N_i|, \varepsilon_i \Big) \Big| A, Z, R\Big],  
\end{aligned} 
\end{equation} 
where $S_i(\pi) = h_\theta\Big(\pi(X_i), \sum_{k \in N_i} \pi(X_k), Z_i, |N_i|, \nu_i\Big)$. It follows that Equation \eqref{eqn:S_pi} equals
$$
\small 
\begin{aligned} 
\sum_{s \in \{0, \cdots, |N_i|\}}  \underbrace{\mathbb{E}\Big[r(d, s, Z_i, |N_i|, \varepsilon_i)\Big| S_i(\pi) = d, \sum_{k \in N_i} S_k(\pi) = s, Z, A \Big]}_{(i)}  \times 
\underbrace{P\Big(S_i(\pi) = d, \sum_{k \in N_i} S_k(\pi) = s \Big| A, Z, R\Big)}_{(ii)} . 
\end{aligned} 
$$ 
Since $(\varepsilon_j)_{j=1}^n \perp \Big(Z, A, (\varepsilon_{D_j}, \nu_j, R_j)_{j = 1}^n\Big)$, I can show
$
(i) =  \mathbb{E}\Big[r(d, s, Z_i, |N_i|, \varepsilon_i)\Big| S_i = d, \sum_{k \in N_i} S_k = s, Z, A, R \Big] 
$.
Consider now $(ii)$. Observe that by indepedence and exogeneity of $(\nu_j)_{j=1}^n$, 
$$
\small 
(ii) = P\Big(S_i(\pi) = d \Big| A, Z, R \Big) \times \sum_{u_1, \cdots, u_{l}: \sum_v u_v = s} \prod_{k = 1}^{|N_i|}  P\Big(S_{N_i^{(k)}}(\pi) = u_k \Big|  A, Z, R \Big).
$$ 
Using exogeneity of $\nu_i$, I have 
$$
\small 
P\Big(S_i(\pi) = d \Big| A, Z, R\Big) = 
P\Big(S_i = d \Big| Z_i, |N_i|, D_i = \pi(X_i), \sum_{k \in N_i} D_k = \sum_{k \in N_i} \pi(X_k), Z_{k \in N_i}, Z_i \Big). 
$$ 
Similar reasoning also applies to neighbors' selected treatments, omitted for brevity.

\vspace{-2mm}

\subsubsection{Proof of Proposition \ref{prop:different2}} \label{app:main_last}

First, we show that  
$
\mathbb{E}\Big[\tilde{W}_n(\pi, m^c, e)\Big| A, Z, A', Z'\Big] = W_{A', Z'}(\pi). 
$ 
Let $L_i = L(Z_i, Z_{k \in N_i}, |N_i|)$ and similarly $L_i' = L'(Z_i, Z_{k \in N_i}, |N_i|)$. Let $T_i', Z_i', |N_i|'$ be the neighbors' exposure, covariates and number of neighbors of $i$ in the target population.  Following Lemma \ref{lem:doublerobust} below, by exogeneity of $(R_1, \cdots, R_n)$ (Assumption \ref{ass:quasi} (i,ii))
$$
\small 
\begin{aligned} 
R_i \mathbb{E}\Big[\frac{I_i(\pi)}{e_i(\pi)}\Big(Y_i - m_i^c(\pi)\Big) + m_i^c(\pi)\Big| A, Z, R_1, \cdots, R_n\Big] & = R_i \mathbb{E}\Big[r\Big((\pi(X_i), T_i(\pi), Z_i, |N_i|, \varepsilon_i\Big) \Big| A, Z\Big] \\ &= R_i m\Big(\pi(X_i), T_i(\pi), Z_i, |N_i|\Big). 
\end{aligned} 
$$ 
Therefore, it follows that 
$$
\small 
\begin{aligned} 
\mathbb{E}\Big[\tilde{W}_n(\pi, m^c, e)\Big|A, Z\Big] = \frac{1}{n} \sum_{i=1}^n \frac{L_i'}{L_i} m\Big(\pi(X_i), T_i(\pi), Z_i, |N_i|\Big) = \frac{1}{n} \sum_{i=1}^n m\Big(\pi(X_i'), T_i'(\pi). Z_i', |N_i|'\Big), 
\end{aligned} 
$$ 
The last equality follows by construction of $L_i', L_i$. $\mathcal{S}_n(A', Z') \subseteq \mathcal{S}_n(A,Z)$ guarantees that there are no individuals in the target population outside the sample population's support. 

Because 
 $\mathbb{E}[\tilde{W}_n(\pi, m^c, e)| A, Z, A', Z'] = W_{A',Z'}(\pi)$, the same argument of the proof of Theorem \ref{thm:thm2} holds, with the difference that the Lipschitz constant in the proof of Theorem \ref{thm:thm2} multiplies by $\bar{L}_{A, Z, n}$.

\end{appendices}

 \end{document}